\documentclass{scrartcl}
\usepackage[numbers]{natbib}
\usepackage{fontenc}
\usepackage{shadethm}
\usepackage{amsthm}
\usepackage{float}
\usepackage{bm}
\usepackage{mathtools}
\usepackage{graphicx}
\usepackage{subfig}
\usepackage{stmaryrd}
\usepackage{mathrsfs} 
\usepackage{amsmath}
\usepackage{amssymb}
\usepackage{wasysym}
\usepackage{sectsty}
\usepackage[hyperfootnotes=false]{hyperref}
\usepackage[a4paper, left=2.7cm, right=2.7cm, top=2.5cm]{geometry}
\usepackage{chngcntr}
\counterwithin*{equation}{section}

\usepackage{cleveref}
\DeclareMathAlphabet{\mathpzc}{OT1}{pzc}{m}{it}

\sectionfont{\normalsize\centering}
\subsectionfont{\footnotesize\centering}

\theoremstyle{plain}
\newtheorem{thm}{Theorem}[section] 

\theoremstyle{definition}
\newtheorem{defn}[thm]{Definition} 
\newtheorem{lem}[thm]{Lemma}
\newtheorem{prop}[thm]{Proposition}
\newtheorem{rem}[thm]{Remark}
\newtheorem{cor}[thm]{Corollary}

\def\Xint#1{\mathchoice
	{\XXint\displaystyle\textstyle{#1}}%
	{\XXint\textstyle\scriptstyle{#1}}%
	{\XXint\scriptstyle\scriptscriptstyle{#1}}%
	{\XXint\scriptscriptstyle\scriptscriptstyle{#1}}%
	\!\int}
\def\XXint#1#2#3{{\setbox0=\hbox{$#1{#2#3}{\int}$ }
		\vcenter{\hbox{$#2#3$ }}\kern-.6\wd0}}

\def\dashint{\Xint-}
\usepackage[utf8]{inputenc}
\usepackage[german,english,russian]{babel}

\newcounter{MPequ}

\newcounter{AppA}

\newcounter{AppB}
\newenvironment{AppB}
{\stepcounter{AppB}%
	\addtocounter{equation}{0}%
	\equation}
{\endequation}
\newcounter{AppC}

\newcounter{AppD}

\newcounter{AppE}
\newenvironment{AppE}
{\stepcounter{AppE}%
	\addtocounter{equation}{0}%
	\equation}
{\endequation}
\begin{document}\selectlanguage{english}
\begin{center}
\normalsize \textbf{\textsf{Kernel and image of the Biot-Savart operator and their applications in stellarator designs}}
\end{center}
\begin{center}
	Wadim Gerner\footnote{\textit{E-mail address:} \href{mailto:wadim.gerner@edu.unige.it}{wadim.gerner@edu.unige.it}}
\end{center}
\begin{center}
{\footnotesize	MaLGa Center, Department of Mathematics, Department of Excellence 2023-2027, University of Genoa, Via Dodecaneso 35, 16146 Genova, Italy}
\end{center}
{\small \textbf{Abstract:} 
We consider the Biot-Savart operator acting on $W^{-\frac{1}{2},2}$ regular, div-free, surface currents $j$
\begin{gather}
	\nonumber
	\operatorname{BS}(j)(x)=\frac{1}{4\pi}\int_{\Sigma}j(y)\times \frac{x-y}{|x-y|^3}d\sigma(y)\text{, }x\in \Omega
\end{gather}
where $\Sigma$ is a connected surface to which $j$ is tangent and where $\Omega$ is the finite domain bounded by $\Sigma$.
\newline
We answer two questions regarding this operator.
\begin{enumerate}
	\item We provide an algorithm which converges (theoretically) exponentially fast to an element of the kernel of the Biot-Savart operator, as well as characterise the elements of the kernel of the Biot-Savart operator in terms of certain solutions to exterior boundary value problems. This allows one to explicitly exploit the non-uniqueness of the coil reconstruction process in the context of stellarator designs.
	\item We provide a simple, concise characterisation of the image of the Biot-Savart operator. This allows to define a 2-step current reconstruction procedure to obtain surface currents which approximate to arbitrary precision a prescribed target magnetic field within the plasma region of a stellarator device.
	\newline
	The first step does not require computing integrals involving singular integral kernels of the form $\frac{x-y}{|x-y|^3}$ but may have a potentially slow convergence rate, while the second step requires the computation of integrals involving singular integral kernels but in turn has (theoretically) an exponential convergence rate.
	\newline
	This approximation procedure always leads to approximating surface currents which are as poloidal as possible.
\end{enumerate}
\section{Introduction}
One promising approach with regards to replicating plasma fusion on earth is magnetic confinement fusion which aims to confine the plasma by means of magnetic fields. The two most prominent designs are the tokamak design and the stellarator design. While the tokamak creates the confining magnetic field by means of simple coil structures and a strong plasma current, the stellarator instead relies mainly on complex coil structures not requiring any strong plasma currents \cite{Xu16}. Both approaches have advantages and disadvantages, c.f. \cite{Xu16}.
\newline
Traditionally a two step optimisation procedure is used in the design of stellarators, even though one-step optimisation procedures are becoming more prominent recently, \cite{HHPH21},\cite{JGLRW23}.
\begin{enumerate}
	\item \underline{Step 1:} In the first step one looks for a plasma shape and supporting magnetic field which optimise confinement properties as well as may take into account engineering constraints, c.f. \cite{HW83}. Mathematically, the output of this procedure is a (bounded) region $P$ known as plasma region or plasma domain and a vector field $B_T$ within $P$ which corresponds to the magnetic field which needs to be produced by the coils.
	\item \underline{Step 2:} One looks for a coil arrangement, as well as a current distribution supported by the coil structures which approximates well the desired vector field $B_T$ within the plasma region, c.f. \cite{L17},\cite{PLBD18},\cite{PRS22},\cite{G24}. Again, physical constraints may be taken into account.
\end{enumerate}
There are different ways to model the coils such as the filament model or the coil windings surface (CWS) method, see \cite[Chapter 13.4]{IGPW24},\cite{M87}. While the CWS method is a less realistic model it is easier to handle from a computational and theoretical point of view, even though a certain recent approach to plasma fusion confinement \cite{PPPV24} makes the CWS model a good model for this specific approach.
\newline
The goal of the present paper is to analyse in more detail the coil reconstruction problem from the point of view of the CWS model.
\newline
To be more precise about our setting let $\Sigma\subset\mathbb{R}^3$ be a closed, connected $C^{1,1}$-surface. A current distribution $j$ on $\Sigma$ is a square integrable vector field $j$ on $\Sigma$, which is tangent to $\Sigma$ at a.e. point and which is divergence-free in the sense that $\int_{\Sigma}j(x)\cdot \nabla \psi(x)d\sigma(x)=0$ for all $\psi\in C^1_c(\mathbb{R}^3)$ and where $d\sigma$ denotes the standard induced surface measure on $\Sigma$. The divergence-free condition is necessary to ensure that the considered currents satisfy Maxwell's equations.
Now, given any such current distribution $j$ it will induce a magnetic field in $3$-space according to the Biot-Savart law
\begin{gather}
	\label{E1}
	\operatorname{BS}_{\Sigma}(j)(x)=\frac{1}{4\pi}\int_{\Sigma}j(y)\times \frac{x-y}{|x-y|^3}d\sigma(y)\text{ for }x\in \mathbb{R}^3\setminus \Sigma.
\end{gather}
What we are interested in here is the question of how, for a given (plasma) domain $P\subset\mathbb{R}^3$, target magnetic field $B_T\in L^2(P,\mathbb{R}^3)$ and CWS $\Sigma\subset\mathbb{R}^3$, one may obtain a current distribution $j$ on $\Sigma$ such that $\|\operatorname{BS}_{\Sigma}(j)-B_T\|_{L^2(P)}\ll 1$.
\newline
One important and simple observation is that the current induced magnetic field satisfies the identities $\operatorname{div}(\operatorname{BS}_{\Sigma}(j))=0$ and $\operatorname{curl}(\operatorname{BS}_{\Sigma}(j))=0$ on $\mathbb{R}^3\setminus\Sigma$. The first equation is simply one of the Maxwell equations valid for all magnetic fields. The second equation physically amounts to saying that there is no current outside of $\Sigma$ which is obviously true since we consider a magnetic field induced by a current contained in $\Sigma$.
\newline
From this it is easy to deduce that if we want to be able to approximate the target magnetic field $B_T$ well, it must also satisfy the same equations $\operatorname{div}(B_T)=0$, $\operatorname{curl}(B_T)=0$ in $P$. This may however always be guaranteed by taking into account the plasma current as we shall discuss now. The main idea in order to obtain a suitable target field consists essentially in trying to find a plasma equilibrium magnetic field $B$ of the equations of magnetohydrodynamics, i.e. a solution of
\begin{gather}
	\label{E2}
	B\times \operatorname{curl}(B)=\nabla p\text{, }\operatorname{div}(B)=0\text{ in }P\text{ and }B,\operatorname{curl}(B)\parallel \partial P
\end{gather}
where $p$ is the pressure and in general $B$ is not curl-free. To extract a target field $B_T$ compatible with the curl-free condition one may consider $J:=\operatorname{curl}(B)$ which according to Maxwell's equations corresponds to the plasma current contained in $P$. One can then consider the magnetic field induced by the (volume-) plasma current $J$ which is once more given by the Biot-Savart law
\begin{gather}
	\label{E3}
	\operatorname{BS}_P(J)(x)=\frac{1}{4\pi}\int_PJ(y)\times \frac{x-y}{|x-y|^3}d^3y.
\end{gather}
Then, as $\operatorname{div}(J)=0$ and $J\parallel \partial P$, we find $\operatorname{div}(\operatorname{BS}_P(J))=0$ and $\operatorname{curl}(\operatorname{BS}_P(J))=J$, see \cite[Theorem A]{CDG01}, so that if we let $B_T:=B-\operatorname{BS}_P(J)$ we arrive at our desired div-free and curl-free magnetic field which we must reproduce by our coils in order to arrive at a desired plasma equilibrium.

We note that while $\operatorname{BS}_{\Sigma}(j)$ is always div- and curl-free inside $P$ it is easy to construct square integrable fields $B_T$ on $P$ which are div- and curl-free but not the magnetic field of any current on any CWS surrounding $P$ at a positive distance, \cite[Proposition 3.1]{G24}.

However, it has been shown \cite[Corollary 3.10 (iii,b)]{G24} that under some technical assumptions on $\Sigma$ and $P$, which are always satisfied in the context of stellarator designs, the image of $\operatorname{BS}_{\Sigma}$ is $L^2(P)$-dense in the space $L^2\mathcal{H}(P):=\{B_T\in L^2(P,\mathbb{R}^3)\mid \operatorname{div}(B)=0=\operatorname{curl}(B)\}$. In particular, for any target field $B_T$ one can (in theory) find a current distribution $j$ on $\Sigma$ with $\|\operatorname{BS}_{\Sigma}(j)-B_T\|_{L^2(P)}\ll 1$.
\newline
In practice one may reconstruct currents $j$ by means of a Tikhonov regularisation procedure, see \cite{PRS22},\cite{L17}.

The advantage, as well as disadvantage, of this regularisation approach is that it singles out current distributions in a way to make appropriate minimisation problems uniquely solvable, allowing to obtain approximating currents in terms of the unique solutions of the associated minimisation problems \cite{PRS22}. By reducing the regularisation parameter one may obtain a sequence of well-approximating currents \cite[Corollary 4.3]{G24}. The currents obtained in this way are $L^2(\Sigma)$-orthogonal to the kernel of the linear operator $\operatorname{BS}_{\Sigma}$.
If $j_0\in \operatorname{Ker}(\operatorname{BS}_{\Sigma})$ and $\alpha\in \mathbb{R}$, then obviously $\operatorname{BS}_{\Sigma}(j+\alpha j_0)=\operatorname{BS}_{\Sigma}(j)$ approximates $B_T$ as precise as does $\operatorname{BS}_{\Sigma}(j)$. Modifying $j$ by adding an element of the kernel provides one with flexibility and one may for instance search for modified currents $j+j_0$ such that $j+j_0$ optimises some other desirable physical feature such as reducing the Laplace force or the coil shape. We note however that the currents $j$ obtained by the procedure in \cite{PRS22} minimise the average (squared) current strength which may in itself be a desirable feature.
\newline
\newline
The goal of the present manuscript is twofold:
\begin{enumerate}
	\item We provide a simple characterisation of the image of the Biot-Savart operator. Based on that we provide a current reconstruction algorithm which does not rely on a regularisation procedure. This algorithm consists of two steps. In a first step, using the newly obtained characterisation of the image, we find an element $B$ of the image of the Biot-Savart operator which approximates well the given target field $B_T$. During this step the convergence rate cannot be easily controlled, but at the same time no singular integral kernels need to be computed and therefore this part of the algorithm is not so computationally complex. The second part of the algorithm consists of approximating a preimage $j$ of the previously found field $B$. During this step it will be necessary to compute singular integral kernels. However, this needs to be done only on a surface and we also show that the provided algorithm converges exponentially fast to a real preimage, so that few iterations are required to achieve a good precision. We further prove that the currents obtained in this way lead to the simplest possible coil shapes.
	\item We provide two ways to reconstruct the kernel of the Biot-Savart operator on a given CWS. The first is an algorithm which also requires the computation of double layer type integrals on a surface but also converges exponentially fast, with the same rate of convergence as the second part of the algorithm in the previous bullet point. The second approach characterises the kernel elements in terms of certain solutions to exterior boundary value problems and hence provides an alternative way to compute kernel elements.
\end{enumerate}
\textbf{Structure of the paper:} In section 2 we introduce the notation used throughout the manuscript and formulate the main results. Section 3 contains the proofs of the characterisation of the image of the Biot-Savart operator. In section 4 we prove the validity of the current reconstruction procedure. Section 5 contains the proof regarding the validity of the kernel reconstruction algorithms. We also include an appendix consisting of three parts. Appendix A discusses the equivalence between certain norms and specifically clarifies the relationship between the $L^2$-norm of a given magnetic field and its normal trace on the boundary as well as its toroidal circulation. Appendix B discusses the relation between the energy of the magnetic field induced by a surface current and certain surface norms of the currents themselves. Further, we analyse some dynamical properties of the kernel elements. The final part of the appendix includes a discussion of the Friedrichs decomposition on less regular domains which we make use of at certain parts of the manuscript.
\section{Main results}
\subsection{Notation}
By a $C^{1,1}$surface $\Sigma\subset\mathbb{R}^3$ we always mean a closed (i.e. compact and without boundary), connected $2$-manifold of class $C^{1,1}$. According to the Jordan-Brouwer-separation theorem, \cite{Li88}, $\mathbb{R}^3\setminus \Sigma$ consists of two connected components. One unbounded component and a bounded component. We will usually denote by $\Omega$ the corresponding bounded component and call it the bounded domain bounded by $\Sigma$ or the finite domain bounded by $\Sigma$. We further denote by $\mathcal{V}(\Sigma)$ the space of $C^{0,1}$-vector fields on $\Sigma$ which are tangent to $\Sigma$. By $L^2\mathcal{V}(\Sigma)$ we denote the completion of $\mathcal{V}(\Sigma)$ with respect to the $L^2(\Sigma)$-norm. We say that $j\in L^2\mathcal{V}(\Sigma)$ is div-free, which we also denote by $\operatorname{div}_{\Sigma}(j)=0$, if $\int_{\Sigma}j(x)\cdot \nabla \psi(x)d\sigma(x)=0$ for all $\psi\in C^1_c(\mathbb{R}^3)$. The space of all square integrable div-free fields is denoted by $L^2\mathcal{V}_0(\Sigma)$. In addition, we define the following norm on $L^2\mathcal{V}(\Sigma)$:
\begin{gather}
	\label{2E1}
	\|j\|_{W^{-\frac{1}{2},2}(\Sigma)}:=\sup_{\psi\in W^{\frac{1}{2},2}(\Sigma,\mathbb{R}^3)\setminus \{0\}}\frac{\left|\int_{\Sigma}j(x)\cdot \psi(x)d\sigma(x)\right|}{\|\psi\|_{W^{\frac{1}{2},2}(\Sigma)}}
\end{gather}
where $W^{\frac{1}{2},2}(\Sigma,\mathbb{R}^3)$ is the completion of $\mathcal{V}(\Sigma)$ with respect to the standard $W^{\frac{1}{2},2}(\Sigma)-$norm which may be taken to be the square root of the sum of the squares of the $W^{\frac{1}{2},2}(\Sigma)$-norms of the components of $\psi$ and $\|f\|^2_{W^{\frac{1}{2},2}(\Sigma)}:=\|f\|^2_{L^2(\Sigma)}+\int_{\Sigma}\int_{\Sigma}\frac{|f(x)-f(y)|^2}{|x-y|^3}d\sigma(y)$ for scalar functions $f$. We then denote by $W^{-\frac{1}{2},2}\mathcal{V}_0(\Sigma)$ the completion of $L^2\mathcal{V}_0(\Sigma)$ with respect to the norm $\|\cdot\|_{W^{-\frac{1}{2},2}(\Sigma)}$. If $\Omega\subset \mathbb{R}^3$ is a bounded $C^{1,1}$-domain with disconnected boundary we make the same definitions where $\Sigma$ is replaced by $\partial\Omega$ accordingly.

Lastly we introduce the following function spaces for a given bounded domain $\Omega\subset\mathbb{R}^3$: $H(\operatorname{curl},\Omega):=\{w\in L^2(\Omega,\mathbb{R}^3)\mid \operatorname{curl}(w)\in L^2(\Omega,\mathbb{R}^3)\}$, $H(\operatorname{div},\Omega):=\{w\in L^2(\Omega,\mathbb{R}^3)\mid \operatorname{div}(w)\in L^2(\Omega)\}$ and $L^2\mathcal{H}(\Omega):=\{w\in L^2(\Omega,\mathbb{R}^3)\mid \operatorname{div}(w)=0=\operatorname{curl}(w)\}$ where $\operatorname{curl}$ and $\operatorname{div}$ are understood in the weak sense and we equip the spaces with the norms $\|w\|_{H(\operatorname{curl},\Omega)}:=\sqrt{\|w\|^2_{L^2(\Omega)}+\|\operatorname{curl}(w)\|_{L^2(\Omega)}},\|w\|_{H(\operatorname{div},\Omega)}:=\sqrt{\|w\|^2_{L^2(\Omega)}+\|\operatorname{div}(w)\|_{L^2(\Omega)}}$ and $\|w\|_{L^2(\Omega)}$ respectively.
\subsection{Statement and discussion of main results}
\subsubsection{Image of the Biot-Savart operator}
Given a bounded domain $\Omega\subset\mathbb{R}^3$ with (possibly disconnected) $C^{1,1}$-boundary $\partial\Omega$ we may consider the following operator (recall $L^2\mathcal{H}(\Omega)$ denotes the square integrable, div- and curl-free fields on $\Omega$)
\begin{gather}
	\nonumber
	\operatorname{BS}_{\partial\Omega}:L^2\mathcal{V}_0(\partial\Omega)\rightarrow L^2\mathcal{H}(\Omega)\text{, }j\mapsto \left(x\mapsto \frac{1}{4\pi}\int_{\partial\Omega}j(y)\times \frac{x-y}{|x-y|^3}d\sigma(y)\right)
\end{gather}
which gives rise to a well-defined bounded linear operator \cite[Lemma 5.5]{G24}. The crucial observation is that the above operator remains continuous if we equip the space $L^2\mathcal{V}_0(\partial\Omega)$ with the norm $\|\cdot\|_{W^{-\frac{1}{2},2}(\partial\Omega)}$ as defined in (\ref{2E1}), see \cite[Lemma C.1]{G24} so that the Biot-Savart operator extends uniquely to the space $W^{-\frac{1}{2},2}\mathcal{V}_0(\partial\Omega)$. To make our setting precise we make the following definition
\begin{defn}[Biot-Savart operator]
	\label{2D1}
	Let $\Omega\subset\mathbb{R}^3$ be a bounded $C^{1,1}$-domain with possibly disconnected boundary. Then we define the \textit{Biot-Savart operator} as
	\begin{gather}
		\nonumber
		\operatorname{BS}_{\partial\Omega}:W^{-\frac{1}{2},2}\mathcal{V}_0(\partial\Omega)\rightarrow L^2\mathcal{H}(\Omega)\text{, }j\mapsto \left(x\mapsto \frac{1}{4\pi}\int_{\partial\Omega}j(y)\times \frac{x-y}{|x-y|^3}d\sigma(y)\right)
	\end{gather}
	which is a well-defined, bounded linear operator.
\end{defn}
Before we characterise the image of this operator we introduce the space of harmonic Dirichlet fields
\begin{gather}
	\label{2E2}
	\mathcal{H}_D(\Omega):=\{\nabla f\mid f\in H^1(\Omega), \Delta f=0\text{, }\nabla f\times \mathcal{N}=0\}
\end{gather}
where $\mathcal{N}$ denotes the outward unit normal on $\partial\Omega$ and the identities $\Delta f=0$ and $\nabla f\times \mathcal{N}=0$ are understood in the weak sense, i.e. $\int_{\Omega}\nabla f\cdot \operatorname{curl}(\psi)d^3x=0$ for all $\psi\in C_c^1(\mathbb{R}^3,\mathbb{R}^3)$ (note that we allow $\psi$ to have non-zero boundary values) and $\int_{\Omega}\nabla f\cdot \nabla \phi d^3x=0$ for all $\phi\in C^1_c(\Omega)$.
We have the following characterisation of the image of the Biot-Savart operator.
\begin{thm}[Image of the Biot-Savart operator]
	\label{2T2}
	Let $\Omega\subset\mathbb{R}^3$ be a bounded $C^{1,1}$-domain with possibly disconnected boundary. Then
	\begin{gather}
		\nonumber
		\operatorname{Im}(\operatorname{BS}_{\partial\Omega})=L^2\mathcal{H}(\Omega)\cap \mathcal{H}^{\perp_{L^2(\Omega)}}_D(\Omega)
	\end{gather}
	where $\mathcal{H}^{\perp_{L^2(\Omega)}}_D(\Omega)$ denotes the $L^2(\Omega)$-orthogonal complement of $\mathcal{H}_D(\Omega)$ within $L^2\mathcal{H}(\Omega)$.
\end{thm}
We observe first that $\dim\left(\mathcal{H}_D(\Omega)\right)=\#\partial\Omega-1$ where $\#\partial\Omega$ denotes the number of connected components of $\partial\Omega$, c.f. \cite[Hodge Decomposition Theorem]{CDG02}. Further, the kernel of the Biot-Savart operator has been investigated in \cite{G24} with the following findings
\begin{thm}[{\cite[Theorem 5.1, Remark C.2]{G24}}]
	\label{2T3}
	Let $\Omega\subset\mathbb{R}^3$ be a bounded $C^{1,1}$-domain with possibly disconnected boundary. Then $\dim\left(\operatorname{Ker}(\operatorname{BS}_{\partial\Omega})\right)=g(\partial\Omega)$ where $g(\partial\Omega)$ denotes the genus of $\partial\Omega$ and equals the sum of the genera of the connected components of $\partial\Omega$ in case $\partial\Omega$ is disconnected.
\end{thm}
As an immediate consequence we obtain the following corollary.
\begin{cor}[Biot-Savart operator is Fredholm]
	\label{2C4}
	Let $\Omega\subset\mathbb{R}^3$ be a bounded $C^{1,1}$-domain. Then
	\begin{gather}
		\nonumber
		\operatorname{BS}_{\partial\Omega}:W^{-\frac{1}{2},2}\mathcal{V}_0(\partial\Omega)\rightarrow L^2\mathcal{H}(\Omega)
	\end{gather}
	is a Fredholm operator of index $\operatorname{ind}(\operatorname{BS}_{\partial\Omega})=g(\partial\Omega)-\#\partial\Omega+1$. In particular, the Fredholm index of the Biot-Savart operator is a topological invariant, i.e. if two $C^{1,1}$-bounded domains are homeomorphic, then the Fredholm-indices of their corresponding Biot-Savart operators coincide.
\end{cor}
\subsubsection{Current reconstruction algorithm}
As discussed in the introduction, we intend to deal with the following inverse problem of relevance in plasma physics: Given a bounded $C^{1,1}$-solid torus (which corresponds to the plasma region) and a $C^{1,1}$-surface $\Sigma$ (corresponding to the CWS) such that $\overline{P}\subset\Omega$ where $\Omega$ is the finite region enclosed by $\Sigma$ and an element $B_T\in L^2\mathcal{H}(P)$, find for given $\epsilon>0$ a $j\in W^{-\frac{1}{2},2}\mathcal{V}_0(\Sigma)$ with $\|\operatorname{BS}_{\Sigma}(j)-B_T\|_{L^2(P)}\leq \epsilon$. In general it is not possible to approximate arbitrary target fields arbitrarily well by elements of the image of the Biot-Savart operator, c.f. \cite[Corollary (iii,a)]{G24}. However, under certain natural assumptions it becomes possible. For a given solid torus $P$, we call a closed $C^1$-curve $\gamma\subset \partial P$ poloidal if it represents a trivial element of the fundamental group when viewed as a curve in $\overline{P}$ but represents a non-trivial element of the fundamental group as a curve in $\partial P$.
\begin{thm}[{\cite[Corollary 3.10 (iii,b)]{G24}}]
	\label{2T5}
	Let $\Sigma$ be a $C^{1,1}$-surface which bounds a solid torus $\Omega$ and let $P\subset \Omega$ be another $C^{1,1}$-solid torus with $\overline{P}\subset \Omega$. Further, suppose that $\Omega$ contains a smooth disc $D$ with $C^1$-boundary such that $\partial D\subset \Sigma$ is a poloidal $C^1$-curve and such that $D\cap \partial P$ is a poloidal $C^1$-curve in $\partial P$, see \Cref{figure1}. Then for every $B_T\in L^2\mathcal{H}(P)$ and every $\epsilon>0$ there is some $j\in L^2\mathcal{V}_0(\Sigma)$ such that $\|\operatorname{BS}_{\Sigma}(j)-B_T\|_{L^2(P)}\leq \epsilon$.
\end{thm}
\begin{figure}[H]
	\hspace{40mm}\includegraphics[width=0.5\textwidth, keepaspectratio]{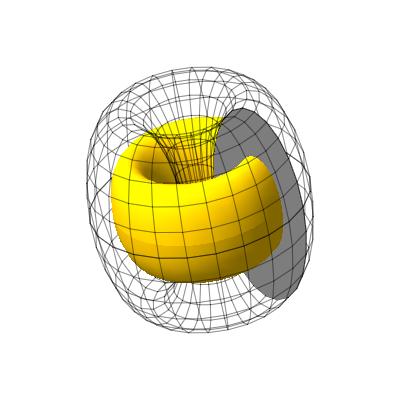}
	\caption{The plasma domain depicted in yellow. The CWS $\Sigma$ depicted by the black grid and the disc $D$ depicted in grey. The disc $D$ bounds a poloidal curve on $\Sigma$ as well as on the boundary of the plasma domain.}
	\label{figure1}
\end{figure}
A rough outline of the algorithm we are about to propose is the following
\begin{enumerate}
	\item Exploit \Cref{2T2} to find some $B\in \operatorname{Im}(\operatorname{BS}_{\Sigma})$ with $\|B-B_T\|_{L^2(P)}\leq \epsilon$.
	\item Knowing that $B\in \operatorname{Im}(\operatorname{BS}_{\Sigma})$ find $j\in W^{-\frac{1}{2},2}\mathcal{V}_0(\Sigma)$ with $\operatorname{BS}_{\Sigma}(j)=B$.
\end{enumerate}
\textbf{\underline{Step 1:}}
\newline
\newline
Define the following two subspaces of $L^2\mathcal{H}(\Omega)$ known as the harmonic Neumann fields and exact harmonic fields respectively
\begin{gather}
	\label{2E3}
	\mathcal{H}_N(\Omega):=\{\Gamma\in L^2(\Omega,\mathbb{R}^3)\mid \operatorname{div}(\Gamma)=0=\operatorname{curl}(\Gamma)\text{, }\Gamma\cdot \mathcal{N}=0\}\text{, }\mathcal{H}_{\operatorname{ex}}(\Omega):=\{\nabla f\mid f\in H^1(\Omega)\text{, }\Delta f=0\}
\end{gather}
where the imposed conditions are understood in the weak sense. In particular, $\Gamma$ being div-free and tangent to $\Sigma$ is equivalent to the statement $\int_{\Omega}\Gamma\cdot \nabla \phi d^3x=0$ for all $\phi\in C^1_c(\mathbb{R}^3)$ (note $\phi$ does not need to be supported in $\Omega$). The relevance of these two spaces is that according to the Hodge-decomposition theorem, c.f. \cite[Hodge decomposition theorem]{CDG02} and \cite[Corollary 3.5.2]{S95} for the smooth setting and \cite[Theorem B.1]{G24} for the $C^{1,1}$-setting, we have the $L^2(\Omega)$-orthogonal decomposition 
\begin{gather}
	\label{2E4}
	L^2\mathcal{H}(\Omega)=\mathcal{H}_{\operatorname{ex}}(\Omega)\oplus \mathcal{H}_N(\Omega).
\end{gather}
We note that $\dim\left(\mathcal{H}_N(\Omega)\right)=g(\partial \Omega)$, c.f. \cite[Hodge decomposition theorem]{CDG02}, and consequently $\dim\left(\mathcal{H}_N(\Omega)\right)=1$ whenever $\Omega$ is a solid torus. On the other hand, it is standard that for any given $\kappa\in W^{\frac{1}{2},2}(\Sigma)$ there exists a unique solution $f\in H^1(\Omega)$ to the following boundary value problem (BVP), 
\begin{gather}
	\label{2E5}
	\Delta f=0\text{ in }\Omega\text{ and }f|_{\partial\Omega}=\kappa.
\end{gather}
The existence can be seen upon extending $\kappa$ to some $k\in H^1(\Omega)$, \cite[Proposition 3.31]{DD12}, and then decomposing $\nabla k$ according to the Hodge-decomposition theorem, \cite[Theorem B.1]{G24}, $\nabla k=\nabla g+\nabla h$ for suitable $g\in W^{1,2}_0(\Omega)$ and $\nabla h\in \mathcal{H}_{\operatorname{ex}}(\Omega)$. We conclude $k=g+h+c$ for some $c\in \mathbb{R}$ and so setting $f:=h+c$ we find $\Delta f=0$ in $\Omega$ and $f|_{\partial\Omega}=\kappa$. Uniqueness of solutions follows from the uniqueness of the solution of the homogenous equation $\Delta f=0$ and $f|_{\partial\Omega}=0$, which is just a consequence of the integration by parts formula $\int_{\Omega} \nabla f \cdot \nabla fd^3x=-\int_{\Omega}f \Delta fd^3x=0$ where we used that $f|_{\partial \Omega}=0$ and $\Delta f=0$. Consequently $f$ is constant and since $f|_{\partial\Omega}=0$ we must have $f=0$ everywhere, proving the uniqueness of solutions. For more general existence and uniqueness results in $C^{1,1}$-domains we refer to \cite[Theorem 2.4.2.5]{Gris85}.
We can therefore obtain a basis of the image of the Biot-Savart operator without the need to work with the Biot-Savart operator itself.
\begin{thm}[Current reconstruction algorithm, Step 1]
	\label{2T6}
	Let $\Sigma\subset\mathbb{R}^3$ be a $C^{1,1}$-surface which bounds a solid torus $\Omega$ and let $P\subset \Omega$ be another $C^{1,1}$-solid torus with $\overline{P}\subset \Omega$. Assume further that $\Omega$ contains a smooth disc $D$ with $C^1$-boundary such that $\partial D\subset \Sigma$ is a poloidal curve and such that $D\cap \partial P$ is a poloidal $C^1$-curve in $\partial P$. Then for any $B_T\in L^2\mathcal{H}(P)$, any $\epsilon>0$, any $\Gamma\in \mathcal{H}_N(\Omega)\setminus \{0\}$ and any basis $\{\kappa_1,\kappa_2,\dots\}$ of $W^{\frac{1}{2},2}(\Sigma)$ there exists some $N\in \mathbb{N}$ and $\alpha_0,\alpha_1,\dots,\alpha_N\in \mathbb{R}$ such that
	\begin{gather}
		\label{2E6}
		\left\|\left(\alpha_0\Gamma+\sum_{k=1}^N\alpha_i\nabla f_i\right)-B_T\right\|_{L^2(P)}\leq \epsilon
	\end{gather}
	where the $f_i$ are the unique solutions to the BVPs (\ref{2E5}).
\end{thm}
\begin{rem}
	\label{2R7}
	\begin{enumerate}
		\item We observe that according to \Cref{2T2} and the fact that $\mathcal{H}_D(\Omega)=\{0\}$, since $\partial \Omega$ is connected, we find that the approximating vector field $B:=\alpha_0\Gamma+\sum_{k=1}^N\alpha_i\nabla f_i$ lies in the image of the Biot-Savart operator.
		\item In order to obtain a basis of $W^{\frac{1}{2},2}(\Sigma)$ one may make use of the fact that $\Sigma$ is a torus. In the realm of plasma physics the CWSs may be modelled as embeddings $\Psi:T^2\rightarrow\Sigma$ where $T^2$ denotes the standard flat $2$-torus viewed as a square with opposite sides identified and usually this embedding is expressed in terms of Fourier coefficients which are used as free-parameters in order to adjust the CWS structure to satisfy desirable features, c.f. \cite[Section 4]{PRS22}.
		From this perspective we see that conceptually if $\Psi$ is a $C^{1,1}$-diffeomorphism we may start with any basis of $W^{\frac{1}{2},2}(T^2)$ and by composition with $\Psi^{-1}$ this will provide us with a basis of $W^{\frac{1}{2},2}(\Sigma)$. For instance, one may start with a standard Fourier basis on $T^2$ satisfying $\Delta \hat{\kappa}_i=\lambda_i\hat{\kappa}_i$ (we compute all quantities on $T^2$ with respect to the flat metric). This basis is known explicitly. According to \cite[Proposition 3.40]{DD12} and its proof the $C^{1,1}(T^2)$-functions are dense in $W^{\frac{1}{2},2}(T^2)$. So given any $\hat{\kappa}\in W^{\frac{1}{2},2}(T^2)$ we may find for given $\epsilon>0$ some $\hat{h}\in C^{1,1}(T^2)$ with $\|\hat{h}-\hat{\kappa}\|_{W^{\frac{1}{2},2}(T^2)}\leq \frac{\epsilon}{2}$. On the other hand, since $\hat{h}\in W^{2,2}(T^2)$ we note that the Fourier series $\hat{h}=\sum_{k=1}^\infty a_k\hat{\kappa}_k$ with $a_k=\int_{T^2}\hat{h}\cdot \hat{\kappa}_k d^2x$ in fact converges in $H^1(T^2)$ and consequently in $W^{\frac{1}{2},2}(T^2)$. We conclude that the span of the $\hat{\kappa}_k$ is dense in $W^{\frac{1}{2},2}(T^2)$ and therefore provides a (non-orthonormal) basis of this space.
	\end{enumerate}
\end{rem}
\textbf{\underline{Step 2:}}
\newline
\newline
	According to the first step we are left with reconstructing a preimage for an element $B\in \operatorname{Im}(\operatorname{BS}_{\Sigma})$. Before we come to the construction itself we recall here the definition of the double layer potential $w_{\Omega}$ and its transpose $w_{\Omega}^{\operatorname{Tr}}$ for a given bounded $C^{1,1}$-domain $\Omega\subset\mathbb{R}^3$ (with possibly disconnected boundary)
	\begin{gather}
		\nonumber
		w_{\Omega}:W^{\frac{1}{2},2}(\partial\Omega)\rightarrow W^{\frac{1}{2},2}(\partial\Omega)\text{, }\phi\mapsto \left(x\mapsto \frac{1}{4\pi}\int_{\partial\Omega}\phi(y)\mathcal{N}(y)\cdot \frac{x-y}{|x-y|^3}d\sigma(y)\right)
		\\
		\label{2E7}
		w^{\operatorname{Tr}}_{\Omega}:W^{-\frac{1}{2},2}(\partial\Omega)\rightarrow W^{-\frac{1}{2},2}(\partial\Omega)\text{, }\phi\mapsto \left(x\mapsto \frac{1}{4\pi}\int_{\partial\Omega}\phi(y)\mathcal{N}(x)\cdot \frac{x-y}{|x-y|^3}d\sigma(y)\right)
	\end{gather}
	which give rise to bounded, linear operators, c.f. \Cref{3L1}. More precisely, we let $W^{-\frac{1}{2},2}(\partial\Omega)$ denote the topological dual space of $W^{\frac{1}{2},2}(\partial\Omega)$ and for given $g\in W^{-\frac{1}{2},2}(\partial\Omega)$ we define $w^{\operatorname{Tr}}_{\Omega}(g)(f):=-g(w_{\Omega}(f))$ for $f\in W^{\frac{1}{2},2}(\partial\Omega)$. Further, we identify $W^{\frac{1}{2},2}(\partial\Omega)$ with a subspace of $W^{-\frac{1}{2},2}(\partial\Omega)$ via $h \mapsto \left(f\mapsto \int_{\partial\Omega}f(x)\cdot h(x)d\sigma(x)\right)$.

	In the upcoming formulation we also make use of the concept of tangential traces, i.e. of the fact that there exists a unique, linear, bounded operator $T:H(\operatorname{div},\Omega)\rightarrow W^{-\frac{1}{2},2}(\partial\Omega)\text{, }X\mapsto \mathcal{N}\cdot X$ such that for every $X\in C^1(\overline{\Omega})$, $T(X)=X|_{\partial\Omega}\cdot \mathcal{N}$ is given by the product of the classical restriction of $X$ to $\partial\Omega$ and the outward unit normal $\mathcal{N}$ on $\partial\Omega$, c.f. \cite[Theorem 2.5]{GR86}. Similarly, one can extend the mapping $X\mapsto X\times \mathcal{N}$ defined in the classical sense on $C^1(\overline{\Omega},\mathbb{R}^3)$ to a bounded, linear map from $H(\operatorname{curl},\Omega)$ into $(W^{-\frac{1}{2},2}(\partial\Omega))^3$, c.f. \cite[Theorem 2.11]{GR86}.
	\newline
	The first result provides an exact preimage for a given element of the Biot-Savart operator.
	\begin{lem}[Constructing preimages]
		\label{2L8}
		Let $\Omega\subset\mathbb{R}^3$ be a bounded $C^{1,1}$-domain with possibly disconnected boundary and let $B\in L^2\mathcal{H}(\Omega)\cap \mathcal{H}^{\perp_{L^2(\Omega)}}_D(\Omega)$. Then $\left(\frac{\operatorname{Id}}{2}+w^{\operatorname{Tr}}_{\Omega}\right)$ admits a bounded, linear inverse and the following boundary value problem (BVP) admits a unique solution $f\in H^1(\Omega)$
			\begin{gather}
				\label{2E8}
				\Delta f=0\text{ in }\Omega\text{, }\mathcal{N}\cdot \nabla f=\left(\frac{\operatorname{Id}}{2}+w^{\operatorname{Tr}}_{\Omega}\right)^{-1}(B\cdot \mathcal{N})\text{ on }\partial\Omega\text{ and }\int_{\Omega}f(x)d^3x=0.
			\end{gather}
			Let $\Gamma\in \mathcal{H}_N(\Omega)$ denote the $L^2(\Omega)$-orthogonal projection of $B$ onto $\mathcal{H}_N(\Omega)$. Then
			\begin{gather}
				\nonumber
				j:=\Gamma\times \mathcal{N}+\nabla f\times \mathcal{N}
			\end{gather}
			is a well-defined element in $W^{-\frac{1}{2},2}\mathcal{V}_0(\partial\Omega)$ with $\operatorname{BS}_{\partial\Omega}(j)=B$.
	\end{lem}
	The main observation now is that the inverse $\left(\frac{\operatorname{Id}}{2}+w^{\operatorname{Tr}}_{\Omega}\right)^{-1}$ admits a Neumann series expression, i.e. we can write $\left(\frac{\operatorname{Id}}{2}+w^{\operatorname{Tr}}_{\Omega}\right)^{-1}(B\cdot \mathcal{N})=\sum_{k=0}^\infty\left(\frac{\operatorname{Id}}{2}-w^{\operatorname{Tr}}_{\Omega}\right)^{k}(B\cdot \mathcal{N})$ which we may truncate to obtain approximate solutions. 
	\begin{thm}[Current reconstruction algorithm, Step 2]
		\label{2T9}
		Let $\Omega\subset\mathbb{R}^3$ be a bounded $C^{1,1}$-domain with possibly disconnected boundary and $B\in L^2\mathcal{H}(\Omega)\cap \mathcal{H}_D^{\perp_{L^2(\Omega)}}(\Omega)$. Define $b_n:=\sum_{k=0}^{n}\left(\frac{\operatorname{Id}}{2}-w^{\operatorname{Tr}}_{\Omega}\right)^{k}(B\cdot \mathcal{N})\in W^{-\frac{1}{2},2}(\partial\Omega)$. Then the following BVPs have unique solutions $f_n\in H^1(\Omega)$
		\begin{gather}
			\label{2E9}
			\Delta f_n=0\text{ in }\Omega\text{, }\mathcal{N}\cdot \nabla f_n=b_n\text{ on }\partial\Omega\text{ and }\int_{\Omega}f_n(x)d^3x=0.
		\end{gather}
		Let $\Gamma\in \mathcal{H}_N(\Omega)$ denote the $L^2(\Omega)$-orthogonal projection of $B$ onto $\mathcal{H}_N(\Omega)$. Then $j_n:=\Gamma\times \mathcal{N}+\nabla f_n\times \mathcal{N}\in W^{-\frac{1}{2},2}\mathcal{V}_0(\partial\Omega)$ for all $n\in \mathbb{N}$ and there exist constants $0<c_1(\partial\Omega),c_2(\partial\Omega)<\infty$ and $0<\lambda(\partial\Omega)<1$ which are independent of $B$ and $n$ such that the following estimates hold true where $j$ is the exact preimage from \Cref{2L8}
		\begin{gather}
			\label{2E10}
			\|j_n-j\|_{W^{-\frac{1}{2},2}(\partial\Omega)}\leq \frac{c_1\lambda^{n+1}}{1-\lambda}\|B\|_{L^2(\Omega)},
		\\
		\label{2E11}
		\|\operatorname{BS}_{\partial\Omega}(j_n)-B\|_{L^2(\Omega)}\leq \frac{c_2\lambda^{n+1}}{1-\lambda}\|B\|_{L^2(\Omega)}.
		\end{gather}
	\end{thm}
\textbf{\underline{End of the algorithm}}
\newline
\newline
\textbf{\underline{Shape complexity of constructed currents:}}
\newline
\newline
We now explain why the currents $j_n$ and $j$ obtained from \Cref{2L8} and \Cref{2T9} have in a sense the "simplest" form possible. To this end we focus on solid $C^{1,1}$-tori $\Omega\subset\mathbb{R}^3$. In recent work \cite{G24SurfaceHelicityArXiv} the notions of asymptotic toroidal and poloidal windings, \cite[Definition 2.9]{G24SurfaceHelicityArXiv}, have been studied on toroidal surfaces within the context of plasma physics. We recall here the necessary definitions. Given a bounded $C^{1,1}$-solid torus $\Omega\subset\mathbb{R}^3$ we can fix two closed curves $\sigma_p,\sigma_t$ on $\Sigma:=\partial\Omega$ such that $\sigma_p$ defines a non-trivial element of the fundamental group of $\Sigma$ but is the trivial element of the fundamental group of $\overline{\Omega}$ and such that $\sigma_t$ defines a non-trivial element of the fundamental group of $\Sigma$ but is contractible within $\mathbb{R}^3\setminus \Omega$. The curves $\sigma_p$ and $\sigma_t$ are called poloidal and toroidal respectively. One can then define the space of harmonic fields on $\Sigma$ by $\mathcal{H}(\Sigma):=\{\gamma\in L^2\mathcal{V}(\Sigma)\mid \operatorname{div}_{\Sigma}(\gamma)=0=\operatorname{curl}_{\Sigma}(\gamma)\}$ which is $2$-dimensional by standard Hodge theory. There is then a unique basis $\gamma_p,\gamma_t\in \mathcal{H}(\Sigma)$ determined by the equations
\begin{gather}
	\nonumber
	\int_{\sigma_p}\gamma_t=0=\int_{\sigma_t}\gamma_p\text{ and }\int_{\sigma_t}\gamma_t=1=\int_{\sigma_p}\gamma_p
\end{gather}
and this basis is independent of the chosen poloidal and toroidal curves (with the exception that a change of orientation of the curves will result in an additional minus sign of the corresponding basis elements). If $\sigma$ is any other closed $C^1$-curve on $\Sigma$, then we can express it as a concatenation $\sigma=P\sigma_p\oplus Q\sigma_t$ since $\sigma_p$ and $\sigma_t$ generate the first fundamental group of $\Sigma$ and where $P$ corresponds to the amount of poloidal windings and $Q$ corresponds to the amount of toroidal windings within one period of $\sigma$. One can then show that $\int_{\sigma}\gamma_t=Q$ and $\int_{\sigma}\gamma_p=P$ and that we have the identities $\lim_{T\rightarrow\infty}\frac{1}{T}\int_{\sigma[0,T]}\gamma_t=\frac{Q}{\tau}$, $\lim_{T\rightarrow\infty}\frac{1}{T}\int_{\sigma[0,T]}\gamma_p=\frac{P}{\tau}$ where $\tau$ is the period of $\sigma$ and $\sigma[0,T]$ denotes the (possibly non-closed) curve $\sigma:[0,T]\rightarrow \Sigma$. The key observation is that for the (possibly non-periodic) integral curves $\sigma_x$ starting at some point $x\in \Sigma$ of any div-free field $j\in \mathcal{V}(\Sigma)$ the limit $\hat{q}(x):=\lim_{T\rightarrow\infty}\frac{1}{T}\int_{\sigma_x[0,T]}\gamma_t$ exists for a.e. $x\in \Sigma$ and is integrable. In correspondence with the case of closed curves we may interpret $\hat{q}(x)$ as the weighted asymptotic toroidal windings of the field line $\sigma_x$ of $j$ starting at $x$. The average of the toroidal windings of the integral curves of $j$ can then be expressed by the following integral, c.f. \cite[Lemma 2.12]{G24SurfaceHelicityArXiv},
\begin{gather}
	\nonumber
	\overline{Q}(j):=\frac{1}{|\Sigma|}\int_{\Sigma}\hat{q}(x)d\sigma(x)=\frac{1}{|\Sigma|}\int_{\Sigma}j(x)\cdot \gamma_t(x)d\sigma(x)
\end{gather}
where $|\Sigma|$ denotes the area of $\Sigma$. Loosely speaking, the fact that $\overline{Q}(j)\neq 0$ tells us that on average the field lines of $j$ will wind toroidally along $\Sigma$. In contrast, $\overline{Q}(j)=0$ implies that the field lines of $j$ on average wind in poloidal direction along $\Sigma$, even though one has to be careful with this interpretation since it might happen that the field lines of $j$ wind in opposite toroidal directions so that the average toroidal contributions cancel each other, see \cite[Figure 4]{G24SurfaceHelicityArXiv} for an example. Nonetheless we may interpret the condition $\overline{Q}(j)=0$ that the field lines of $j$ tend to be more poloidal and hence $j$ having a "simple" shape, in contrast to the situation $\overline{Q}(j)\neq 0$ where we expect on average to observe toroidal windings of the field lines of the current distribution.
For less regular currents we can still make the following definition for any given $C^{1,1}$-surface $\Sigma\subset\mathbb{R}^3$
\begin{gather}
	\label{2E12}
	\overline{Q}:W^{-\frac{1}{2},2}\mathcal{V}_0(\Sigma)\rightarrow\mathbb{R}\text{, }j\mapsto \frac{1}{|\Sigma|}\int_{\Sigma}j(x)\cdot \gamma_t(x)d\sigma(x)
\end{gather}
which defines a linear, bounded operator since $\mathcal{H}(\Sigma)\subset W^{\frac{1}{2},2}\mathcal{V}(\Sigma)$. We have the following result regarding the current complexity.
\begin{prop}[Shape complexity of constructed currents]
	\label{2P10}
	Let $\Omega\subset \mathbb{R}^3$ be a $C^{1,1}$-solid torus and $\Sigma:=\partial\Omega$. Then for every $\Gamma\in \mathcal{H}_N(\Omega)$ and $\nabla f\in \mathcal{H}_{\operatorname{ex}}(\Omega)$ we have $J:=\Gamma\times \mathcal{N}+\nabla f\times \mathcal{N}\in W^{-\frac{1}{2},2}\mathcal{V}_0(\Sigma)$ and $\overline{Q}(J)=0$. In particular, for given $B\in L^2\mathcal{H}(\Omega)$, the exact preimage $j$ from \Cref{2L8} and the approximating currents $j_n$ defined in \Cref{2T9} all satisfy $\overline{Q}(j)=0=\overline{Q}(j_n)$ for all $n$.
\end{prop}
\Cref{2P10} therefore tells us that our reconstructed currents have in a sense the simplest possible shape. Working with such currents can be of importance when one wishes to find simple coil designs since $\overline{Q}(j)\neq 0$ implies that $j$ has on average non-zero toroidal windings and therefore one needs to use coils which wind toroidally around the plasma, whereas $\overline{Q}(j)=0$ tells us that we may be able to approximate $j$ well by coils which wind only poloidally around the plasma.
\newline
\newline
\textbf{\underline{Comparison to previous current reconstruction algorithms:}}
\newline
\newline
Let us compare our reconstruction procedure with other procedures studied in the literature. We focus here on a reconstruction procedure which was for instance studied in \cite{PRS22}. Let $\Omega\subset \mathbb{R}^3$ be a $C^{1,1}$-solid torus which corresponds to the finite region bounded by the CWS $\Sigma:=\partial\Omega$ and let $P$ be another $C^{1,1}$-solid torus with $\overline{P}\subset\Omega$ which corresponds to the plasma region. Given a target magnetic field $B_T\in L^2\mathcal{H}(P)$ one wishes to solve the minimisation problem
\begin{gather}
	\nonumber
	\arg\min_{j}\|\operatorname{BS}_{\Sigma}(j)-B_T\|^2_{L^2(P)}.
\end{gather}
The issue is that this minimisation problem does not admit a solution and so one introduces a regularising parameter $\lambda>0$ and considers the modified minimisation problem
\begin{gather}
	\label{2E13}
	\min_{j\in L^2\mathcal{V}_0(\Sigma)}\left(\|\operatorname{BS}_{\Sigma}(j)-B_T\|^2_{L^2(P)}+\lambda\|j\|^2_{L^2(\Sigma)}\right).
\end{gather}
It follows from standard variational techniques and the convexity of the functional involved that for every $\lambda>0$ there exists a unique current distribution $j_{\lambda}\in L^2\mathcal{V}_0(\Sigma)$ which realises the global minimum in (\ref{2E13}). It has been then shown in \cite[Corollary 4.3]{G24} that under the same assumptions on $\Omega$ and $P$ as in \Cref{2T6} we have $\lim_{\lambda\searrow 0}\|\operatorname{BS}_{\Sigma}(j_{\lambda})-B_T\|_{L^2(P)}=0$. Consequently, the minimisers $j_{\lambda}$ may serve as our desired reconstructed currents.
\newline
\newline
\underline{\textit{Comparison from a theoretical point of view:}} By construction of the regularisation procedure, the minimising currents $j_{\lambda}$ of (\ref{2E13}) are necessarily $L^2(\Sigma)$-orthogonal to $\operatorname{Ker}(\operatorname{BS}_{\Sigma})$. On the other hand, we note that the condition $\overline{Q}(j)=0$ for a $j\in L^2\mathcal{V}_0(\Sigma)$ is equivalent to the statement that $j$ is $L^2(\Sigma)$-orthogonal to the $1$-dimensional span of $\gamma_t\in \mathcal{H}(\Sigma)$ defined by the relations $\int_{\sigma_p}\gamma_t=0$ and $\int_{\sigma_t}\gamma_t=1$ for fixed poloidal and toroidal curves $\sigma_p$ and $\sigma_t$ respectively. The kernel of $\operatorname{BS}_{\Sigma}$ does not need to coincide with the span of $\gamma_t$ and therefore these two orthogonality conditions are generally distinct.

So if one wishes to obtain approximating currents which minimise the average (squared) current strength one should preferably use the approach via the minimisation problem (\ref{2E13}). If, on the other hand, one is willing to accept potentially stronger currents at the expense of simplicity of the coil design it would be preferable to follow the current reconstruction algorithm proposed in the present manuscript or impose the additional constraint $\overline{Q}(j)=0$ in the minimisation problem (\ref{2E13}). The convergence of $\|\operatorname{BS}_{\Sigma}(j_{\lambda})-B_T\|_{L^2(P)}\rightarrow 0$ under the additional constraint, $\overline{Q}(j)=0$, has been established recently \cite{G24SurfaceHelicityArXiv}.
\newline
\newline
\underline{\textit{Comparison from a practical point of view:}} We note that in practice in order to find approximations to the minimisers of the functional in (\ref{2E13}) one has to compute a volume integral involving the Biot-Savart operator which is costly. Therefore, it is customary, when doing numerical computations, to work instead on the plasma boundary $\partial P$ and instead try to minimise the norm $\|\mathcal{N}\cdot \operatorname{BS}_{\Sigma}(j)-\mathcal{N}\cdot B_T\|_{L^2(\partial P)}$ and prescribe the total poloidal current of $j$, see \cite[Section 4.1.2]{PRS22}. To be more precise if we fix a toroidal loop $\sigma_t$ within the plasma domain $P$ we may define $I_p:=\int_{\sigma_t}B_T$ where $B_T\in L^2\mathcal{H}(P)$ is our given target field. Every current $j\in L^2\mathcal{V}_0(\Sigma)$ can then be expressed according to the Hodge-decomposition theorem in the form $j=\nabla f\times \mathcal{N}+\alpha\gamma_t\times \mathcal{N}+\beta \gamma_p\times \mathcal{N}$, where $\gamma_p,\gamma_t$ as usual form a basis of $\mathcal{H}(\Sigma)$ induced by, but independent (except for the orientation) of the specific choice of, a poloidal and toroidal curve within $\Sigma$. Further, we note that the restriction of any $\Gamma\in \mathcal{H}_N(\Omega)$ to $\Sigma$ gives rise to a closed vector-field because $\operatorname{curl}(\Gamma)\parallel \Sigma$, i.e. we can write $\Gamma|_{\Sigma}=\nabla_{\Sigma}\kappa+\gamma$ for a suitable function $\kappa$ and $\gamma\in \mathcal{H}(\Sigma)$. Even more, $\gamma\neq 0$, because otherwise we must have $\Gamma=0$ throughout $\Omega$ which can be seen by an integration by parts and the fact that each $\Gamma\in \mathcal{H}_N(\Omega)$ admits a vector potential. Further, $\int_{\sigma_t}\Gamma\neq 0$ assuming that the toroidal loop within $P$ defines also a toroidal loop within $\Omega$, see again \Cref{figure1}. Finally, $\int_{\sigma_p}\Gamma=0$ because in our applications we assume that we can choose $\sigma_p$ such that it bounds a disc within $\Omega$ and hence we may apply Stokes' theorem. We conclude that $\gamma$ must be a multiple of $\gamma_t$ and upon scaling $\Gamma$ appropriately we may achieve that $\Gamma=\nabla_{\Sigma}\kappa+\gamma_t$. We can therefore equivalently express $j=\nabla \tilde{f}\times \mathcal{N}+\alpha \Gamma\times \mathcal{N}+\beta\gamma_p\times \mathcal{N}$ where $\alpha$ and $\beta$ are as in the previous expression for $j$ and $\tilde{f}$ is a possibly modified function. We observe that on the one hand $\overline{Q}(j)=\frac{\beta}{|\Sigma|}\int_{\Sigma}\gamma_t\cdot (\gamma_p\times \mathcal{N})d\sigma$. On the other hand we observe that $\gamma_p\times \mathcal{N}\in \mathcal{H}(\Sigma)$ and that $\gamma_p\times \mathcal{N}$ is linearly independent of $\gamma_p$ so that $\int_{\Sigma}\gamma_t\cdot (\gamma_p\times \mathcal{N})d\sigma\neq 0$. Thus, $\overline{Q}(j)=0\Leftrightarrow \beta=0$. It is discussed in \cite[Section 4.1.3 Lemma 10]{PRS22} that $\alpha$ and $\beta$ correspond to the poloidal and toroidal flux of $j$ respectively and that setting the toroidal flux to zero should lead to more poloidal field lines. As we can see, the reasoning in \cite{PRS22} is consistent with our reasoning involving the quantity $\overline{Q}$ defined in (\ref{2E12}). Further, it is argued \cite[Section 4.1.2 \& 4.1.3]{PRS22} that one should set $\alpha=I_p$ to obtain good approximations. This can be justified as follows. One can show that the $L^2(\Omega)$-orthogonal projection of $\operatorname{BS}_{\Sigma}(j)$ onto the space $\mathcal{H}_N(\Omega)$ is given by $\alpha\Gamma$ whenever $j=\nabla \tilde{f}\times \mathcal{N}+\alpha \Gamma\times \mathcal{N}+\beta\gamma_p\times \mathcal{N}$ and we can find a toroidal loop within $\Sigma$ which bounds a $C^{1,1}$-surface outside of $\Omega$. We therefore find $\operatorname{BS}_{\Sigma}(j)=\nabla \psi+\alpha\Gamma$ for a suitable $\nabla\psi\in \mathcal{H}_{\operatorname{ex}}(\Omega)$. We observe that $\nabla\psi|_P\in \mathcal{H}_{\operatorname{ex}}(P)$ and that the restriction $\Gamma|_{P}$ will decompose further according to the Hodge-decomposition theorem into a $\Gamma|_P=\widetilde{\Gamma}+\nabla \tilde{\psi}$ for suitable $\nabla \tilde{\psi}\in \mathcal{H}_{\operatorname{ex}}(P)$ and $\widetilde{\Gamma}\in \mathcal{H}_N(P)$. Further, the space $\mathcal{H}_N(P)$ is $1$-dimensional and so we can fix some element $\hat{\Gamma}$ in it uniquely determined by $\int_{\sigma_t}\hat{\Gamma}=1$. We can then express $\widetilde{\Gamma}=\left(\int_{\sigma_t}\widetilde{\Gamma}\right)\hat{\Gamma}=\left(\int_{\sigma_t}\Gamma\right)\hat{\Gamma}=\hat{\Gamma}$ where we used that $\widetilde{\Gamma}$ and $\Gamma$ differ only by a gradient field, that $\sigma_t$ is by assumption also toroidal within $\Omega$ and hence homotopic to a toroidal loop on $\Sigma$ and that $\Gamma$ integrates to $1$ by our chosen scaling along any such loop. Letting $\Gamma_T\in \mathcal{H}_N(P)$ be the projection of $B_T$ onto $\mathcal{H}_N(P)$ we have $\Gamma_T=\left(\int_{\sigma_t}\Gamma_T\right)\hat{\Gamma}=I_p\hat{\Gamma}$ because $\Gamma_T$ and $B_T$ differ only by a gradient field. We conclude that
\begin{gather}
	\nonumber
	\|\operatorname{BS}_{\Sigma}(j)-B_T\|^2_{L^2(P)}\geq (\alpha-I_p)^2\|\hat{\Gamma}\|^2_{L^2(P)}
\end{gather}
and therefore we must pick $\alpha=I_p$ to be able to obtain a good approximation.
In order to obtain simpler coil designs and to ease computations the following modified minimisation procedure was numerically implemented in \cite[Section 4]{PRS22}
\begin{gather}
	\label{2E14}
	\min_{f\in H^1(\Omega)}\left(\|\mathcal{N}\cdot \operatorname{BS}_{\Sigma}(\nabla f\times \mathcal{N}+I_p\gamma_t\times \mathcal{N})-\mathcal{N}\cdot B_T\|^2_{L^2(\partial P)}+\lambda \|\nabla f\times \mathcal{N}+I_p\gamma_t\times \mathcal{N}\|^2_{L^2(\Sigma)}\right).
\end{gather}
We note that the $L^2(\partial P)$-norm of the normal traces dominates the $L^2(P)$-norm of the underlying vector fields once the poloidal current is fixed, \cite[Lemma 11]{PRS22}. This is not true the other way around, see \Cref{AppC}. From a mathematical perspective it is more natural to replace the $\|\cdot\|_{L^2(\partial P)}$-norm in (\ref{2E14}) by the $W^{-\frac{1}{2},2}(\partial P)$-norm since this turns out to be a norm which is equivalent to the $\|\cdot\|_{L^2(P)}$-norm, c.f. \Cref{AppC}. Therefore one can adapt the reasoning of \cite[Section 4]{G24} in order to show that once again arbitrary precision may be achieved if the $L^2(\partial P)$-norm in (\ref{2E14}) is replaced by the $W^{-\frac{1}{2},2}(\partial\Omega)$-norm. In addition, the normal trace $\mathcal{N}\cdot B_T$ will in general only be an element of $W^{-\frac{1}{2},2}(\partial P)$ if we allow arbitrary $B_T\in L^2\mathcal{H}(P)$ as target fields. From the point of view of applications, one should however expect to face more regular target fields $B_T$ which admit more regular traces.

If we compare the minimisation procedure (\ref{2E14}) with our proposed procedure \Cref{2T6} \& \Cref{2T9} we see that in both cases we obtain currents satisfying $\overline{Q}(j)=0$ and so the simplicity of the shape is incorporated in both procedures. The main difference between these two approaches from a theoretical point of view is that there is no a-priori guarantee that we can achieve an arbitrary small error as $\lambda\searrow 0$ in (\ref{2E14}) while (theoretically) arbitrary precision may be achieved in the algorithm \Cref{2T6} \& \Cref{2T9}. Note also that recently, c.f. \cite[Theorem 2.35]{G24SurfaceHelicityArXiv}, it has been shown that if we add the additional constraint $\overline{Q}(j)=0$ in (\ref{2E13}) we obtain a corresponding sequence of minimisers $j^0_{\lambda}$ satisfying $\overline{Q}(j^0_{\lambda})=0$ and such that $\|\operatorname{BS}_{\Sigma}(j^0_{\lambda})-B_T\|_{L^2(P)}\rightarrow0$ as $\lambda\searrow 0$.
\newline
\newline
\underline{\textit{Comparison from a computational point of view:}} In both minimisation problems (\ref{2E13}) and (\ref{2E14}) we are performing integrations over $\Sigma$ and either over a volume $P$ or a surface $\partial P$. We note that the integration over $\Sigma$ does not require computing singular integral kernels, whereas the integrations over $P$ and $\partial P$ require the computation of $\operatorname{BS}_{\Sigma}(j)$ which involves a kernel of the form $\frac{x-y}{|x-y|^3}$ with $x\in P\cup \partial P$ and $y\in \Sigma$. These computations can be costly but we note also that $P$ as well as $\partial P$ have a positive distance to $\Sigma$ so that strictly speaking the kernels do not become singular. The downside of both procedures is that no convergence rate is a priori known and so it is not known how small $\lambda$ needs to be chosen to achieve good results, but see \cite[Section 4]{PRS22} for some numerical results regarding (\ref{2E14}).

In the first step, \Cref{2T6}, of the proposed algorithm in the present work we do not need to compute any integrals involving the Biot-Savart operator but instead need to solve boundary value problems. This appears to be computationally simpler, even though for this part of the algorithm there is also no a priori known convergence rate and therefore it is not clear which of the methods leads faster to good approximations. We want to point out that while the formulation in \Cref{2T6} proposes to work on a volume $P$, we could similarly as is done in the reformulation (\ref{2E14}) exchange the volume integral by a surface integral, c.f. \Cref{AppC}, and identify the coefficient $\alpha_0$ with the toroidal circulation of $B_T$. The second part of the algorithm, \Cref{2T9}, includes the computation of singular integral kernels. In contrast to the situation in (\ref{2E13}) and (\ref{2E14}) the kernels in \Cref{2T9} really become singular since both $x$ and $y$ in this scenario run through $\Sigma$. It is therefore more difficult to handle these integrals. However, the main feature of the second part of the proposed algorithm is that we have an exponential a priori convergence rate so that one expects that only a few iterations should lead to reasonable results. In conclusion, the approaches in \cite{PRS22} as well as the algorithm proposed here have their advantages and disadvantages.
\subsubsection{Kernel reconstruction algorithm}
Before we formulate our algorithm we first recall that $\mathcal{H}_N(\Omega)$ denotes the space of square integrable fields which are div-free, curl-free within $\Omega$ and tangent to its boundary, (\ref{2E3}), which is always finite dimensional whenever $\Omega$ is a bounded $C^{1,1}$-domain.
Furthermore we define the volume Biot-Savart operator of a given domain $\Omega$
\begin{gather}
	\label{2E15}
	\operatorname{BS}_{\Omega}:L^2\mathcal{V}(\Omega)\rightarrow L^2\mathcal{V}(\Omega)\text{, }B\mapsto \left(x\mapsto\frac{1}{4\pi}\int_{\Omega}B(y)\times\frac{x-y}{|x-y|^3}d^3y\right).
\end{gather}
Given some $\Gamma\in \mathcal{H}_N(\Omega)$ it is known, \cite[Lemma A.1]{G24} that $\Gamma$ is of class $W^{1,p}(\Omega)$ for all $1\leq p<\infty$ and of class $C^{0,\alpha}(\overline{\Omega})$ for all $0<\alpha<1$ and that $\operatorname{BS}_{\Omega}(\Gamma)\in C^{1,\alpha}(\overline{\Omega})$ for all $0<\alpha<1$, see step 1 of the proof of \cite[Proposition 5.8]{G24}. In particular, the traces $\operatorname{BS}_{\Omega}(\Gamma)\cdot \mathcal{N}$ and $\operatorname{BS}_{\Omega}(\Gamma)\times \mathcal{N}$ exist. To simplify computations we further note that we have the following identity, see the proof of \cite[Proposition 3.6]{G24}
\begin{gather}
	\label{2E16}
	\operatorname{BS}_{\Omega}(B)(x)=\frac{1}{4\pi}\int_{\partial\Omega}\frac{B(y)\times \mathcal{N}(y)}{|x-y|}d\sigma(y)+\frac{1}{4\pi}\int_{\Omega}\frac{\operatorname{curl}(B)(y)}{|x-y|}d^3y\text{ for all }B\in W^{1,q}(\Omega,\mathbb{R}^3)\text{, }q>3.
\end{gather}
Consequently
\begin{gather}
	\label{2E17}
	\operatorname{BS}_{\Omega}(\Gamma)(x)=\frac{1}{4\pi}\int_{\partial\Omega}\frac{\Gamma(y)\times \mathcal{N}(y)}{|x-y|}d\sigma(y).
\end{gather}
Further, we observe that $\operatorname{div}(\operatorname{BS}_{\Omega}(\Gamma))=0$ and $\operatorname{curl}(\operatorname{BS}_{\Omega}(\Gamma))=\Gamma$ because $\Gamma$ is div-free and tangent to the boundary of $\Omega$, c.f. \cite[Theorem A]{CDG01}. Consequently $\Delta \operatorname{BS}_{\Omega}(\Gamma)=0$ where $\Delta$ denotes the vector Laplacian and this equation may be understood in the classical sense since it follows from standard interior elliptic regularity theory that $\operatorname{BS}_{\Omega}(\Gamma)$ is in fact analytic within $\Omega$. We conclude that $\operatorname{BS}_{\Omega}(\Gamma)$ is the unique solution to the following boundary value problem
\begin{gather}
	\label{2E18}
	\Delta A=0\text{ in }\Omega\text{ and }A|_{\partial\Omega}=\frac{1}{4\pi}\int_{\partial\Omega}\frac{\Gamma(y)\times \mathcal{N}(y)}{|x-y|}d\sigma(y)\text{, }x\in \partial\Omega.
\end{gather}
We note that the original definition (\ref{2E15}) requires to compute singular integral kernels in a volume to determine $\operatorname{BS}_{\Omega}(\Gamma)$, while the new formulation (\ref{2E18}) only requires to do so on a surface and solving a BVP which seems computationally easier.
Before we come to the kernel reconstruction algorithm we provide an exact formula for elements of the kernel.
\begin{thm}[Exact kernel elements]
	\label{2T11}
	Let $\Omega\subset\mathbb{R}^3$ be a bounded $C^{1,1}$-domain with possibly disconnected boundary. Fix a basis $\Gamma_1,\dots,\Gamma_n\in \mathcal{H}_N(\Omega)$ of $\mathcal{H}_N(\Omega)$. Then the following boundary value problems admit unique solutions $g_i\in H^1(\Omega)$
	\begin{gather}
		\nonumber
		\Delta g_i=0\text{ in }\Omega\text{, }\mathcal{N}\cdot \nabla g_i=\left(\frac{\operatorname{Id}}{2}+w^{\operatorname{Tr}}_{\Omega}\right)^{-1}\left(\left(\frac{\operatorname{Id}}{2}-w^{\operatorname{Tr}}_{\Omega}\right)(\operatorname{BS}_{\Omega}(\Gamma_i)\cdot \mathcal{N})\right)\text{ on }\partial\Omega\text{ and }\int_{\Omega}g_id^3x=0
	\end{gather}
	where $w^{\operatorname{Tr}}_{\Omega}$ denotes the transpose of the double layer potential, (\ref{2E7}) and $j_i:=\operatorname{BS}_{\Omega}(\Gamma_i)\times \mathcal{N}+\nabla g_i\times \mathcal{N}$, $i=1,\dots,n$ forms a basis of $\operatorname{Ker}(\operatorname{BS}_{\partial\Omega})$.
\end{thm}
The idea of the kernel reconstruction algorithm comes once again from expressing $\left(\frac{\operatorname{Id}}{2}+w^{\operatorname{Tr}}_{\Omega}\right)^{-1}$ as a Neumann series.
\begin{thm}[Kernel reconstruction algorithm]
	\label{2T12}
	Let $\Omega\subset\mathbb{R}^3$ be a bounded $C^{1,1}$-domain with possibly disconnected boundary and let $\Gamma\in \mathcal{H}_N(\Omega)\setminus \{0\}$. Then for every $n\in \mathbb{N}_0$ there is a unique solution $f_n\in H^1(\Omega)$ of the following boundary value problem
	\begin{gather}
		\label{E219}
		\Delta f_n=0\text{ in }\Omega\text{, }\mathcal{N}\cdot \nabla f=\sum_{k=1}^{n}\left(\frac{\operatorname{Id}}{2}-w^{\operatorname{Tr}}_{\Omega}\right)^k\left(\operatorname{BS}_{\Omega}(\Gamma)\cdot \mathcal{N}\right)\text{ on }\partial\Omega\text{ and }\int_{\Omega}f_nd^3x=0
	\end{gather}
	where $w^{\operatorname{Tr}}_{\Omega}$ denotes the transpose of the double layer potential, (\ref{2E7}). Further, there exist $0<\lambda(\partial\Omega)<1$, $0<c(\partial\Omega),\tilde{c}(\partial\Omega)<\infty$ independent of the chosen $\Gamma$, and some $j_0\in \operatorname{Ker}(\operatorname{BS}_{\partial\Omega})\setminus \{0\}$, which depends on $\Gamma$, such that
	\begin{gather}
		\label{2E20}
		\|j_0-j_n\|_{W^{-\frac{1}{2},2}(\partial\Omega)}\leq \frac{c\lambda^{n+1}}{1-\lambda}\|\Gamma\|_{L^2(\Omega)}
		\\
		\label{EXTRA21}
		\|\operatorname{BS}_{\partial\Omega}(j_n)\|_{L^2(\Omega)}\leq \frac{\tilde{c}\lambda^{n+1}}{1-\lambda}\|\Gamma\|_{L^2(\Omega)}
	\end{gather}
	where we set $j_n:=\operatorname{BS}_{\Omega}(\Gamma)\times \mathcal{N}+\nabla f_n\times \mathcal{N}\in W^{-\frac{1}{2},2}\mathcal{V}_0(\partial\Omega)$.
\end{thm}
The kernel reconstruction procedure can also be formulated as an exterior boundary value problem.
\begin{thm}[Kernel Elements and exterior BVP]
	\label{2T13}
	Let $\Omega\subset\mathbb{R}^3$ be a bounded $C^{1,1}$-domain with possibly disconnected boundary. Fix a basis $\Gamma_1,\dots,\Gamma_n\in \mathcal{H}_N(\Omega)$ of $\mathcal{H}_N(\Omega)$. Then the following exterior boundary value problems admit unique solutions $g_i\in \dot{H}^1(\overline{\Omega}^c)\equiv \{f\in L^6(\overline{\Omega}^c)\mid \nabla f\in L^2(\overline{\Omega}^c)\}$
	\begin{gather}
		\nonumber
		\Delta g_i(x)=0\text{ in }\mathbb{R}^3\setminus \overline{\Omega}\text{, }\mathcal{N}\cdot \nabla g_i=-\mathcal{N}\cdot \operatorname{BS}_{\Omega}(\Gamma_i)\text{ on }\partial\Omega\text{, }g_i(x)\rightarrow 0\text{ as }x\rightarrow\infty\text{, }\int_{\Omega_k}g_id^3x=0
	\end{gather}
	with $1\leq k\leq \#\partial\Omega-1$ (and the last condition is empty if $\partial\Omega$ is connected), where the $\Omega_k$ are the (connected) finite volumes enclosed by the connected components $\partial\Omega_k$ of $\partial\Omega$ which satisfy $\Omega_k\cap \Omega=\emptyset$. Further, $j_i:=\operatorname{BS}_{\Omega}(\Gamma_i)\times \mathcal{N}+\nabla g_i\times \mathcal{N}$, $i=1,\dots,n$ forms a basis of $\operatorname{Ker}(\operatorname{BS}_{\partial\Omega})$.
\end{thm}
\textbf{\underline{Comparison to previous kernel reconstruction algorithms:}}
\newline
\newline
In the recent work \cite[Section 6]{G24} the following kernel reconstruction algorithm has been proposed: Recall first that $\mathcal{H}_{\operatorname{ex}}(\Omega)\subset L^2\mathcal{H}(\Omega)$ denotes the subspace consisting of harmonic gradient fields, \ref{2E3}. Then the following operator was introduced
\begin{gather}
	\label{2E21}
	S:\mathcal{H}_{\operatorname{ex}}(\Omega)\rightarrow\mathcal{H}_{\operatorname{ex}}(\Omega)\text{, }\nabla f\mapsto \left(x\mapsto \frac{\nabla_x}{4\pi}\int_{\partial\Omega}\frac{\operatorname{BS}_{\Omega}(\Gamma)(y)\cdot \mathcal{N}(y)}{|x-y|}d\sigma(y)+\frac{\nabla_x}{4\pi}\int_{\Omega}\nabla f(y)\cdot \frac{x-y}{|x-y|^3}d^3y\right)
\end{gather}
and it was shown in \cite[Theorem 6.5]{G24} that if we define the recursive sequence $X_0:=0$, $X_{n+1}:=S(X_n)$ and $j_n:=\mathcal{N}\times \operatorname{BS}_{\Omega}(\Gamma)+\mathcal{N}\times X_n\in W^{-\frac{1}{2},2}\mathcal{V}_0(\partial\Omega)$ then the $j_n$ converge weakly to a non-trivial element of the kernel provided $\Gamma\in \mathcal{H}_N(\Omega)\setminus \{0\}$ and even more, that if we start with a basis of $\mathcal{H}_N(\Omega)$ the corresponding weak limits of the $(j_n)_n$ will form a basis of $\operatorname{Ker}(\operatorname{BS}_{\partial\Omega})$.

As we shall see, the algorithm proposed in \Cref{2T12} is in fact an equivalent reformulation of the algorithm proposed in \cite[Theorem 6.5]{G24} and the results established in the present manuscript in fact imply that the convergence in \cite[Theorem 6.5]{G24} is not only weakly in $W^{-\frac{1}{2},2}$ but in fact the sequence of $(j_n)_n$ converges in the strong $W^{-\frac{1}{2},2}(\partial\Omega)$ topology exponentially fast to a non-trivial element of the kernel of the Biot-Savart operator.

The advantage of the formulation \Cref{2T12} is that it only requires computing singular integrals on a surface while the approach via the operator $S$ in (\ref{2E21}) requires the computation of singular integrals over volumes.

Similarly, the characterisation in \Cref{2T13} only requires solving an exterior boundary value problem as well as the computation of $\operatorname{BS}_{\Omega}(\Gamma)$ for given $\Gamma\in \mathcal{H}_N(\Omega)$ which by means of (\ref{2E18}) can be reduced to computing a singular boundary integral and an interior BVP.
\section{Image of the Biot-Savart operator}
\subsection{Preliminary results}
We define for a given $f\in W^{\frac{1}{2},2}(\partial\Omega)$ of a domain $\Omega\subset\mathbb{R}^3$
\begin{gather}
	\label{3E1}
	W_{\Omega}(f)(x):=\frac{1}{4\pi}\int_{\partial\Omega}f(y)\left(\mathcal{N}(y)\cdot \frac{x-y}{|x-y|^3}\right)d\sigma(y)\text{ for }x\in \Omega.
\end{gather}
In the following we mean by $\psi\in H^1_{\operatorname{loc}}(\overline{\Omega})$ that for every $x\in \overline{\Omega}$ there is some open $x\in U\subset\mathbb{R}^3$ such that $\psi\in H^1(U\cap \Omega)$. In particular, if $\Omega$ is bounded, then $H^1_{\operatorname{loc}}(\overline{\Omega})=H^1(\Omega)$.
\begin{lem}
	\label{3L1}
	Let $\Omega\subset \mathbb{R}^3$ be a (not necessarily bounded) $C^{1,1}$-domain with compact boundary $\partial\Omega$. Then for every $f\in W^{\frac{1}{2},2}(\partial\Omega)$ we have $W_{\Omega}(f)\in H^1_{\operatorname{loc}}(\overline{\Omega})$ and we have the following jump formula
	\begin{gather}
		\nonumber
		\operatorname{Tr}(W_{\Omega}(f))(x)=-\frac{f(x)}{2}+w_{\Omega}(f)(x)\text{ for }\mathcal{H}^2\text{-a.e. }x\in\partial\Omega
	\end{gather}
	where $w_{\Omega}$ is defined in the first equation of (\ref{2E7}) and $\mathcal{H}^2$ denotes the standard surface measure on $\partial\Omega$.
\end{lem}
\begin{proof}[Proof of \Cref{3L1}]
	\underline{Step 1, $W_{\Omega}(f)\in H^1_{\operatorname{loc}}(\overline{\Omega})$:} We prove the stronger assertion that $W_{\Omega}$ is a linear, bounded operator from $W^{\frac{1}{2},2}(\partial\Omega)$ into $\dot{H}^1(\Omega):=\{h\in L^6(\Omega)\mid \nabla h\in L^2(\Omega)\}$. To see this we may fix an extension $\widetilde{f}\in H^1(\Omega)$ whose $H^1(\Omega)$-norm is bounded above by the $W^{\frac{1}{2},2}(\partial\Omega)$-norm of $f$ (modulo a constant which is independent of $f$) and whose support is contained in some bounded subset $U\subset \overline{\Omega}$ (again independent of $f$), see \cite[Proposition 3.31]{DD12}. We can then apply Gauss' formula and use the fact that $\operatorname{div}_y\left(\frac{x-y}{|x-y|^3}\right)=-4\pi\delta(x-y)$ where $\delta$ denotes the Dirac-delta to arrive at
	\begin{gather}
		\nonumber
		W_{\Omega}(f)(x)=-\widetilde{f}(x)+\frac{1}{4\pi}\int_{\Omega}\nabla \widetilde{f}(y)\cdot \frac{x-y}{|x-y|^3}d^3y\text{ for all }x\in \Omega.
	\end{gather}
	Now on the one hand $\widetilde{f}\in H^1(\Omega)$. On the other hand, according to the Hardy-Littlewood-Sobolev inequality, \cite[Chapter V]{S70}, we can control the $L^6(\Omega)$-norm of $\int_{\Omega}\nabla \widetilde{f}(y)\cdot \frac{x-y}{|x-y|^3}d^3y$ by means of the $L^2(\Omega)$-norm of $\nabla \widetilde{f}$. Further, we observe that $\int_{\Omega}\nabla \widetilde{f}(y)\cdot \frac{x-y}{|x-y|^3}d^3y=-\partial_{x^k}\int_{\Omega}\frac{\partial_{y^k}\widetilde{f}(y)}{|x-y|}d^3y$ and that $\int_{\Omega}\frac{\partial_{y^k}\widetilde{f}(y)}{|x-y|}d^3y$ corresponds to the Newton potential of the $L^2(\Omega)$-function $\partial_{x^k}\widetilde{f}$. We may replace the integral over $\Omega$ in $\int_{\Omega}\frac{\partial_{x^k}\widetilde{f}(y)}{|x-y|}d^3y$ by an integral over the bounded set $U$ since $\widetilde{f}$ is supported in $U$. It then follows from the regularity of the Newton potential, c.f. \cite[Theorem 9.9]{GT01}, that all second order derivatives of the Newton potential exist, are of class $L^2(\mathbb{R}^3)$ and that their $L^2(\mathbb{R}^3)$ and consequently their $L^2(\Omega)$-norm may be controlled by means of the $L^2(U)$-norm (and hence of the $L^2(\Omega)$-norm) of $\nabla \widetilde{f}$. We conclude
	\begin{gather}
		\label{3E2}
		\|W_{\Omega}(f)\|_{L^6(\Omega)}+\|\nabla W_{\Omega}(f)\|_{L^2(\Omega)}\leq c(\Omega)\|\nabla \widetilde{f}\|_{L^2(\Omega)}\leq C(\Omega)\|f\|_{W^{\frac{1}{2},2}(\partial\Omega)}
	\end{gather}
	for suitable constants $0<c,C<\infty$ independent of $f$.
	\newline
	\newline
	\underline{Step 2:} Here we prove that if $f\in H^1_{\operatorname{loc}}(\overline{\Omega})$ solves weakly $\Delta f=0$ in $\Omega$, then
	\begin{gather}
		\label{3E3}
		\lim_{r\searrow 0}f(x-r\mathcal{N}(x))=\operatorname{Tr}(f)(x)\text{ for }\mathcal{H}^2\text{-a.e. }x\in \partial\Omega
	\end{gather}
	where $\mathcal{N}(x)$ denotes the outward unit normal at $x$.
	We start with the following fact, c.f. \cite[Theorem 5.7]{EG15},
	\begin{gather}
		\label{3E4}
		\lim_{r\searrow 0}\dashint_{B_r(x)\cap \Omega}|f(y)-\operatorname{Tr}(f)(x)|d^3y=0\text{ for }\mathcal{H}^2\text{-a.e. }x\in \partial\Omega.
	\end{gather}
	It then follows further \cite[Theorem 2.6]{Dal18} that there is some $r_0>0$ such that for all $0<r\leq r_0$ we have $B_r(x-r\mathcal{N}(x))\subset \Omega$. In particular, if we fix $0<\lambda<1$ then $\overline{B_{\lambda r}(x-r\mathcal{N}(x))}\subset \Omega$ and we have the inclusion $B_{\lambda r}(x-r\mathcal{N}(x))\subset B_{2r}(x)$ because for every $z\in B_r(x-r\mathcal{N}(x))$ we have $|z-(x-r\mathcal{N}(x))|\leq r$ and on the other hand $|z-(x-r\mathcal{N}(x))|\geq |z-x|-r$ so that $|z-x|\leq 2r$ for any $0<\lambda<1$ and $z\in B_{\lambda r}(x-r\mathcal{N}(x))\subset B_r(x-r\mathcal{N}(x))$. Since $B_{\lambda r}(x-r\mathcal{N}(x))\subset \Omega$ we conclude $B_{\lambda r}(x-r\mathcal{N}(x))\subset \Omega\cap B_{2r}(x)$ and can compute
	\begin{gather}
		\nonumber
		\dashint_{B_{\lambda r}(x-r\mathcal{N}(x))}|f(y)-\operatorname{Tr}(f)(x)|d^3y=\frac{\int_{B_{\lambda r}(x-r\mathcal{N}(x))}|f(y)-\operatorname{Tr}(f)(x)|d^3y}{|B_{\lambda r}(x)|}
		\\
		\nonumber
		=\frac{|B_{2r}(x)\cap \Omega|}{|B_{\lambda r}(x)|}\frac{\int_{B_{\lambda r}(x-r\mathcal{N}(x))}|f(y)-\operatorname{Tr}(f)(x)|d^3y}{|B_{2r}(x)\cap \Omega|}\leq \frac{|B_{2r}(x)\cap \Omega|}{|B_{\lambda r}(x)|}\frac{\int_{B_{ r}(x-r\mathcal{N}(x))}|f(y)-\operatorname{Tr}(f)(x)|d^3y}{|B_{2r}(x)\cap \Omega|}
		\\
		\nonumber
		\leq \frac{|B_{2r}(x)\cap \Omega|}{|B_{\lambda r}(x)|}\dashint_{B_{2r}(x)\cap \Omega}|f(y)-\operatorname{Tr}(f)(x)|d^3y\leq \frac{|B_{2r}(x)|}{|B_{\lambda r}(x)|}\dashint_{B_{2r}(x)\cap \Omega}|f(y)-\operatorname{Tr}(f)(x)|d^3y
		\\
		\nonumber
		\leq \frac{8}{\lambda^3}\dashint_{B_{2r}(x)\cap \Omega}|f(y)-\operatorname{Tr}(f)(x)|d^3y\rightarrow 0\text{ as }r\searrow 0
	\end{gather}
	by means of (\ref{3E4}) for any fixed $0<\lambda<1$ and for $\mathcal{H}^2$-a.e. $x\in \partial\Omega$. In particular we deduce
	\begin{gather}
		\label{3E5}
		\dashint_{B_{\frac{r}{2}(x-r\mathcal{N}(x))}}f(y)d^3y\rightarrow \operatorname{Tr}(f)(x)\text{ for }\mathcal{H}^2\text{-a.e. }x\in \partial\Omega.
	\end{gather}
	However, since we assume $\Delta f=0$ in the weak sense, it follows from standard interior elliptic regularity results that $f$ is analytic in $\Omega$ and hence harmonic in the classical sense. Then, since $\overline{B_{\frac{r}{2}}(x-r\mathcal{N}(x))}\subset \Omega$, we deduce from the mean value property that $\dashint_{B_{\frac{r}{2}(x-r\mathcal{N}(x))}}f(y)d^3y=f(x-r\mathcal{N}(x))$ for all $0<r\leq r_0$ and consequently (\ref{3E5}) implies (\ref{3E3}) as desired.
	\newline
	\newline
	\underline{Step 3:} In this last step we deduce the claimed identity in the lemma. We conclude first from \cite[Theorem XIV]{Cial95} that for any fixed $\phi\in L^q(\partial\Omega)$ with $q>2$ we have
	\begin{gather}
		\nonumber
		\lim_{r\searrow 0}W_{\Omega}(\phi)(x-r\mathcal{N}(x))=-\frac{\phi(x)}{2}+w_{\Omega}(\phi)(x)\text{ for }\mathcal{H}^2\text{-a.e. }x\in \partial\Omega.
	\end{gather}
	We have the embedding $W^{\frac{1}{2},2}(\partial\Omega)\hookrightarrow L^4(\partial\Omega)$, \cite[Theorem 3.81]{DD12}, and accordingly we conclude
	\begin{gather}
		\label{3E6}
		\lim_{r\searrow 0}W_{\Omega}(f)(x-r\mathcal{N}(x))=-\frac{f(x)}{2}+w_{\Omega}(f)(x)\text{ for all }f\in W^{\frac{1}{2},2}(\partial\Omega)\text{ and }\mathcal{H}^2\text{-a.e. }x\in \partial\Omega.
	\end{gather}
	According to step 1 we have $W_{\Omega}(f)\in H^1_{\operatorname{loc}}(\overline{\Omega})$ and it is easy to verify that $\Delta W_{\Omega}(f)=0$ in the classical sense within $\Omega$ and so in particular in the weak sense, see also \cite[proposition 4.28]{RCM21}. So according to step 2 we find $\lim_{r\searrow 0}W_{\Omega}(f)(x-r\mathcal{N}(x))=\operatorname{Tr}(W_{\Omega}(f))(x)$ for $\mathcal{H}^2$-a.e. $x\in \partial\Omega$ and so the lemma follows now from the identity in (\ref{3E6}).
\end{proof}
\Cref{3L1} will allow us to prove the following important fact
\begin{cor}
	\label{3C2}
	Let $\Omega\subset\mathbb{R}^3$ be a bounded $C^{1,1}$-domain with possibly disconnected boundary. Then the operator
	\begin{gather}
		\nonumber
		T:\mathcal{H}_{\operatorname{ex}}(\Omega)\cap \mathcal{H}^{\perp_{L^2(\Omega)}}_D(\Omega)\rightarrow \mathcal{H}_{\operatorname{ex}}(\Omega)\cap \mathcal{H}^{\perp_{L^2(\Omega)}}_D(\Omega)\text{, }
		\nabla f\mapsto \left(x\mapsto \frac{\nabla_x}{4\pi}\int_{\Omega}\nabla f(y)\cdot \frac{x-y}{|x-y|^3}d^3y\right)
	\end{gather}
	is a well-defined linear contraction with respect to the $\|\cdot\|_{L^2(\Omega)}$-norm.
\end{cor}
During the course of the proof we will need the following simple fact
\begin{lem}
	\label{3L3}
	Let $\Omega\subset\mathbb{R}^3$ be a bounded $C^{1,1}$-domain with possibly disconnected boundary. Then there is some $c>0$ such that for all $f\in H^1(\Omega)$ which weakly solve $\Delta f=0$ in $\Omega$ we have
	\begin{gather}
		\|\nabla f\|_{L^2(\Omega)}\leq c\|\operatorname{Tr}(f)\|_{W^{\frac{1}{2},2}(\partial\Omega)}.
	\end{gather}
\end{lem}
\begin{proof}[Proof of \Cref{3L3}]
	Fix any $f\in H^1(\Omega)$ with $\Delta f=0$ in $\Omega$ and compute
	\begin{gather}
		\nonumber
		\|\nabla f\|^2_{L^2(\Omega)}=\int_{\partial\Omega}\operatorname{Tr}(f)(x)\left(\mathcal{N}\cdot \nabla f(x)\right)d\sigma(x)\leq \|\operatorname{Tr}(f)\|_{W^{\frac{1}{2},2}(\partial\Omega)}\|\mathcal{N}\cdot \nabla f\|_{W^{-\frac{1}{2},2}(\partial\Omega)}
		\\
		\nonumber
		\leq c\|\nabla f\|_{H(\operatorname{div},\Omega)}\|\operatorname{Tr}(f)\|_{W^{\frac{1}{2},2}(\partial\Omega)}=c\|\nabla f\|_{L^2(\Omega)}\|\operatorname{Tr}(f)\|_{W^{\frac{1}{2},2}(\partial\Omega)}
	\end{gather}
	where we used the continuity of the normal trace with respect to the $H(\operatorname{div},\Omega)$-norm.
\end{proof}
\begin{proof}[Proof of \Cref{3C2}]
	Linearity is clear. Further we observe that $T(\nabla f)=S(\nabla f)-S(0)$ with $S$ being defined in (\ref{2E21}) and that it was shown in \cite[Lemma 6.2]{G24} that $S$ is a continuous mapping from $\mathcal{H}_{\operatorname{ex}}(\Omega)$ into $\mathcal{H}_{\operatorname{ex}}(\Omega)$.
	
	We show now that if $\nabla f\in \mathcal{H}_{\operatorname{ex}}(\Omega)$ is $L^2$-orthogonal to $\mathcal{H}_D(\Omega)$, then so is $T(\nabla f)$. Recall first that, (\ref{2E2}), $\mathcal{H}_D(\Omega)$ is the space of fields $\nabla h$ with $\Delta h=0$ in $\Omega$, $h\in H^1(\Omega)$ and $\nabla h\times \mathcal{N}=0$ on $\partial\Omega$. We note that by an approximation argument we may suppose that $f\in W^{2,p}(\Omega)$ for every $1\leq p<\infty$ which allows us to justify all of the upcoming integral manipulations rigorously. We then make use of the fact that $\frac{x-y}{|x-y|^3}=\nabla_y\frac{1}{|x-y|}$ and thus
	\begin{gather}
		\nonumber
		\int_{\Omega}\nabla f (y)\cdot \frac{x-y}{|x-y|^3}d^3y=\int_{\Omega}\nabla f\cdot \nabla_y\frac{1}{|x-y|}d^3y=\int_{\partial\Omega}\frac{\mathcal{N}(y)\cdot \nabla f(y)}{|x-y|}d\sigma(y)
	\end{gather}
	where we used that $\Delta f=0$. Consequently
	\begin{gather}
		\label{3E7}
		T(\nabla f)(x)=\frac{1}{4\pi}\int_{\partial\Omega}(\mathcal{N}(y)\cdot \nabla f(y))\frac{y-x}{|x-y|^3}d\sigma(y).
	\end{gather}
	Now $\partial\Omega$ has finitely many connected components and a unique connected component $\partial\Omega_0$ such that the finite domain $\Omega_0$ enclosed by $\partial\Omega_0$ contains $\Omega$. The finite domains $\Omega_1,\dots,\Omega_n$ enclosed by the $\partial\Omega_i$, $i=1,\dots,n$, are disjoint to $\Omega$. We can then fix any $\nabla h\in \mathcal{H}_D(\Omega)$ and note first, \cite[Lemma A.2]{G24}, $\nabla h\in W^{1,p}(\Omega)$ for all $1\leq p<\infty$ and consequently $h\in W^{2,p}(\Omega)$ for all $1\leq p<\infty$. From this one easily concludes that the boundary condition $\nabla h\times \mathcal{N}$ implies $h|_{\partial\Omega_i}=c_i\in \mathbb{R}$. Further, since $h$ is unique only up to a constant, we may suppose that $h|_{\partial\Omega_0}=0$. With these preliminary considerations we compute
	\begin{gather}
		\nonumber
		\int_{\Omega}\nabla h(x)\cdot T(\nabla f)(x)d^3x=\frac{1}{4\pi}\int_{\partial\Omega}\left(\mathcal{N}(y)\cdot \nabla f(y)\right)\int_{\Omega}\frac{y-x}{|x-y|^3}\cdot \nabla h(x)d^3xd\sigma(y)
	\end{gather}
	and 
	\begin{gather}
		\nonumber
		\int_{\Omega}\frac{y-x}{|x-y|^3}\cdot \nabla h(x)d^3x=-4\pi h(y)+\int_{\partial\Omega}h(x)\mathcal{N}(x)\cdot \frac{y-x}{|x-y|^3}d\sigma(x)
		\\
		\nonumber
		=-4\pi h(y)+\sum_{i=1}^nc_i\int_{\partial\Omega_i}\mathcal{N}(x)\cdot \frac{y-x}{|y-x|^3}d\sigma(x)=-4\pi h(y)
	\end{gather}
	where we used that $\operatorname{div}_x\left(\frac{y-x}{|x-y|^3}\right)=\operatorname{div}_x\left(\nabla_x\frac{1}{|x-y|}\right)=-4\pi \delta(x-y)$ where $\delta(x-y)$ denotes the Dirac-delta and where we used that $h|_{\partial\Omega_0}=0$ and that $\Omega_i\cap\Omega=\emptyset$ for $i=1,\dots,n$. Combining these calculations yields
	\begin{gather}
		\nonumber
		\int_{\Omega}\nabla h(x)\cdot T(\nabla f)(x)d^3x=-\int_{\partial\Omega}h(y)\left(\mathcal{N}(y)\cdot \nabla f(y)\right)d\sigma(y)=-\int_{\Omega}\nabla h(y)\cdot \nabla f(y)d^3y
	\end{gather}
	and therefore $\nabla f\in \mathcal{H}^{\perp_{L^2(\Omega)}}_D(\Omega)$ implies $T(\nabla f)\in \mathcal{H}^{\perp_{L^2(\Omega)}}_D(\Omega)$.
	\newline
	\newline
	We are left with proving that $T$ is a contraction. By linearity of $T$ we need to prove that there exists some $0<\lambda<1$ such that
	\begin{gather}
		\label{3E9}
		\|T(\nabla f)\|_{L^2(\Omega)}\leq \lambda \|\nabla f\|_{L^2(\Omega)}.
	\end{gather}
	We recall the relationship $T(\nabla f)=S(\nabla f)-S(0)$ between the operator $T$ and the operator $S$, c.f. (\ref{2E21}). It has been shown in \cite[Lemma 6.2]{G24} that the operator $S$ is firmly non-expansive, i.e. $\|S(\nabla f)-S(\nabla \widetilde{f})\|^2_{L^2(\Omega)}\leq \langle \nabla f-\nabla \widetilde{f},S(\nabla f)-S(\nabla \widetilde{f})\rangle_{L^2(\Omega)}$ and consequently, setting $\nabla \widetilde{f}=0$, so is $T$. This shows that $T$ is a weak contraction, i.e. (\ref{3E9}) is satisfied with $\lambda=1$. Our goal now will be to exclude the possibility $\lambda=1$ by arguing by contradiction.
	\newline
	First we will equivalently reformulate the property of $T$ being a contraction. It follows from \cite[Equation (6.2)]{G24} and the relation between $S$ and $T$ that
	\begin{gather}
		\label{3E10}
		\|T(\nabla f)\|^2_{L^2(\Omega)}=\|\nabla f\|^2_{L^2(\Omega)}-\|\nabla f-T(\nabla f)\|^2_{L^2(\Omega)}-2\|T(\nabla f)\|^2_{L^2(\mathbb{R}^3\setminus \overline{\Omega})}
	\end{gather}
	where we observe that $T(\nabla f)$ is still well-defined on the complement of $\Omega$. We now observe that for $x\in \mathbb{R}^3\setminus \partial\Omega$
	\begin{gather}
		\nonumber
		\frac{1}{4\pi}\int_{\Omega}\nabla f(y)\cdot \frac{x-y}{|x-y|^3}d^3y=\frac{1}{4\pi}\int_{\partial\Omega}f(y)\left(\mathcal{N}(y)\cdot \frac{x-y}{|x-y|^3}\right)d\sigma(y)+\begin{cases}
			f(x) & \text{ if }x\in \Omega \\
			0 & \text{ if }x\in \mathbb{R}^3\setminus \overline{\Omega}
		\end{cases}
	\end{gather}
	where we used that $-\Delta_y\frac{1}{|x-y|}=4\pi \delta (x-y)$ with the Dirac-delta $\delta$. We recall the definition of $W_{\Omega}$, (\ref{3E1}), and note that we may then express (\ref{3E10}) equivalently as
	\begin{gather}
		\label{3E11}
		\|T(\nabla f)\|^2_{L^2(\Omega)}=\|\nabla f\|^2_{L^2(\Omega)}-\|\nabla W_{\Omega}(f)\|^2_{L^2(\Omega)}-2\|\nabla W_{\overline{\Omega}^c}(f)\|^2_{L^2(\overline{\Omega}^c)}.
	\end{gather}
	With this we find the following equivalent characterisation of the contraction property of $T$
	\begin{gather}
		\nonumber
		\|T(\nabla f)\|^2_{L^2(\Omega)}\leq \lambda ^2\|\nabla f\|^2_{L^2(\Omega)}\text{ for some }0<\lambda<1
		\\
		\nonumber
		\Leftrightarrow
		(1-\lambda^2)\|\nabla f\|^2_{L^2(\Omega)}\leq \|\nabla W_{\Omega}(f)\|^2_{L^2(\Omega)}+2\|\nabla W_{\overline{\Omega}^c}(f)\|^2_{L^2(\overline{\Omega}^c)}\text{ for a }0<\lambda<1
		\\
		\nonumber
		\Leftrightarrow
		\|\nabla f\|^2_{L^2(\Omega)}\leq \frac{\|\nabla W_{\Omega}(f)\|^2_{L^2(\Omega)}+2\|\nabla W_{\overline{\Omega}^c}(f)\|^2_{L^2(\overline{\Omega}^c)}}{1-\lambda^2}\text{ for a }0<\lambda<1
		\\
		\nonumber
		\Leftrightarrow
		\|\nabla f\|^2_{L^2(\Omega)}\leq c\left(\|\nabla W_{\Omega}(f)\|^2_{L^2(\Omega)}+2\|\nabla W_{\overline{\Omega}^c}(f)\|^2_{L^2(\overline{\Omega}^c)}\right)\text{ for some }1<c<\infty
		\\
		\label{3E12}
		\Leftrightarrow 	\|\nabla f\|^2_{L^2(\Omega)}\leq C\left(\|\nabla W_{\Omega}(f)\|^2_{L^2(\Omega)}+\|\nabla W_{\overline{\Omega}^c}(f)\|^2_{L^2(\overline{\Omega}^c)}\right)\text{ for some }0<C<\infty.
	\end{gather}
	We will use the equivalent reformulation (\ref{3E12}) to argue by contradiction. We suppose that there exists a sequence $(\nabla f_n)_n\subset \mathcal{H}_{\operatorname{ex}}(\Omega)\cap \mathcal{H}^{\perp_{L^2(\Omega)}}_D(\Omega)$ such that
	\begin{gather}
		\label{3E13}
		\|\nabla f_n\|^2_{L^2(\Omega)}\geq n\left(\|\nabla W_{\Omega}(f_n)\|^2_{L^2(\Omega)}+\|\nabla W_{\overline{\Omega}^c}(f_n)\|^2_{L^2(\overline{\Omega}^c)}\right)\text{ for all }n.
	\end{gather}
	We note that the $f_n$ are determined up to additive constants and so we may suppose that $\int_{\Omega}f_nd^3x=0$ for all $n$. In addition we may replace the $f_n$ by $\frac{f_n}{\|\operatorname{Tr}(f_n)\|_{W^{\frac{1}{2},2}(\partial\Omega)}}$ and observe that according to \Cref{3L3} for this choice of $f_n$ we will also have $\|\nabla f_n\|_{L^2(\Omega)}\leq c$ for some $c$ independent of $n$. Consequently we obtain a sequence of function $(f_n)_n\subset H^1(\Omega)$ with $\Delta f_n=0$ in $\Omega$ for all $n$ and
	\begin{gather}
		\label{3E14}
		\|\nabla W_{\Omega}(f_n)\|^2_{L^2(\Omega)}+\|\nabla W_{\overline{\Omega}^c}(f_n)\|^2_{L^2(\overline{\Omega}^c)}\leq \frac{c}{n}\text{, }\int_{\Omega}f_nd^3x=0\text{ and }\|\operatorname{Tr}(f_n)\|_{W^{\frac{1}{2},2}(\partial\Omega)}=1\text{ for all }n.
	\end{gather}
	By Poincar\'{e}'s inequality we see that also $\|f_n\|_{H^1(\Omega)}\leq c$ for some constant $c$ independent of $n$ (here we use the letter $c$ to denote a generic constant which may differ in distinct expressions). Consequently $f_n\rightharpoonup f$ weakly in $H^1(\Omega)$.
	As we have seen previously $T(\nabla f)=\nabla W_{\Omega}(f)+\nabla f$ and since $T$ is $L^2(\Omega)$-continuous and $\nabla f_n$ converges weakly to $\nabla f$ we conclude that $T(\nabla f_n)$ converges weakly in $L^2(\Omega)$ to $T(\nabla f)$ and consequently $\nabla W_{\Omega}(f_n)$ converges weakly to $\nabla W_{\Omega}(f)$ in $L^2(\Omega)$. Further, we have also seen that $T(\nabla f)=-\nabla W_{\overline{\Omega}^c}(f)$ on $\overline{\Omega}^c$ where the minus sign stems from the fact that the outward unit normal to $\overline{\Omega}^c$ equals the inward unit normal of $\Omega$ at any given point of the boundary. It follows further from (\ref{3E10}) and the $L^2(\Omega)$-boundedness of $T$ that if $T$ is viewed as a map $T:\mathcal{H}_{\operatorname{ex}}(\Omega)\rightarrow L^2(\overline{\Omega}^c,\mathbb{R}^3)$ then it is a well-defined, linear bounded operator. Hence, $\nabla W_{\overline{\Omega}^c}(f_n)$ converges weakly to $\nabla W_{\overline{\Omega}^c}(f)$ in $L^2(\overline{\Omega}^c)$. On the other hand (\ref{3E14}) tells us that $\nabla W_{\Omega}(f_n)$ and $\nabla W_{\overline{\Omega}^c}(f_n)$ converge strongly in $L^2$ to zero from which we infer that $\nabla W_{\Omega}(f)=0$ in $\Omega$ and $\nabla W_{\overline{\Omega}^c}(f)=0$ in $\overline{\Omega}^c$.
	We can now define the following linear, bounded operator, where the well-definedness and boundedness is a consequence of the regularity properties of the Newton potential \cite[Theorem 9.9]{GT01} and the Hardy-Littlewood-Sobolev inequality \cite[Chapter V]{S70}, see also the proof of \Cref{3L1} and \cite[Lemma 6.2]{G24},
	\begin{gather}
		\nonumber
		H:L^2(\Omega,\mathbb{R}^3)\rightarrow H^1(\mathbb{R}^3)\text{, }X\mapsto \left(x\mapsto\frac{1}{4\pi}\int_{\Omega}X(y)\cdot \frac{x-y}{|x-y|^3}d^3y\right).
	\end{gather}
	The main observation now is that $T(\nabla f)=\nabla H(\nabla f)$ on $\Omega$ as well as $\overline{\Omega}^c$ and that we have seen that $T(\nabla f)=0$ on $\overline{\Omega}^c$. This implies that $H(\nabla f)$ is locally constant on $\overline{\Omega}^c$ and hence $\operatorname{Tr}(H(\nabla f))$ is also locally constant. Since $H(\nabla f)\in H^1(\mathbb{R}^3)$, its trace when viewed as a function on $\overline{\Omega}^c$ and when viewed as a function on $\Omega$ coincide. From this we conclude that $\operatorname{Tr}_{\Omega}(H(\nabla f))$ is locally constant. Finally, we have also seen that $\nabla W_{\Omega}(f)=0$ in $\Omega$ which we can express equivalently as $T(\nabla f)=\nabla f$ in $\Omega$ and hence $f$ and $H(\nabla f)$ differ only by a constant which implies that $\operatorname{Tr}(f)$ is locally constant on $\partial\Omega$. From this we conclude $\nabla f\times \mathcal{N}=0$ and consequently $\nabla f\in \mathcal{H}_D(\Omega)$ and thus $\nabla f=0$ since $(\nabla f_n)_n\subset \mathcal{H}^{\perp_{L^2(\Omega)}}_D(\Omega)$. The weak $H^1(\Omega)$-convergence of the $f_n$ to $f$ also implies that $\int_{\Omega}fd^3x=0$ from which we conclude that $f=0$ in $\Omega$.
	
	Due to the continuity of the trace operator we find $\operatorname{Tr}(f_n)\rightharpoonup \operatorname{Tr}(f)=0$ weakly in $W^{\frac{1}{2},2}(\partial\Omega)$ and therefore (upon passing to a subsequence if necessary) $\operatorname{Tr}(f_n)\rightarrow \operatorname{Tr}(f)=0$ strongly in $L^1(\partial\Omega)$, \cite[Theorem 3.85]{DD12}. This allows us to estimate the average of $W_{\Omega}(f_n)$ as follows
	\begin{gather}
		\nonumber
		\left|\int_{\Omega}W_{\Omega}(f_n)(x)d^3x\right|=\frac{1}{4\pi}\left|\int_{\partial\Omega}f_n(y)\int_{\Omega}\mathcal{N}(y)\cdot \frac{x-y}{|x-y|^3}d^3xd\sigma(y)\right|
		\\
		\nonumber
		\leq \frac{1}{4\pi}\int_{\partial\Omega}|f_n(y)|\int_{\Omega}\frac{1}{|x-y|^2}d^3xd\sigma(y)\leq \frac{\sup_{y\in \partial\Omega}\int_{\Omega}\frac{1}{|x-y|^2}d^3x}{4\pi}\|f_n\|_{L^1(\partial\Omega)}=c\|\operatorname{Tr}(f_n)\|_{L^1(\partial\Omega)}\rightarrow 0.
	\end{gather}
	We can then make use of Poincar\'{e}'s inequality to estimate
	\begin{gather}
		\nonumber
		\|W_{\Omega}(f_n)\|_{L^2(\Omega)}\leq c\left(\|\nabla W_{\Omega}(f_n)\|_{L^2(\Omega)}+\|\operatorname{Tr}(f_n)\|_{L^1(\partial\Omega)}\right)\rightarrow 0
	\end{gather}
	where we used once more (\ref{3E14}). We conclude that $\|W_{\Omega}(f_n)\|_{H^1(\Omega)}\rightarrow 0$ and therefore we obtain $\|\operatorname{Tr}(W_{\Omega}(f_n))\|_{W^{\frac{1}{2},2}(\partial\Omega)}\rightarrow0$. \Cref{3L1} then allows us to conclude
	\begin{gather}
		\label{3E15}
		\lim_{n\rightarrow\infty}\left(w_{\Omega}(f_n)-\frac{\operatorname{Tr}(f_n)}{2}\right)=0\text{ strongly in }W^{\frac{1}{2},2}(\partial\Omega).
	\end{gather}
	The goal now is to establish a corresponding version of (\ref{3E15}) where $\Omega$ is replaced by $\overline{\Omega}^c$. The additional technical issue is that $\overline{\Omega}^c$ may be disconnected and that it contains an unbounded component. We label the boundary components of $\partial\Omega$ by $\partial\Omega_0$, $\partial\Omega_1,\dots,\partial\Omega_n$ and the corresponding finite volumes enclosed by these components by $\Omega_0,\Omega_1,\dots,\Omega_n$ respectively where we pick the components such that $\Omega_0$ is the unique finite volume which contains $\Omega$. We can therefore write $\overline{\Omega}^c=U\cup \bigcup_{i=1}^n\Omega_i$ with $U:=\overline{\Omega}^c_0$. We know already from (\ref{3E14}) that $\|\nabla W_{\overline{\Omega}^c}(f_n)\|_{L^2(\overline{\Omega}^c)}\rightarrow 0$ as $n\rightarrow\infty$. Our goal is once more to show that $\|W_{\overline{\Omega}^c}(f_n)\|_{L^2(\overline{\Omega}^c)}\rightarrow 0$ which would prove that the $H^1(\overline{\Omega})$-norm of $W_{\overline{\Omega}^c}(f_n)$ converges to $0$. We start by fixing some $R\gg1$ such that $\overline{\Omega}\subset B_R(0)$ and such that $\frac{1}{|x-y|^2}\leq \frac{2}{|x|^2}$ for all $y\in \overline{\Omega}$ and all $x$ with $|x|\geq R$. We can then estimate
	\begin{gather}
		\nonumber
		|W_{\overline{\Omega}^c}(f_n)(x)|\leq \frac{1}{4\pi}\int_{\partial\Omega}\frac{|f_n(y)|}{|x-y|^2}d\sigma(y)\leq \frac{\|\operatorname{Tr}(f_n)\|_{L^1(\partial\Omega)}}{2\pi|x|^2}\text{ for all }|x|\geq R.
	\end{gather}
	This allows us to conclude
	\begin{gather}
		\label{3E16}
		\|W_{\overline{\Omega}^c}(f_n)\|_{L^2(B^c_R)}\leq c(R)\|\operatorname{Tr}(f_n)\|_{L^1(\partial\Omega)}\rightarrow 0\text{ as }n\rightarrow\infty
	\end{gather}
	where we used that we had shown that $\operatorname{Tr}(f_n)$ converges strongly to $0$ in $L^1(\partial\Omega)$. We are therefore left with estimating $\|W_{\overline{\Omega}^c}(f_n)\|_{L^2(B_R\setminus \overline{\Omega})}$. We define $U_R:=B_R(0)\cap U$ and we fix any open subset $V\in \{U_R,\Omega_1,\dots,\Omega_n\}$. We can now estimate similar as in the case of $\Omega$
	\begin{gather}
		\nonumber
		\left|\int_V W_{\overline{\Omega}^c}(f_n)(x)d^3x\right|\leq \frac{1}{4\pi}\int_{\partial \Omega}|f_n(y)|\int_V\frac{1}{|x-y|^2}d^3xd\sigma(y)\leq \frac{\sup_{y\in \partial\Omega}\int_V\frac{1}{|x-y|^2}d^3x}{4\pi}\|\operatorname{Tr}(f_n)\|_{L^1(\partial\Omega)}.
	\end{gather}
	Now $\sup_{y\in \partial\Omega}\int_V\frac{1}{|x-y|^2}d^3x<\infty$ because each $V$ is bounded and $\partial\Omega$ is compact. We conclude by means of Poincar\'{e}'s inequality
	\begin{gather}
		\nonumber
		\|W_{\overline{\Omega}^c}(f_n)\|_{L^2(V)}\leq c(V)\left(\|\nabla W_{\overline{\Omega}^c}(f_n)\|_{L^2(V)}+\|\operatorname{Tr}(f_n)\|_{L^1(\partial\Omega)}\right)\rightarrow 0\text{ as }n\rightarrow\infty.
	\end{gather}
	Combining this with (\ref{3E16}) and the $L^2(\overline{\Omega}^c)$-gradient estimate of $\nabla W_{\overline{\Omega}^c}(f_n)$ we conclude the relation $\|W_{\overline{\Omega}^c}(f_n)\|_{H^1(\overline{\Omega}^c)}\rightarrow 0$ as $n\rightarrow\infty$. In turn, the continuity of the trace map implies that $\operatorname{Tr}(W_{\overline{\Omega}^c}(f_n))$ converges strongly to zero in $W^{\frac{1}{2},2}(\partial\Omega)$-norm. It follows once more from \Cref{3L1}
	\begin{gather}
		\label{3E17}
		\lim_{n\rightarrow\infty}\left(w_{\overline{\Omega}^c}(f_n)-\frac{\operatorname{Tr}(f_n)}{2}\right)=0\text{ strongly in }W^{\frac{1}{2},2}(\partial\Omega).
	\end{gather}
	We lastly observe that for every $h\in W^{\frac{1}{2},2}(\partial\Omega)$ we have $w_{\overline{\Omega}^c}(h)=-w_{\Omega}(h)$ since the outward unit normal along the boundary of $\overline{\Omega}^c$ coincides with the inner unit normal along the boundary of $\Omega$, i.e. it equals minus the outward unit normal along the boundary of $\Omega$. We can therefore add equations (\ref{3E15}) and (\ref{3E17}) and conclude
	\begin{gather}
		\nonumber
		\operatorname{Tr}(f_n)\rightarrow 0\text{ strongly in }W^{\frac{1}{2},2}(\partial\Omega).
	\end{gather}
	This contradicts the fact that the $f_n$ were chosen such that $\|\operatorname{Tr}(f_n)\|_{W^{\frac{1}{2},2}(\partial\Omega)}=1$ for all $n$, recall (\ref{3E14}). We conclude that there must exist some $C>0$ satisfying the last inequality in (\ref{3E12}) which we have shown to be equivalent to the contraction property of the operator $T$.
\end{proof}
In the upcoming proofs we will often construct currents $j$ as the twisted tangential trace of some $B\in H(\operatorname{curl},\Omega)$ in the sense that we will set $j:=B\times \mathcal{N}$ which will be well-defined elements of $\left(W^{-\frac{1}{2},2}(\partial\Omega)\right)^3$. However, to be valid currents we further need to guarantee that $B\times \mathcal{N}$ is div-free, i.e. that it belongs to the more restrictive space $W^{-\frac{1}{2},2}\mathcal{V}_0(\partial\Omega)$ which we defined as the completion of the space of $L^2\mathcal{V}_0(\partial\Omega)$, the square integrable, div-free fields tangent to $\partial\Omega$, with respect to the $\|\cdot \|_{W^{-\frac{1}{2},2}(\partial\Omega)}$-norm. The following provides a sufficient condition for this to be the case and the currents which we construct in the upcoming proofs are always of this type so that we will make repeatedly use of \Cref{3L4} without further explicit mention. Before we state the theorem we note that for any given $X\in H(\operatorname{curl},\Omega)$ we have $\operatorname{curl}(X)\in H(\operatorname{div},\Omega)$ and that therefore $\operatorname{curl}(X)$ always has a well-defined normal trace within the space $W^{-\frac{1}{2},2}(\partial\Omega)$.
\begin{lem}
	\label{3L4}
	Let $\Omega\subset\mathbb{R}^3$ be a bounded $C^{1,1}$-domain. Let $X\in H(\operatorname{curl},\Omega)$ such that $\mathcal{N}\cdot \operatorname{curl}(X)=0$. Then $X\times \mathcal{N}\in W^{-\frac{1}{2},2}\mathcal{V}_0(\partial\Omega)$.
\end{lem}
\begin{proof}[Proof of \cref{3L4}]
	According to our definition we have to show that there exists a sequence $(j_n)_n\subset L^2\mathcal{V}_0(\partial\Omega)$ which converges to $X\times \mathcal{N}$ in $W^{-\frac{1}{2},2}(\partial\Omega)$-norm. We start by approximating $X$ by a sequence $(X_n)_n\subset C^\infty_c(\mathbb{R}^3,\mathbb{R}^3)$ in $H(\operatorname{curl},\Omega)$-norm which is possible according to \cite[Theorem 2.10]{GR86}. Then on the one hand
	\begin{gather}
		\label{3Extra1}
		X_n\times \mathcal{N}\rightarrow X\times \mathcal{N}\text{ and }\mathcal{N}\cdot \operatorname{curl}(X_n)\rightarrow \mathcal{N}\cdot \operatorname{curl}(X)=0 \text{ in }W^{-\frac{1}{2},2}(\partial\Omega)\text{ respectively }
	\end{gather}
	where we used the continuity of the respective traces and that $\|\operatorname{curl}(X)\|_{H(\operatorname{div},\Omega)}=\|\operatorname{curl}(X)\|_{L^2(\Omega)}$ since $\operatorname{curl}(X)$ is always divergence-free. We note that $\int_{\partial\Omega}\mathcal{N}\cdot \operatorname{curl}(X_n)d\sigma=\int_{\Omega}\operatorname{div}(\operatorname{curl}(X))d^3x=0$ and so there exist solutions $f_n$ of class $W^{2,2}(\Omega)$ to the following boundary value problems
	\begin{gather}
		\label{3Extra2}
		\Delta f_n=0\text{ in }\Omega\text{, }\mathcal{N}\cdot \nabla f_n=\mathcal{N}\cdot \operatorname{curl}(X_n)\text{ on }\partial\Omega\text{, }\int_{\Omega}f_nd^3x=0\text{, }f_n\in W^{2,2}(\Omega).
	\end{gather}
	We observe that $\nabla f_n\in \mathcal{H}^{\perp_{L^2(\Omega)}}_D(\Omega)$ since for any $\nabla h\in \mathcal{H}_D(\Omega)$ we compute
	\begin{gather}
		\nonumber
		\int_{\Omega}\nabla f_n\cdot \nabla hd^3x=\int_{\partial\Omega}h \left(\mathcal{N}\cdot \nabla f_n\right)d\sigma=\int_{\partial\Omega}h\left(\mathcal{N}\cdot \operatorname{curl}(X_n)\right)d\sigma=\int_{\Omega}\nabla h\cdot \operatorname{curl}(X_n)d^3x=0
	\end{gather}
	where we integrated by parts, c.f. \cite[Theorem 2.11]{GR86}, in the last step and used that $\operatorname{curl}(\nabla h)=0$ and $\nabla h\times \mathcal{N}=0$. We conclude that $\nabla f_n\in L^2\mathcal{H}(\Omega)\cap \mathcal{H}^{\perp_{L^2(\Omega)}}_D(\Omega)$. According to \Cref{AppEL2} we can then find some $A_n\in H^1(\Omega,\mathbb{R}^3)$ with $\operatorname{curl}(A_n)=\nabla f_n$ and satisfying the a priori estimate $\|A_n\|_{H^1(\Omega)}\leq c\|\nabla f_n\|_{L^2(\Omega)}$ for some $c>0$ independent of $n$. We define $\widetilde{X}_n:=X_n-A_n$ and observe that $\operatorname{curl}(\widetilde{X}_n)=\operatorname{curl}(X_n)-\nabla f_n$ and consequently $\mathcal{N}\cdot \operatorname{curl}(\widetilde{X}_n)=0$ for all $n$ by definition of the $f_n$. Further, we find
	\begin{gather}
		\nonumber
		\|A_n\times \mathcal{N}\|_{W^{-\frac{1}{2},2}(\partial\Omega)}\leq c\|A_n\|_{H(\operatorname{curl},\Omega)}\leq \tilde{c}\|A_n\|_{H^1(\Omega)}\leq \hat{c}\|\nabla f_n\|_{L^2(\Omega)}
	\end{gather}
	for some suitable constants $c,\tilde{c},\hat{c}>0$ independent of $n$. We observe further that according to \Cref{CT1} we have the estimate $\|\nabla f_n\|_{L^2(\Omega)}\leq c\|\mathcal{N}\cdot \nabla f_n\|_{W^{-\frac{1}{2},2}(\partial\Omega)}=\|\mathcal{N}\cdot \operatorname{curl}(X_n)\|_{W^{-\frac{1}{2},2}(\partial\Omega)}\rightarrow 0$ for a suitable $c>0$ independent of $n$ and where the last claim follows from (\ref{3Extra1}). We overall conclude the following
	\begin{gather}
		\label{3Extra3}
		\widetilde{X}_n\times \mathcal{N}\rightarrow X\times \mathcal{N}\text{ in }W^{-\frac{1}{2},2}(\partial\Omega)\text{, }\mathcal{N}\cdot \operatorname{curl}(\widetilde{X}_n)=0\text{, }\widetilde{X}_n\in H^1(\Omega,\mathbb{R}^3)\text{ for all }n.
	\end{gather}
	We lastly claim that $\widetilde{X}_n\times \mathcal{N}\in L^2\mathcal{V}_0(\partial\Omega)$ for all $n$ so that the convergence $\widetilde{X}_n\times \mathcal{N}\rightarrow X\times \mathcal{N}$ in $W^{-\frac{1}{2},2}(\partial\Omega)$ will conclude the proof of the lemma. To see this we note that by standard trace inequalities we have $\widetilde{X}_n|_{\partial\Omega}\in \left(W^{\frac{1}{2},2}(\partial\Omega)\right)^3\subset \left(L^2(\partial\Omega)\right)^3$. Then consequently $\widetilde{X}_n\times \mathcal{N}$ is of the same class and clearly also tangent to $\partial\Omega$ a.e. so that we are left with proving that $\widetilde{X}_n\times \mathcal{N}$ is div-free on the boundary. Fix any $\psi\in C^{\infty}_c(\mathbb{R}^3)$, then we can perform the following integral manipulations which are justified because the $\widetilde{X}_n$ are of class $H^1(\Omega,\mathbb{R}^3)$
	\begin{gather}
		\nonumber
		\int_{\partial\Omega}\nabla \psi\cdot (\widetilde{X}_n\times \mathcal{N})d\sigma=\int_{\partial\Omega}\mathcal{N}\cdot\left(\nabla \psi\times \widetilde{X}_n\right)d\sigma=\int_{\Omega}\operatorname{div}\left(\nabla \psi\times \widetilde{X}_n\right)d^3x
		\\
		\nonumber
		=\int_{\Omega}\operatorname{curl}(\nabla \psi)\cdot \widetilde{X}_nd^3x-\int_{\Omega}\nabla \psi\cdot \operatorname{curl}(\widetilde{X}_n)d^3x=-\int_{\Omega}\nabla \psi\cdot \operatorname{curl}(\widetilde{X}_n)d^3x=0
	\end{gather}
	where the last identity follows because $\operatorname{div}(\operatorname{curl}(\widetilde{X}_n))=0$ and $\mathcal{N}\cdot \operatorname{curl}(\widetilde{X}_n)=0$ according to (\ref{3Extra3}).
\end{proof}
\subsection{Proof of \Cref{2T2}}
\begin{proof}[Proof of \Cref{2T2}]
	$\quad$
	\newline
	\newline
	\underline{Step 1 $\operatorname{Im}(\operatorname{BS}_{\partial\Omega})\subseteq L^2\mathcal{H}(\Omega)\cap \mathcal{H}^{\perp_{L^2(\Omega)}}_D(\Omega)$:}
	\newline
	\newline
	First it follows from \cite[Lemma C.1]{G24} that $\operatorname{Im}(\operatorname{BS}_{\partial\Omega})\subseteq L^2\mathcal{H}(\Omega)$. So we only need to show that the image of the Biot-Savart operator is $L^2(\Omega)$-orthogonal to the space $\mathcal{H}_D(\Omega)$. This essentially follows from the arguments provided in \cite[Proposition 3.6]{G24} which we recall here. By an approximation argument we may suppose that $j\in L^2\mathcal{V}_0(\partial\Omega)$. We can then fix any $B\in \mathcal{H}_D(\Omega)$ and observe that
	\begin{gather}
		\nonumber
		\int_{\Omega}\operatorname{BS}_{\partial\Omega}(j)(x)\cdot B(x)d^3x=\int_{\partial\Omega}j(y)\cdot \operatorname{BS}_{\Omega}(B)(y)d\sigma(y)
	\end{gather}
	where the volume Biot-Savart operator is given by $\operatorname{BS}_{\Omega}(B)(x)=\frac{1}{4\pi}\int_{\Omega}B(y)\times \frac{x-y}{|x-y|^3}d^3y$. We know that $B\in W^{1,p}(\Omega,\mathbb{R}^3)$ for all $1\leq p<\infty$, \cite[Lemma A.2]{G24}, and we notice that upon integrating by parts we may write $\operatorname{BS}_{\Omega}(Y)(x)=\frac{1}{4\pi}\int_{\partial\Omega}\frac{Y(y)\times \mathcal{N}(y)}{|x-y|}d\sigma(y)+\frac{1}{4\pi}\int_{\Omega}\frac{\operatorname{curl}(Y)(y)}{|x-y|}d^3y$ for all $Y\in W^{1,q}(\Omega,\mathbb{R}^3)$ for some $q>3$, recall also (\ref{2E16}). Since for any $B\in \mathcal{H}_D(\Omega)$ we have $B\times \mathcal{N}=0$ and $\operatorname{curl}(B)=0$ we find $\operatorname{BS}_{\Omega}(B)=0$, see also \cite[Theorem B]{CDG01} for a characterisation of the kernel of the volume Biot-Savart operator in the context of smooth domains. We overall conclude that $\operatorname{Im}(\operatorname{BS}_{\partial\Omega})\subset \mathcal{H}^{\perp_{L^2(\Omega)}}_D(\Omega)\cap L^2\mathcal{H}(\Omega)$.
	\newline
	\newline
	\underline{Step 2 $\mathcal{H}_{\operatorname{ex}}(\Omega)\cap \mathcal{H}^{\perp_{L^2(\Omega)}}_D(\Omega)\subseteq \operatorname{Im}(\operatorname{BS}_{\partial\Omega})$:}
	\newline
	\newline
	Fix any $\nabla h\in \mathcal{H}_{\operatorname{ex}}(\Omega)\cap \mathcal{H}^{\perp_{L^2(\Omega)}}_D(\Omega)$. We assume first that $\nabla h\in W^{1,p}(\Omega,\mathbb{R}^3)$ for some $p>3$. Consider the operator
	\begin{gather}
		\nonumber
		T_{\nabla h}:\mathcal{H}_{\operatorname{ex}}(\Omega)\cap \mathcal{H}^{\perp_{L^2(\Omega)}}_D(\Omega)\rightarrow \mathcal{H}_{\operatorname{ex}}(\Omega)\cap \mathcal{H}^{\perp_{L^2(\Omega)}}_D(\Omega)\text{, }\nabla f\mapsto \nabla h+T(\nabla f).
	\end{gather}
	According to \Cref{3C2} we have $T(\nabla f)\in \mathcal{H}_{\operatorname{ex}}(\Omega)\cap \mathcal{H}^{\perp_{L^2(\Omega)}}_D(\Omega)$ and since $\nabla h\in \mathcal{H}_{\operatorname{ex}}(\Omega)\cap \mathcal{H}^{\perp_{L^2(\Omega)}}_D(\Omega)$ we see that $T_{\nabla h}$ is well-defined and a contraction with the same contraction constant as $T$. According to the Banach fixed-point theorem the operator $T_{\nabla h}$ admits a unique fix point $\nabla f_*\in \mathcal{H}_{\operatorname{ex}}(\Omega)\cap \mathcal{H}^{\perp_{L^2(\Omega)}}_D(\Omega)$ which then satisfies
	\begin{gather}
		\label{3E18}
		\nabla h=\nabla f_*-T(\nabla f_*).
	\end{gather}
	We can then define the current $j:=\nabla f_*\times \mathcal{N}$ which is tangent to $\partial\Omega$, div-free on $\partial\Omega$ and of class $L^2(\partial\Omega)$, i.e. $j\in L^2\mathcal{V}_0(\partial\Omega)$. It follows then from the proof of \cite[Lemma 5.5]{G24} that we have the identity
	\begin{gather}
		\nonumber
		\operatorname{BS}_{\partial\Omega}(j)=\nabla f_*-T(\nabla f_*)-\operatorname{BS}_{\Omega}(\operatorname{curl}(\nabla f_*))=\nabla f_*-T(\nabla f_*)=\nabla h
	\end{gather}
	according to the fix point identity (\ref{3E18}). Now, if $\nabla h\in \mathcal{H}_{\operatorname{ex}}(\Omega)\cap \mathcal{H}^{\perp_{L^2(\Omega)}}_D(\Omega)$ is arbitrary, we can approximate it in $L^2(\Omega)$-norm by elements $\nabla h_n\in \mathcal{H}_{\operatorname{ex}}(\Omega)\cap \mathcal{H}^{\perp_{L^2(\Omega)}}_D(\Omega)$ of class $\nabla h_n\in \bigcap_{1\leq p<\infty}W^{1,p}(\Omega,\mathbb{R}^3)$. We can then construct currents $j_n:=\nabla f_n\times \mathcal{N}$ by means of the fix point procedure (\ref{3E18}) satisfying $\operatorname{BS}_{\partial\Omega}(j_n)=\nabla h_n$ and accordingly we may set $j:=\nabla f\times \mathcal{N}\in W^{-\frac{1}{2},2}\mathcal{V}_0(\partial\Omega)$ where $\nabla f$ denotes the unique fix point of (\ref{3E18}) for our given $\nabla h$. We observe that $\|j_n-j\|_{W^{-\frac{1}{2},2}(\partial\Omega)}\leq c\|\nabla f_n-\nabla f\|_{H(\operatorname{curl},\Omega)}=c\|\nabla f_n-\nabla f\|_{L^2(\Omega)}$ for some constant $c>0$ independent of $n$ by means of the continuity of the twisted tangential trace with respect to the $H(\operatorname{curl},\Omega)$-norm. We can now exploit the contraction property of $T$ and the defining equation (\ref{3E18}) of the $\nabla f_n$ to conclude
	\begin{gather}
		\nonumber
		\|\nabla f_n-\nabla f\|_{L^2(\Omega)}\leq \|\nabla h_n-\nabla h\|_{L^2(\Omega)}+\|T(\nabla f_n)-T(\nabla f)\|_{L^2(\Omega)}
		\\
		\nonumber
		\leq \|\nabla h_n-\nabla h\|_{L^2(\Omega)}+\lambda\|\nabla f_n-\nabla f\|_{L^2(\Omega)}\text{ for a suitable }0<\lambda<1.
	\end{gather}
	We overall infer that $\|j_n-j\|_{W^{-\frac{1}{2},2}(\partial\Omega)}\leq c\|\nabla h_n-\nabla h\|_{L^2(\Omega)}$ for some constant $c>0$ independent of $n$ and hence $j_n\rightarrow j$ in $W^{-\frac{1}{2},2}(\partial\Omega)$. The continuity of $\operatorname{BS}_{\partial\Omega}$ with respect to the $W^{-\frac{1}{2},2}(\partial\Omega)$-norm, c.f. \cite[Lemma C.1]{G24}, implies that $\operatorname{BS}_{\partial\Omega}(j_n)\rightarrow \operatorname{BS}_{\partial\Omega}(j)$ in $L^2(\Omega)$. On the other hand we know that $\operatorname{BS}_{\partial\Omega}(j_n)=\nabla h_n$ converges strongly to $\nabla h$ in $L^2(\Omega)$ from which we conclude $\operatorname{BS}_{\partial\Omega}(j)=\nabla h$ and consequently $\mathcal{H}_{\operatorname{ex}}(\Omega)\cap \mathcal{H}^{\perp_{L^2(\Omega)}}_D(\Omega)\subseteq \operatorname{Im}(\operatorname{BS}_{\partial\Omega})$.
	\newline
	\newline
	\underline{Step 3 $\mathcal{H}_N(\Omega)\subseteq \operatorname{Im}(\operatorname{BS}_{\partial\Omega})$:}
	\newline
	\newline
	Fix any $\Gamma\in \mathcal{H}_N(\Omega)$ and observe that $\Gamma\in W^{1,p}(\Omega,\mathbb{R}^3)$ for all $1\leq p<\infty$, c.f. \cite[Lemma A.1]{G24}. We can then define the current $j:=\Gamma\times \mathcal{N}\in W^{-\frac{1}{2},2}\mathcal{V}_0(\partial\Omega)$. It follows then similarly from the proof of \cite[Lemma 5.5]{G24} that
	\begin{gather}
		\label{3E19}
		\operatorname{BS}_{\partial\Omega}(j)=\Gamma-\operatorname{BS}_{\Omega}(\operatorname{curl}(\Gamma))-\frac{\nabla_x}{4\pi}\int_{\Omega}\Gamma(y)\cdot\frac{x-y}{|x-y|^3}d^3y=\Gamma-\frac{\nabla_x}{4\pi}\int_{\Omega}\Gamma(y)\cdot\frac{x-y}{|x-y|^3}d^3y
	\end{gather}
	where we used that $\operatorname{curl}(\Gamma)=0$. We finally observe that $\frac{x-y}{|x-y|^3}=\nabla_y\frac{1}{|x-y|}$ and hence compute
	\begin{gather}
		\nonumber
		\int_{\Omega}\Gamma(y)\cdot \frac{x-y}{|x-y|^3}d^3y=-\int_{\Omega}\frac{\operatorname{div}(\Gamma)(y)}{|x-y|}d^3y+\int_{\partial\Omega}\frac{\Gamma(y)\cdot \mathcal{N}(y)}{|x-y|}d\sigma(y)=0
	\end{gather}
	where we used that $\operatorname{div}(\Gamma)=0$ and $\mathcal{N}\cdot \Gamma=0$. It follows from (\ref{3E19}) that $\operatorname{BS}_{\partial\Omega}(j)=\Gamma$ and hence $\mathcal{H}_N(\Omega)\subseteq\operatorname{Im}(\operatorname{BS}_{\partial\Omega})$ as claimed.
	\newline
	\newline
	\underline{Step 4 $L^2\mathcal{H}(\Omega)\cap \mathcal{H}^{\perp_{L^2(\Omega)}}_D(\Omega)\subseteq \operatorname{Im}(\operatorname{BS}_{\partial\Omega})$:}
	\newline
	\newline
	Fix any $B\in L^2\mathcal{H}(\Omega)\cap \mathcal{H}^{\perp_{L^2(\Omega)}}_D(\Omega)$. We can perform a Hodge-decomposition, \cite[Theorem B.1]{G24}, of $B$ and write $B=\nabla h+\Gamma$ for suitable $\Gamma\in \mathcal{H}_N(\Omega)$ and $\nabla h\in \mathcal{H}_{\operatorname{ex}}(\Omega)$. We can now further $L^2(\Omega)$-decompose $\nabla h=\nabla f+\nabla \widetilde{f}$ for suitable $\nabla f\in \mathcal{H}_{\operatorname{ex}}(\Omega)\cap \mathcal{H}^{\perp_{L^2(\Omega)}}_D(\Omega)$ and $\nabla \widetilde{f}\in \mathcal{H}_D(\Omega)$. We observe that by step 3 $\Gamma$ lies in the image of the Biot-Savart operator and that by step 1 the image of the Biot-Savart operator is $L^2(\Omega)$-orthogonal to $\mathcal{H}_D(\Omega)$. We conclude that $B$, $\Gamma$ and $\nabla f$ are $L^2(\Omega)$-orthogonal to $\mathcal{H}_D(\Omega)$ and hence $\nabla \widetilde{f}\in \mathcal{H}_D(\Omega)\cap \mathcal{H}^{\perp_{L^2(\Omega)}}_D(\Omega)=\{0\}$ which yields $B=\nabla f+\Gamma$. According to step 2 and step 3 we can find currents $j_1,j_2\in W^{-\frac{1}{2},2}\mathcal{V}_0(\partial\Omega)$ with $\operatorname{BS}_{\partial\Omega}(j_1)=\nabla f$ and $\operatorname{BS}_{\partial\Omega}(j_2)=\Gamma$. By linearity of $\operatorname{BS}_{\partial\Omega}$ we find $\operatorname{BS}_{\partial\Omega}(j_1+j_2)=B$ and consequently $B\in \operatorname{Im}(\operatorname{BS}_{\partial\Omega})$ as desired.
\end{proof}
\section{Current reconstruction algorithm}
\subsection{Proof of \Cref{2T6}}
\begin{proof}[Proof of \Cref{2T6}]
	It follows first from \Cref{2T5} that for given $\epsilon>0$ there exists some $j\in W^{-\frac{1}{2},2}\mathcal{V}_0(\Sigma)$ satisfying $\|\operatorname{BS}_{\Sigma}(j)-B_T\|_{L^2(P)}\leq \epsilon$. Then according to (the easy direction of) \Cref{2T2} we find $B:=\operatorname{BS}_{\Sigma}(j)\in L^2\mathcal{H}(\Omega)$. We can then decompose $B$ according to the Hodge-decomposition theorem \cite[Theorem B.1]{G24} as $B=\widetilde{\Gamma}+\nabla f$ for suitable $\widetilde{\Gamma}\in \mathcal{H}_N(\Omega)$ and $\nabla f\in \mathcal{H}_{\operatorname{ex}}(\Omega)$. Our goal now is to show that we can find $N\in \mathbb{N}$ and constants $\alpha_0,\alpha_1,\dots,\alpha_N$ such that $\|\alpha_0\Gamma+\sum_{k=1}^N\alpha_k\nabla f_k-B\|_{L^2(\Omega)}\leq \epsilon$ for the given solutions $f_k$ of the BVPs $\Delta f_k=0$ in $\Omega$, $f_k|_{\partial\Omega}=\kappa_k$ for the fixed basis $\{\kappa_1,\kappa_2,\dots\}$ of $W^{\frac{1}{2},2}(\Sigma)$ and any fixed $\Gamma\in \mathcal{H}_N(\Omega)\setminus \{0\}$. First, we note that $\mathcal{H}_N(\Omega)$ is $1$-dimensional because $\Omega$ is assumed to be a solid torus so that for any fixed $\Gamma\in \mathcal{H}_N(\Omega)\setminus\{0\}$ there is a unique $\mu\in \mathbb{R}$ with $\widetilde{\Gamma}=\mu\Gamma$ and so we may pick $\alpha_0=\mu$. We are left with approximating $\nabla f$. We note first that $f$ is unique only up to constants and so we may fix a unique scalar potential by demanding $\int_{\Omega}fd^3x=0$. We then note that $f\in H^1(\Omega)$ and so $f$ has a well defined trace $\kappa:=f|_{\partial\Omega}\in W^{\frac{1}{2},2}(\Sigma)$. Since $\{\kappa_1,\kappa_2,\dots\}$ forms a basis of $W^{\frac{1}{2},2}(\Sigma)$ we can find constants $\alpha_1,\dots,\alpha_N$ for some $N\in \mathbb{N}$ such that $\|\sum_{k=1}^N\alpha_k\kappa_k-\kappa\|_{W^{\frac{1}{2},2}(\Sigma)}\leq \epsilon$. We can then make use of \Cref{3L3} to deduce
	\begin{gather}
		\nonumber
		\left\|\sum_{k=1}\alpha_k\nabla f_k-\nabla f\right\|_{L^2(\Omega)}\leq c\left\|\sum_{k=1}\alpha_k\kappa_k-\kappa\right\|_{W^{\frac{1}{2},2}(\Sigma)}\leq c\epsilon
	\end{gather}
	for some $c>0$ independent of $\nabla f$ and consequently we may achieve the estimate $\|\alpha_0\Gamma+\sum_{k=1}^N\alpha_k\nabla f_k-B\|_{L^2(\Omega)}\leq \epsilon$ which in combination with the initial estimate $\|B-B_T\|_{L^2(P)}\leq \epsilon$ implies the statement of the theorem.
\end{proof}
\subsection{Proof of \Cref{2L8}}
For the proof of \Cref{2L8} we first need to understand the boundary behaviour of the operator $T$. The following \Cref{4L1} is known in different contexts and here we provide a proof for our specific situation at hand for the convenience of the reader, see for instance \cite[Theorem 4.24 \& Equation (6.3)]{RCM21} for related results in the context of H\"{o}lder continuous functions.
\begin{lem}[Normal trace of the operator $T$]
	\label{4L1}
	Let $\Omega\subset\mathbb{R}^3$ be a bounded $C^{1,1}$-domain and let $T:\mathcal{H}_{\operatorname{ex}}(\Omega)\cap \mathcal{H}^{\perp_{L^2(\Omega)}}_D(\Omega)\rightarrow \mathcal{H}_{\operatorname{ex}}(\Omega)\cap \mathcal{H}^{\perp_{L^2(\Omega)}}_D(\Omega)$ be given by $T(\nabla f)(x):=\frac{\nabla_x}{4\pi}\int_{\Omega}\nabla f(y)\cdot \frac{x-y}{|x-y|^3}d^3y$. Then
	\begin{gather}
		\nonumber
		\mathcal{N}\cdot T(\nabla g)=\frac{\mathcal{N}\cdot \nabla g}{2}-w^{\operatorname{Tr}}_{\Omega}(\mathcal{N}\cdot \nabla g)
	\end{gather}
	where $w_{\Omega}^{\operatorname{Tr}}$ is the transpose of the double layer potential as defined in (\ref{2E7}).
\end{lem}
\begin{proof}[Proof of \Cref{4L1}]
	We assume first that $\nabla g\in \mathcal{H}_{\operatorname{ex}}(\Omega)\cap \mathcal{H}^{\perp_{L^2(\Omega)}}_D(\Omega)\cap H^1(\Omega,\mathbb{R}^3)$. The claimed identity is then equivalent to the integral identity
	\begin{gather}
		\label{4E1}
		\int_{\partial\Omega}\psi\cdot (\mathcal{N}\cdot T(\nabla g))d\sigma=\int_{\partial\Omega}\psi\cdot\left(\frac{\mathcal{N}\cdot \nabla g}{2}-w_{\Omega}^{\operatorname{Tr}}(\mathcal{N}\cdot \nabla g)\right)d\sigma\text{ for all }\psi\in W^{\frac{1}{2},2}(\partial\Omega).
	\end{gather}
	Since the $C^1(\partial\Omega)$-functions are dense in $W^{\frac{1}{2},2}(\partial\Omega)$, c.f. \cite[Proposition 3.40]{DD12}, we may assume that $\psi\in C^1(\overline{\Omega})$. We can then express by definition of the normal trace
	\begin{gather}
		\label{4E2}
		\int_{\partial\Omega}\psi\cdot (\mathcal{N}\cdot T(\nabla g))d\sigma=\int_{\Omega}\nabla \psi\cdot T(\nabla g)d^3x+\int_{\Omega}\psi\cdot\operatorname{div}(T(\nabla f))d^3x=\int_{\Omega}\nabla \psi\cdot T(\nabla g)d^3x
	\end{gather}
since $T(\nabla g)$ maps into $\mathcal{H}_{\operatorname{ex}}(\Omega)$ and hence is div-free. We then observe that $\frac{x-y}{|x-y|^3}=\nabla_y\frac{1}{|x-y|}$ and due to the regularity of $\nabla g$ we may perform an integration by parts in the following expression
\begin{gather}
	\label{4E3}
	\frac{1}{4\pi}\int_{\Omega}\nabla g(y)\cdot \frac{x-y}{|x-y|^3}d^3y=\frac{1}{4\pi}\int_{\Omega}\nabla g(y)\cdot \nabla_y\frac{1}{|x-y|}d^3y=\frac{1}{4\pi}\int_{\partial\Omega}\frac{\mathcal{N}(y)\cdot \nabla g(y)}{|x-y|}d\sigma(y)
\end{gather}
where we used that $\Delta g=0$. Since $T(\nabla g)(x)$ is the gradient of the left hand side in (\ref{4E3}), we obtain
\begin{gather}
	\nonumber
	T(\nabla g)(x)=\frac{1}{4\pi}\int_{\partial\Omega}\left(\mathcal{N}(y)\cdot \nabla g(y)\right)\frac{y-x}{|x-y|^3}d\sigma(y).
\end{gather}
We can hence write
\begin{gather}
	\label{4E4}
	\int_{\Omega}\nabla \psi(x)\cdot T(\nabla g)(x)d^3x=\int_{\partial\Omega}\left(\mathcal{N}\cdot \nabla g(y)\right)\frac{1}{4\pi}\int_{\Omega}\nabla \psi(x)\cdot \frac{y-x}{|x-y|^3}d^3xd\sigma(y).
\end{gather}
We observe that since $\psi\in C^1(\overline{\Omega})$, the map $\mathbb{R}^3\rightarrow \mathbb{R}$, $y\mapsto \int_{\Omega}\nabla \psi(x)\cdot \frac{y-x}{|x-y|^3}d^3x$ is continuous and therefore we can fix for some given $y\in \partial\Omega$ any sequence $(y_n)_n\subset \overline{\Omega}^c$ with $y_n\rightarrow y$ and find
\begin{gather}
	\label{4E5}
	\frac{1}{4\pi}\int_{\Omega}\nabla \psi(x)\cdot \frac{y-x}{|x-y|^3}d^3x=\frac{\lim_{n\rightarrow\infty}}{4\pi}\int_{\Omega}\nabla \psi(x)\cdot \frac{y_n-x}{|y_n-x|^3}d^3x.
\end{gather}
Now we compute
\begin{gather}
	\nonumber
	\frac{1}{4\pi}\int_{\Omega}\nabla \psi(x)\cdot \frac{y_n-x}{|y_n-x|^3}d^3x=\frac{1}{4\pi}\int_{\partial\Omega}\psi(x)\mathcal{N}(x)\cdot \frac{y_n-x}{|y_n-x|^3}d\sigma(x)=-W_{\overline{\Omega}^c}(\psi)(y_n)
\end{gather}
where we used that $\frac{y_n-x}{|y_n-x|^3}$ is div-free in $\Omega$ and that the outward unit normal to $\Omega$ equals minus the outward unit normal to $\overline{\Omega}^c$, recall also (\ref{3E1}) for the definition of $W_{\Omega}$. We can insert this into (\ref{4E5}) and find
\begin{gather}
	\nonumber
	\frac{1}{4\pi}\int_{\Omega}\nabla \psi(x)\cdot \frac{y-x}{|x-y|^3}d^3x=-\lim_{n\rightarrow\infty}W_{\overline{\Omega}^c}(\psi)(y_n)
\end{gather}
where $(y_n)_n\subset \overline{\Omega}^c$ is an arbitrary sequence converging to $y\in \partial\Omega$. According to (\ref{3E6}) we may in particular find a sequence such that $\lim_{n\rightarrow\infty}W_{\overline{\Omega}^c}(\psi)(y_n)\rightarrow -\frac{\psi(y)}{2}+w_{\overline{\Omega}^c}(\psi)(x)=-\frac{\psi(y)}{2}-w_{\Omega}(\psi)(y)$ where we used once more that the outer normal of $\Omega$ equals minus the outer unit of $\overline{\Omega}^c$ in the last step. We find
\begin{gather}
	\nonumber
	\frac{1}{4\pi}\int_{\Omega}\nabla \psi(x)\cdot \frac{y-x}{|x-y|^3}d^3x=\frac{\psi(y)}{2}+w_{\Omega}(\psi)(y).
\end{gather}
We can insert this into (\ref{4E4}) which together with the identity $\int_{\partial\Omega} \alpha\cdot  w_{\Omega}^{\operatorname{Tr}}(\beta)d\sigma=-\int_{\partial\Omega}\beta\cdot w_{\Omega}(\alpha)d\sigma$ for all $\beta\in W^{-\frac{1}{2},2}(\partial\Omega)$ and $\alpha\in W^{\frac{1}{2},2}(\partial\Omega)$ (note the minus sign) yields
\begin{gather}
	\nonumber
	\int_{\Omega}\nabla \psi(x)\cdot T(\nabla g)(x)d^3x=\int_{\partial\Omega}\psi(y)\cdot \left(\frac{\mathcal{N}\cdot \nabla g}{2}-w^{\operatorname{Tr}}_{\Omega}(\mathcal{N}\cdot \nabla g)\right)d\sigma(y).
\end{gather}
Then (\ref{4E2}) yields the desired identity (\ref{4E1}). The general case $\nabla g\in \mathcal{H}_{\operatorname{ex}}(\Omega)\cap \mathcal{H}^{\perp_{L^2(\Omega}}_D(\Omega)$ follows by approximation by elements of class $H^1(\Omega,\mathbb{R}^3)$ in $L^2(\Omega)$-norm and the continuity of all quantities involved with respect to this convergence.
\end{proof}
\begin{proof}[Proof of \Cref{2L8}]
	We start with $B\in L^2\mathcal{H}(\Omega)\cap \mathcal{H}^{\perp_{L^2(\Omega)}}_D(\Omega)$ and decompose it by means of the Hodge-decomposition theorem as $B=\nabla h+\Gamma$ for suitable $\Gamma\in \mathcal{H}_N(\Omega)$ and $\nabla h\in \mathcal{H}_{\operatorname{ex}}(\Omega)$. According to step 3 in the proof of \Cref{2T2} we have $\operatorname{BS}_{\partial\Omega}(\Gamma\times \mathcal{N})=\Gamma$ and according to step 4 of the proof of \Cref{2T2} we see that in fact $\nabla h\in \mathcal{H}_{\operatorname{ex}}(\Omega)\cap \mathcal{H}^{\perp_{L^2(\Omega)}}_D(\Omega)$ and so we see that according to step 2 of the proof of \Cref{2T2} we have $\operatorname{BS}_{\partial\Omega}(\nabla f_*\times \mathcal{N})=\nabla h$ where $\nabla f_*\in \mathcal{H}_{\operatorname{ex}}(\Omega)\cap \mathcal{H}^{\perp_{L^2(\Omega)}}_D(\Omega)$ is the unique fix point of the operator $T_{\nabla h}:=\nabla h+T$ where $T(\nabla g)(x):=\frac{\nabla_x}{4\pi}\int_{\Omega}\nabla g(y)\cdot \frac{x-y}{|x-y|^3}d^3y$.
	\newline
	Our goal now will be to show that the fix point $\nabla f_*$ of $T_{\nabla h}$ can be equivalently characterised as the gradient of the unique solution $f$ of the BVP
	\begin{gather}
		\label{4E6}
		\Delta f=0\text{ in }\Omega\text{, }\mathcal{N}\cdot \nabla f=\left(\frac{\operatorname{Id}}{2}+w^{\operatorname{Tr}}_{\Omega}\right)^{-1}\left(B\cdot \mathcal{N}\right)\text{ on }\partial\Omega\text{ and }\int_{\Omega}f^3x=0.
	\end{gather}
	Once we show that $\nabla f=\nabla f_*$ we can conclude $\operatorname{BS}_{\partial\Omega}(\Gamma\times \mathcal{N}+\nabla f\times \mathcal{N})=B$ which is the claim of the lemma. We note that the uniqueness of solutions to the BVP (\ref{4E6}) follows immediately from the uniqueness of solutions to Neumann problems with prescribed mean value. The existence will follow once we show that the fix point function $f_{*}$ normalised by $\int_{\Omega}f_*d^3x=0$ is a solution and in turn the uniqueness of solutions to (\ref{4E6}) will provide an equivalent characterisation of $f_*$ as the unique solution of the BVP (\ref{4E6}). We start with the fix point identity
	\begin{gather}
		\nonumber
		\nabla f_*=T_{\nabla h}(\nabla f_*)=\nabla h+T(\nabla f_*).
	\end{gather}
	We make use of \Cref{4L1} to conclude by means of the fix point property
	\begin{gather}
		\nonumber
		\frac{\mathcal{N}\cdot \nabla f_*}{2}-w^{\operatorname{Tr}}_{\Omega}(\mathcal{N}\cdot \nabla f_*)=\mathcal{N}\cdot  T(\nabla f_*)=\mathcal{N}\cdot \nabla f_*-\mathcal{N}\cdot \nabla h \\
		\label{4E7}
		\Leftrightarrow \left(\frac{\operatorname{Id}}{2}+w^{\operatorname{Tr}}_{\Omega}\right)(\mathcal{N}\cdot \nabla f_*)=\mathcal{N}\cdot \nabla h=\mathcal{N}\cdot B
	\end{gather}
	where we used that $B=\nabla h+\Gamma$ and $\mathcal{N}\cdot \Gamma=0$ since $\Gamma\in \mathcal{H}_N(\Omega)$. We note that once we argue that $\frac{\operatorname{Id}}{2}+w^{\operatorname{Tr}}_{\Omega}$ is invertible it follows from (\ref{4E7}) that $f_{*}$ satisfies the Neumann boundary condition of (\ref{4E6}). Further, $\int_{\Omega}f_*d^3x=0$ holds by our normalisation and $\nabla f_*\in \mathcal{H}_{\operatorname{ex}}(\Omega)$ which implies $\Delta f_*=0$ in $\Omega$ and hence the theorem will be proven. But it follows from the upcoming \Cref{4L2} that we may invert $\frac{\operatorname{Id}}{2}+w^{\operatorname{Tr}}_{\Omega}$. More precisely it is shown that the operator
	\begin{gather}
		\nonumber
		\frac{\operatorname{Id}}{2}+w^{\operatorname{Tr}}_{\Omega}:\left\{\mathcal{N}\cdot \nabla f\mathrel{\bigg|} \nabla f\in \mathcal{H}_{\operatorname{ex}}(\Omega)\cap\mathcal{H}^{\perp_{L^2(\Omega)}}_D(\Omega)\right\}\rightarrow \left\{\mathcal{N}\cdot B\mathrel{\bigg|} B\in L^2\mathcal{H}(\Omega)\cap \mathcal{H}^{\perp_{L^2(\Omega)}}_D(\Omega)\right\}
	\end{gather}
	is invertible and thus the proof is complete.
\end{proof}
We note that the invertibility and characterisation of the operator $\frac{\operatorname{Id}}{2}+w^{\operatorname{Tr}}_{\Omega}$ has been studied in different contexts, see for instance \cite[Chapter 6.5]{RCM21} for the case of H\"{o}lder regular functions. The following lemma contains the invertibility of this operator in our context, whose proof is straightforward with the results already established in the present manuscript. We note that we have the identity
\begin{gather}
	\nonumber
	\left\{\mathcal{N}\cdot \nabla f\mathrel{\bigg|} \nabla f\in \mathcal{H}_{\operatorname{ex}}(\Omega)\cap\mathcal{H}^{\perp_{L^2(\Omega)}}_D(\Omega)\right\}= \left\{\mathcal{N}\cdot B\mathrel{\bigg|} B\in L^2\mathcal{H}(\Omega)\cap \mathcal{H}^{\perp_{L^2(\Omega)}}_D(\Omega)\right\}\subset W^{-\frac{1}{2},2}\mathcal{V}(\partial\Omega).
\end{gather}
which follows immediately from the Hodge-decomposition theorem as has been seen in the course of the proof of \Cref{2L8}.
\begin{lem}
	\label{4L2}
	Let $\Omega\subset\mathbb{R}^3$ be a bounded $C^{1,1}$-domain and define
	\begin{gather}
		\nonumber
	\mathcal{D}:=\left\{\mathcal{N}\cdot \nabla f\mathrel{\bigg|} \nabla f\in \mathcal{H}_{\operatorname{ex}}(\Omega)\cap\mathcal{H}^{\perp_{L^2(\Omega)}}_D(\Omega)\right\}.	
	\end{gather}
	Then the operator
	\begin{gather}
		\nonumber
		\frac{\operatorname{Id}}{2}+w^{\operatorname{Tr}}_{\Omega}:\left(\mathcal{D},\|\cdot\|_{W^{-\frac{1}{2},2}(\partial\Omega)}\right)\rightarrow \left(\mathcal{D},\|\cdot\|_{W^{-\frac{1}{2},2}(\partial\Omega)}\right)
	\end{gather}
	is a bounded linear isomorphism with a bounded linear inverse.
\end{lem}
\begin{proof}[Proof of \Cref{4L2}]
	We first note that $\mathcal{D}$ is a closed subspace of $W^{-\frac{1}{2},2}(\partial\Omega)$ and thus complete. To see this take any sequence $\mathcal{N}\cdot \nabla f_n$ with $\nabla f_n\in \mathcal{H}_{\operatorname{ex}}(\Omega)\cap \mathcal{H}^{\perp_{L^2(\Omega)}}_D(\Omega)$ which converges in $W^{-\frac{1}{2},2}(\partial\Omega)$ to some $\psi\in W^{-\frac{1}{2},2}(\partial\Omega)$. Then \Cref{CT1} implies that the $L^2(\Omega)$-norm of the $\nabla f_n$ is bounded and since $\mathcal{H}_{\operatorname{ex}}(\Omega)\cap \mathcal{H}^{\perp_{L^2(\Omega)}}_D(\Omega)$ is a Hilbert space together with the $L^2$-inner product we conclude that the $\nabla f_n$ converge weakly to some $\nabla f\in \mathcal{H}_{\operatorname{ex}}(\Omega)\cap \mathcal{H}^{\perp_{L^2(\Omega)}}_D(\Omega)$. By continuity of the normal trace we conclude that $\mathcal{N}\cdot \nabla f_n$ converges weakly to $\mathcal{N}\cdot \nabla f$ in $W^{-\frac{1}{2},2}(\partial\Omega)$. Since weak and strong limits coincide we find $\psi=\mathcal{N}\cdot \nabla f\in \mathcal{D}$. Therefore, by means of the bounded inverse theorem, we only need to prove that $\frac{\operatorname{Id}}{2}+w^{\operatorname{Tr}}_{\Omega}:\mathcal{D}\rightarrow\mathcal{D}$ is a well-defined, bounded, linear bijective map. The linearity is clear and the boundedness follows from the boundedness of the operator $w^{\operatorname{Tr}}_{\Omega}$ as a map from $W^{-\frac{1}{2},2}(\partial\Omega)$ into $W^{-\frac{1}{2},2}(\partial\Omega)$.
	
	We argue now that the operator is well-defined, i.e. it maps elements of $\mathcal{D}$ to elements in $\mathcal{D}$. We start with an arbitrary fixed element $\mathcal{N}\cdot \nabla f\in \mathcal{D}$ and we write $\left(\frac{\operatorname{Id}}{2}+w^{\operatorname{Tr}}_{\Omega}\right)(\mathcal{N}\cdot \nabla f)=-\left(\frac{\operatorname{Id}}{2}-w^{\operatorname{Tr}}_{\Omega}\right)(\mathcal{N}\cdot \nabla f)+\mathcal{N}\cdot \nabla f$. It then follows from \Cref{4L1} that $\left(\frac{\operatorname{Id}}{2}+w^{\operatorname{Tr}}_{\Omega}\right)(\mathcal{N}\cdot \nabla f)=\mathcal{N}\cdot \nabla f-\mathcal{N}\cdot T(\nabla f)\in \mathcal{D}$ because $\nabla f\in \mathcal{H}_{\operatorname{ex}}(\Omega)\cap \mathcal{H}^{\perp_{L^2(\Omega)}}_D(\Omega)$ and $T$ maps $\mathcal{H}_{\operatorname{ex}}(\Omega)\cap \mathcal{H}^{\perp_{L^2(\Omega)}}_D(\Omega)$ into $\mathcal{H}_{\operatorname{ex}}(\Omega)\cap \mathcal{H}^{\perp_{L^2(\Omega)}}_D(\Omega)$, c.f. \Cref{3C2}.
	
	To see that $\frac{\operatorname{Id}}{2}+w^{\operatorname{Tr}}_{\Omega}$ is surjective we may simply follow the arguments of the proof of \Cref{2L8} until (\ref{4E7}) which shows that for any $\nabla h\in \mathcal{H}_{\operatorname{ex}}(\Omega)\cap \mathcal{H}^{\perp_{L^2(\Omega)}}_D(\Omega)$ there is some $\nabla f_*\in \mathcal{H}_{\operatorname{ex}}(\Omega)\cap \mathcal{H}^{\perp_{L^2(\Omega)}}_D(\Omega)$ with $\left(\frac{\operatorname{Id}}{2}+w^{\operatorname{Tr}}_{\Omega}\right)(\mathcal{N}\cdot \nabla f_*)=\mathcal{N}\cdot \nabla h$.
	
	To see that $\frac{\operatorname{Id}}{2}+w^{\operatorname{Tr}}_{\Omega}$ is injective, suppose that $\left(\frac{\operatorname{Id}}{2}+w^{\operatorname{Tr}}_{\Omega}\right)(\mathcal{N}\cdot \nabla f)=0$. Following \Cref{4L1} we can express this condition as
	\begin{gather}
		\nonumber
		\mathcal{N}\cdot \nabla f-\mathcal{N}\cdot T(\nabla f)=0.
	\end{gather}
	In other words $\mathcal{N}\cdot \left(\nabla f-T(\nabla f)\right)=0$ and since $\nabla f-T(\nabla f)\in\mathcal{H}_{\operatorname{ex}}(\Omega)$ it is div- and curl-free so that $\nabla f-T(\nabla f)\in \mathcal{H}_N(\Omega)\cap \mathcal{H}_{\operatorname{ex}}(\Omega)$. Since $\mathcal{H}_N(\Omega)$ and $\mathcal{H}_{\operatorname{ex}}(\Omega)$ are $L^2(\Omega)$-orthogonal we infer $\nabla f-T(\nabla f)=0$ or equivalently $T(\nabla f)=\nabla f$. Hence, $\nabla f$ is a fix point of $T$. But $T$ is a contraction and so has a unique fix point. By linearity of $T$ we get $T(0)=0$ and thus $0$ is the unique fix point, i.e. $\nabla f=0$, which in turn implies $\mathcal{N}\cdot \nabla f=0$ which proves injectivity of $\frac{\operatorname{Id}}{2}+w^{\operatorname{Tr}}_{\Omega}$ and completes the proof of the lemma.
\end{proof}
\subsection{Proof of \Cref{2T9}}
Before we come to the proof of \Cref{2T9} we introduce the following inner product on the space $W^{-\frac{1}{2},2}(\partial\Omega)$ which gives rise to a norm equivalent to the standard $W^{-\frac{1}{2},2}(\partial\Omega)$-norm on $\partial\Omega$, c.f. \cite[Th\'{e}or\`{e}me 1.1]{NedPlan73},
\begin{gather}
	\label{4E8}
	\langle \cdot,\cdot \rangle: W^{-\frac{1}{2},2}(\partial\Omega)\times W^{-\frac{1}{2},2}(\partial\Omega)\rightarrow\mathbb{R}\text{, }(\psi,\phi)\mapsto \frac{1}{4\pi}\int_{\partial\Omega}\int_{\partial\Omega}\frac{\psi(x)\cdot \phi(y)}{|x-y|}d\sigma(y)d\sigma(x).
\end{gather}
\begin{lem}
	\label{4L3}
	Let $\Omega\subset\mathbb{R}^3$ be a bounded $C^{1,1}$-domain. Then $\mathcal{D}:=\left\{\mathcal{N}\cdot \nabla f\mid \nabla f\in \mathcal{H}_{\operatorname{ex}}(\Omega)\cap \mathcal{H}^{\perp_{L^2(\Omega)}}_D(\Omega)\right\}$ together with the inner product defined in (\ref{4E8}) is a Hilbert space and for every $\mathcal{N}\cdot \nabla f,\mathcal{N}\cdot \nabla h\in \mathcal{D}$ we have the identity
	\begin{gather}
		\label{4E9}
		\left\langle \left(w_{\Omega}^{\operatorname{Tr}}-\frac{\operatorname{Id}}{2}\right)(\mathcal{N}\cdot \nabla f),\mathcal{N}\cdot \nabla h\right\rangle=-\int_{\Omega}T(\nabla f)(x)\cdot T(\nabla h)(x)d^3x
	\end{gather}
	where the operator $T$ is as usual defined by
	\begin{gather}
		\nonumber
		T:\mathcal{H}_{\operatorname{ex}}(\Omega)\cap \mathcal{H}^{\perp_{L^2(\Omega)}}_D(\Omega)\rightarrow\mathcal{H}_{\operatorname{ex}}(\Omega)\cap \mathcal{H}^{\perp_{L^2(\Omega)}}_D(\Omega)\text{, }\nabla f\mapsto \left(x\mapsto \frac{\nabla_x}{4\pi}\int_{\Omega}\nabla f(y)\cdot \frac{x-y}{|x-y|^3}d^3y\right).
	\end{gather}
\end{lem}
\begin{proof}[Proof of \Cref{4L3}]
	Just like in the proof of \Cref{4L1} we may by a density argument assume that $\nabla f,\nabla h\in H^1(\Omega,\mathbb{R}^3)$. It further follows from \Cref{4L1} that we have the identity 
	\begin{gather}
		\label{4E10}
		\left(w_{\Omega}^{\operatorname{Tr}}-\frac{\operatorname{Id}}{2}\right)(\mathcal{N}\cdot \nabla f)=-\mathcal{N}\cdot T(\nabla f).
	\end{gather}
	In addition, we observe that the map
	\begin{gather}
		\label{4E11}
		\mathbb{R}^3\rightarrow \mathbb{R}\text{, }x\mapsto \int_{\partial\Omega}\frac{\mathcal{N}(y)\cdot \nabla h(y)}{|x-y|}d\sigma(y)
	\end{gather}
	is continuous because $\mathcal{N}\cdot \nabla h\in W^{\frac{1}{2},2}(\partial\Omega)\hookrightarrow L^4(\partial\Omega)$ by standard trace inequalities and fractional embedding theorems and since $4>2$. For fixed $x\in \partial\Omega$ we can now take any sequence $(x_n)_n\subset\mathbb{R}^3\setminus \overline{\Omega}$ converging to $x$ and find
	\begin{gather}
		\nonumber
		\int_{\partial\Omega}\frac{\mathcal{N}(y)\cdot \nabla h(y)}{|x-y|}d\sigma(y)=\lim_{n\rightarrow\infty}\int_{\partial\Omega}\frac{\mathcal{N}(y)\cdot \nabla h(y)}{|x_n-y|}d\sigma(y)=\lim_{n\rightarrow\infty}\int_{\Omega}\nabla h(y)\cdot \frac{x_n-y}{|x_n-y|^3}d^3y
	\end{gather}
	where we used the continuity of (\ref{4E11}) and that $\nabla h$ is div-free. We finally note that $H^1(\Omega)\hookrightarrow L^6(\Omega)$ and that $6>3$ so that it follows easily that the map
	\begin{gather}
		\nonumber
		\mathbb{R}^3\rightarrow\mathbb{R}\text{, }x\mapsto \int_{\Omega}\nabla h(y)\cdot \frac{x-y}{|x-y|^3}d^3y
	\end{gather}
	is continuous. We conclude overall
	\begin{gather}
		\label{4E12}
		\int_{\partial\Omega}\frac{\mathcal{N}(y)\cdot \nabla h(y)}{|x-y|}d\sigma(y)=\int_{\Omega}\nabla h(y)\cdot \frac{x-y}{|x-y|^3}d^3y\text{ for all }x\in \partial\Omega.
	\end{gather}
	We combine (\ref{4E12}) and (\ref{4E10}) and find
	\begin{gather}
		\nonumber
		\left\langle \left(w_{\Omega}^{\operatorname{Tr}}-\frac{\operatorname{Id}}{2}\right)(\mathcal{N}\cdot \nabla f),\mathcal{N}\cdot \nabla h\right\rangle=-\frac{1}{4\pi}\int_{\partial\Omega}(\mathcal{N}\cdot T(\nabla f))\cdot \int_{\Omega}\nabla h(y)\cdot \frac{x-y}{|x-y|^3}d^3yd\sigma(x).
	\end{gather}
	We finally note that $\frac{1}{4\pi}\int_{\Omega}\nabla h(y)\cdot \frac{x-y}{|x-y|^3}d^3y$, $x\in \partial\Omega$, is the trace of the $H^1(\Omega)$ function $H:\Omega\rightarrow \mathbb{R}\text{, }x\mapsto \frac{1}{4\pi}\int_{\Omega}\nabla h(y)\cdot \frac{x-y}{|x-y|^3}d^3y$ (by continuity of $H$). Consequently we obtain
	\begin{gather}
		\nonumber
		\left\langle \left(w_{\Omega}^{\operatorname{Tr}}-\frac{\operatorname{Id}}{2}\right)(\mathcal{N}\cdot \nabla f),\mathcal{N}\cdot \nabla h\right\rangle=-\int_{\partial\Omega}\operatorname{Tr}(H)(x)\mathcal{N}(x)\cdot T(\nabla f)(x)d\sigma(y)
		\\
		\nonumber
		=-\int_{\Omega}\nabla H(x)\cdot T(\nabla f)(x)d^3x
	\end{gather}
	where we used that $T(\nabla f)$ is div-free. The claim now follows by observing that $\nabla_xH(x)=T(\nabla h)(x)$.
\end{proof}
\begin{cor}
	\label{4C4}
	Let $\Omega\subset\mathbb{R}^3$ be a bounded $C^{1,1}$-domain. With the same notation as in \Cref{4L3}, we consider the operator
	\begin{gather}
		\nonumber
		w^{\operatorname{Tr}}_{\Omega}-\frac{\operatorname{Id}}{2}:\mathcal{D}\rightarrow\mathcal{D}.
	\end{gather}
	Then $\left\|w^{\operatorname{Tr}}_{\Omega}-\frac{\operatorname{Id}}{2}\right\|\leq \lambda <1$, where $\left\|w^{\operatorname{Tr}}_{\Omega}-\frac{\operatorname{Id}}{2}\right\|$ denotes the operator norm induced by (\ref{4E8}) and where $0<\lambda<1$ denotes the contraction constant of the operator $T$, c.f. \Cref{3C2}.
\end{cor}
\begin{proof}[Proof of \Cref{4C4}]
	We fix any $\mathcal{N}\cdot \nabla f\in \mathcal{D}$ and observe that according to \Cref{4L1} $\left(w^{\operatorname{Tr}}_{\Omega}-\frac{\operatorname{Id}}{2}\right)(\mathcal{N}\cdot \nabla f)=\mathcal{N}\cdot \nabla h$ with $\nabla h:=-T(\nabla f)$. Hence, according to \Cref{4L3} we obtain
	\begin{gather}
		\nonumber
		\left\|\left(w^{\operatorname{Tr}}_{\Omega}-\frac{\operatorname{Id}}{2}\right)(\mathcal{N}\cdot \nabla f)\right\|^2=\left\langle \left(w^{\operatorname{Tr}}_{\Omega}-\frac{\operatorname{Id}}{2}\right)(\mathcal{N}\cdot \nabla f),\mathcal{N}\cdot\nabla h\right\rangle=-\int_{\Omega}T(\nabla f)\cdot T(\nabla h)d^3x
		\\
		\nonumber
		=\int_{\Omega}T(\nabla f)\cdot T^2(\nabla f)d^3x\leq \|T(\nabla f)\|_{L^2(\Omega)}\|T^2(\nabla f)\|_{L^2(\Omega)}\leq \lambda \|T(\nabla f)\|^2_{L^2(\Omega)}
	\end{gather}
	where we used the definition of $\nabla h$, the Cauchy-Schwarz inequality and the contraction property of $T$, c.f. \Cref{3C2}. We employ once more \Cref{4L3} and the Cauchy-Schwarz inequality to arrive at
	\begin{gather}
		\nonumber
		\|T(\nabla f)\|^2_{L^2(\Omega)}=-\left\langle\left(w^{\operatorname{Tr}}_{\Omega}-\frac{\operatorname{Id}}{2}\right)(\mathcal{N}\cdot \nabla f),\mathcal{N}\cdot \nabla f\right\rangle\leq \left\|\left(w^{\operatorname{Tr}}_{\Omega}-\frac{\operatorname{Id}}{2}\right)(\mathcal{N}\cdot \nabla f)\right\|\|\mathcal{N}\cdot \nabla f\|
	\end{gather}
	so that we overall arrive at
	\begin{gather}
		\nonumber
		\left\|\left(w^{\operatorname{Tr}}_{\Omega}-\frac{\operatorname{Id}}{2}\right)(\mathcal{N}\cdot \nabla f)\right\|\leq \lambda \|\mathcal{N}\cdot \nabla f\|
	\end{gather}
	which proves the claim.
\end{proof}
We are now in the position to prove \Cref{2T9}.
\begin{proof}[Proof of \Cref{2T9}]
	We first observe that the operator $w^{\operatorname{Tr}}_{\Omega}+\frac{\operatorname{Id}}{2}:\mathcal{D}\rightarrow\mathcal{D}$ is well-defined according to \Cref{4L2}. Then given $B\in L^2\mathcal{H}(\Omega)\cap \mathcal{H}^{\perp_{L^2(\Omega)}}_D(\Omega)$ we recall that $\mathcal{N}\cdot B\in \mathcal{D}$ which follows from the Hodge-decomposition theorem and the fact that $\mathcal{H}_N(\Omega)$ is $L^2(\Omega)$-orthogonal to $\mathcal{H}_D(\Omega)$ which follows from step 1 and 3 of the proof of \Cref{2T2}. Consequently $w^{\operatorname{Tr}}_{\Omega}-\frac{\operatorname{Id}}{2}=\left(w^{\operatorname{Tr}}_{\Omega}+\frac{\operatorname{Id}}{2}\right)-\operatorname{Id}$ also maps $\mathcal{D}$ into $\mathcal{D}$. So if we define $b_n:=\sum_{k=0}^{n}\left(\frac{\operatorname{Id}}{2}-w^{\operatorname{Tr}}_{\Omega}\right)^k(\mathcal{N}\cdot B)$ for fixed $n\in \mathbb{N}$ we find $b_n\in \mathcal{D}\subset W^{-\frac{1}{2},2}(\partial\Omega)$ for every $n$. It follows from definition of $\mathcal{D}$ that $\int_{\partial\Omega}b_nd\sigma=0$ for every $n$ and consequently the Neumann boundary value problem (\ref{2E9}) admits a unique solution $f_n\in H^1(\Omega)$ for any fixed $n$. We recall that if we let $\Gamma\in \mathcal{H}_N(\Omega)$ denote the $L^2(\Omega)$-orthogonal projection of $B$ onto $\mathcal{H}_N(\Omega)$ and if we let $f$ denote the unique solution to the BVP
	\begin{gather}
		\nonumber
		\Delta f=0\text{ in }\Omega\text{, }\mathcal{N}\cdot \nabla f=\left(\frac{\operatorname{Id}}{2}+w^{\operatorname{Tr}}_{\Omega}\right)^{-1}(\mathcal{N}\cdot B)\text{ on }\partial\Omega\text{ and }\int_{\Omega}fd^3x=0
	\end{gather}
	then according to \Cref{2L8} the current $j:=\Gamma\times \mathcal{N}+\nabla f\times \mathcal{N}$ is a preimage of $B$, i.e. $\operatorname{BS}_{\partial\Omega}(j)=B$. Further, we define the approximations $j_n:=\Gamma\times \mathcal{N}+\nabla f_n\times \mathcal{N}$ and we need to prove that
	\begin{gather}
		\label{4E13}
		\|j-j_n\|_{W^{-\frac{1}{2},2}(\partial\Omega)}\leq\frac{c_1\lambda^n}{1-\lambda}\|B\|_{L^2(\Omega)}\text{, }
		\\
		\label{4E14}
		\|\operatorname{BS}_{\partial\Omega}(j_n)-B\|_{L^2(\Omega)}\leq \frac{c_2\lambda^n}{1-\lambda}\|B\|_{L^2(\Omega)}
	\end{gather}
	for some constants $c_1,c_2>0$ which are independent of $B$ and $n$. We recall that $B=\operatorname{BS}_{\partial\Omega}(j)$ and therefore the continuity of the Biot-Savart operator, c.f. \cite[Lemma C.1]{G24}, implies that (\ref{4E14}) is an immediate consequence of (\ref{4E13}). We are hence left with establishing (\ref{4E13}). We observe that
	\begin{gather}
		\nonumber
		\|j-j_n\|_{W^{-\frac{1}{2},2}(\partial\Omega)}=\|\nabla f\times \mathcal{N}-\nabla f_n\times \mathcal{N}\|_{W^{-\frac{1}{2},2}(\partial\Omega)}\leq c\|\nabla f-\nabla f_n\|_{H(\operatorname{curl},\Omega)}=c\|\nabla f-\nabla f_n\|_{L^2(\Omega)}
	\end{gather}
	for a suitable $c>0$ which is independent of $f$ and $n$ by means of the continuity of the tangential trace. It now follows from \Cref{CT1} that
	\begin{gather}
		\nonumber
		\|\nabla f-\nabla f_n\|_{L^2(\Omega)}\leq \tilde{c}\|\mathcal{N}\cdot \nabla f-\mathcal{N}\cdot \nabla f_n\|_{W^{-\frac{1}{2},2}(\partial\Omega)}
		\\
		\nonumber
		=\tilde{c}\left\|\left(\frac{\operatorname{Id}}{2}+w^{\operatorname{Tr}}_{\Omega}\right)^{-1}(\mathcal{N}\cdot B)-\sum_{k=0}^n\left(\frac{\operatorname{Id}}{2}-w^{\operatorname{Tr}}_{\Omega}\right)^k(\mathcal{N}\cdot B)\right\|_{W^{-\frac{1}{2},2}(\partial\Omega)}
		\\
		\nonumber
		\leq \tilde{C}\left\|\left(\frac{\operatorname{Id}}{2}+w^{\operatorname{Tr}}_{\Omega}\right)^{-1}-\sum_{k=0}^n\left(\frac{\operatorname{Id}}{2}-w^{\operatorname{Tr}}_{\Omega}\right)^k\right\|\cdot \|\mathcal{N}\cdot B\|_{W^{-\frac{1}{2},2}(\partial\Omega)}
	\end{gather}
	where $\tilde{c},\tilde{C}>0$ are constants independent of $B$ and $n$, where $\|\cdot\|$ denotes the operator norm induced by the inner product (\ref{4E8}) and where we used the equivalence of the norm $\|\cdot \|$ induced by (\ref{4E8}) and the $W^{-\frac{1}{2},2}(\partial\Omega)$-norm. We can then use the continuity of the normal trace to conclude that $\|\mathcal{N}\cdot B\|_{W^{-\frac{1}{2},2}(\partial\Omega)}\leq \hat{c}\|B\|_{L^2(\Omega)}$ (keeping in mind that $\operatorname{div}(B)=0$) for some suitable $\hat{c}>0$ independent of $B$. We overall arrive at
	\begin{gather}
		\label{4E15}
		\|j-j_n\|_{W^{-\frac{1}{2},2}(\partial\Omega)}\leq c_1\left\|\left(\frac{\operatorname{Id}}{2}+w^{\operatorname{Tr}}_{\Omega}\right)^{-1}-\sum_{k=0}^n\left(\frac{\operatorname{Id}}{2}-w^{\operatorname{Tr}}_{\Omega}\right)^k\right\|\cdot \|B\|_{L^2(\Omega)}
	\end{gather}
	for some $c_1>0$ independent of $B$ and $n$. We now write $\frac{\operatorname{Id}}{2}+w^{\operatorname{Tr}}_{\Omega}=\operatorname{Id}-\left( \frac{\operatorname{Id}}{2}-w^{\operatorname{Tr}}_{\Omega}\right)$. According to \Cref{4C4} we see that $\|\frac{\operatorname{Id}}{2}-w^{\operatorname{Tr}}_{\Omega}\|\leq \lambda<1$ and thus the inverse of $\frac{\operatorname{Id}}{2}+w^{\operatorname{Tr}}_{\Omega}$ admits a Neumann series expression as
	\begin{gather}
		\nonumber
		\left(\frac{\operatorname{Id}}{2}+w^{\operatorname{Tr}}_{\Omega}\right)^{-1}=\sum_{k=0}^{\infty}\left(\frac{\operatorname{Id}}{2}-w^{\operatorname{Tr}}_{\Omega}\right)^k
	\end{gather}
	and we find the estimate
	\begin{gather}
		\nonumber
		\left\|\left(\frac{\operatorname{Id}}{2}+w^{\operatorname{Tr}}_{\Omega}\right)^{-1}-\sum_{k=0}^n\left(\frac{\operatorname{Id}}{2}-w^{\operatorname{Tr}}_{\Omega}\right)^k\right\|\leq \sum_{k={n+1}}^{\infty}\lambda^k=\frac{\lambda^{n+1}}{1-\lambda}
	\end{gather}
	which in combination with (\ref{4E15}) proves the theorem.
\end{proof}
\subsection{Proof of \Cref{2P10}}
\begin{proof}[Proof of \Cref{2P10}]
	We recall that we are given a $C^{1,1}$-solid torus $\Omega\subset\mathbb{R}^3$ and we need to prove that for every $\Gamma\in \mathcal{H}_N(\Omega)$ and $\nabla f\in \mathcal{H}_{\operatorname{ex}}(\Omega)$ we have $J:=\Gamma\times \mathcal{N}+\nabla f\times \mathcal{N}\in W^{-\frac{1}{2},2}\mathcal{V}_0(\Sigma)$, where $\Sigma:=\partial\Omega$, and $\overline{Q}(J)=0$ where $\overline{Q}$ is defined in (\ref{2E12}). The fact that $J\in W^{-\frac{1}{2},2}\mathcal{V}_0(\Sigma)$ follows immediately from \Cref{3L4} since $\operatorname{curl}(\Gamma)=0=\operatorname{curl}(\nabla f)$. To see that $\overline{Q}(J)=0$ we recall that we have to show that
	\begin{gather}
		\nonumber
		\int_{\Sigma}J\cdot \gamma_td\sigma=0
	\end{gather}
	where $\gamma_t\in \mathcal{H}(\Sigma)=\{\gamma\in L^2\mathcal{V}(\Sigma)\mid \operatorname{div}_{\Sigma}(\gamma)=0=\operatorname{curl}_{\Sigma}(\gamma)\}$ is uniquely determined by the conditions $\int_{\sigma_t}\gamma_t=1$ and $\int_{\sigma_p}\gamma_t=0$ where $\sigma_t$ and $\sigma_p$ are some fixed toroidal and poloidal closed curve respectively. We observe first that $\mathcal{H}_N(\Omega)\subset \bigcap_{1<p<\infty}W^{1,p}(\Omega,\mathbb{R}^3)\subset \bigcap_{0<\alpha<1}C^{0,\alpha}(\overline{\Omega},\mathbb{R}^3)$, c.f. \cite[Lemma A.1]{G24}. Now, since $\sigma_p$ is poloidal it bounds a disc $D\subset \Omega$ and we may fix any $\widetilde{\Gamma}\in \mathcal{H}_N(\Omega)\setminus \{0\}$ and compute by means of Stokes' theorem $\int_{\sigma_p}\widetilde{\Gamma}=\int_D\operatorname{curl}(\widetilde{\Gamma})\cdot \mathcal{N}d\sigma=0$. It further follows from the fact that $\operatorname{curl}(\widetilde{\Gamma})=0$ that the restriction $\widetilde{\Gamma}|_{\Sigma}$ can be expressed, by means of the Hodge decomposition theorem, as $\widetilde{\Gamma}|_{\Sigma}=\nabla_{\Sigma}\kappa+\gamma$ for suitable $\gamma\in \mathcal{H}(\Sigma)$ and $\kappa\in H^1(\Sigma)$. It follows from the regularity of $\widetilde{\Gamma}$ and the regularity of $\mathcal{H}(\Sigma)\subset \bigcap_{1\leq p<\infty}W^{1,p}(\Sigma)\subset \bigcap_{0<\alpha<1}C^{0,\alpha}(\Sigma)$ that we also have $\kappa\in \bigcap_{0< \alpha<1} C^{1,\alpha}(\Sigma)$ so that $\nabla_{\Sigma}\kappa$ and $\gamma$ admit well-defined line-integrals and that we in particular have $\int_{\sigma_p}\nabla_{\Sigma}\kappa=0$ since $\sigma_p$ is a closed curve. Consequently $\int_{\sigma_p}\gamma=0$. Since $\sigma_p$ and $\sigma_t$ form a set of generators of the first fundamental group of $\Sigma$ we must have $\int_{\sigma_t}\gamma\neq 0$, since otherwise $\int_{\sigma}\gamma=0$ for any closed curve $\sigma\subset \Sigma$ which would imply that $\gamma$ is a gradient field and hence must be identically zero, which in turn would imply that $\widetilde{\Gamma}|_{\Sigma}$ is a gradient field which is only the case if $\widetilde{\Gamma}=0$ since $\mathcal{H}_N(\Omega)$ is $L^2(\Omega)$-orthogonal to the gradient fields, is curl-free and admits a vector potential. We conclude that with the right scaling we have $\int_{\sigma_t}\gamma=1$ and $\int_{\sigma_p}\gamma=0$, i.e. $\gamma=\gamma_t$. Further, we conclude that $\widetilde{\Gamma}|_{\Sigma}$ and $\gamma_t$ only differ by a gradient field $\nabla_{\Sigma} \kappa\in \bigcap_{0< \alpha<1}C^{0,\alpha}(\Sigma)=\bigcap_{1<p<\infty}W^{1-\frac{1}{p},p}(\Sigma)$. Since $J\in W^{-\frac{1}{2},2}\mathcal{V}_0(\Sigma)$ is div-free, it follows by an approximation argument that $\int_{\Sigma}J\cdot \nabla_{\Sigma}\kappa d\sigma=0$. Consequently, we find $\int_{\Sigma}J\cdot \gamma_t d\sigma=\int_{\Sigma}J\cdot \widetilde{\Gamma} d\sigma=0$ where we used that $(\Gamma\times \mathcal{N})\cdot \widetilde{\Gamma}=0$ on all of $\Sigma$ because $\mathcal{H}_N(\Omega)$ is $1$-dimensional and that $\nabla f\times \mathcal{N}$ is co-exact, while $\widetilde{\Gamma}|_{\Sigma}$ is a closed field, so that $\int_{\Sigma}(\nabla f\times \mathcal{N})\cdot \widetilde{\Gamma}d\sigma=0$. It then follows from definition of $\overline{Q}$ that $\overline{Q}(J)=0$ as desired.
\end{proof}
\section{Kernel reconstruction algorithm}
\subsection{Proof of \Cref{2T11}}
\begin{proof}[Proof of \Cref{2T11}]
	It follows from the proof of \cite[Proposition 5.8]{G24}, see also \cite[Equation (6.1)]{G24} that if we fix a basis $\Gamma_1,\dots,\Gamma_n$ of $\mathcal{H}_N(\Omega)$, then for every $1\leq i\leq n$ there exists a function $f_i\in \bigcap_{0<\alpha<1}C^{1,\alpha}(\partial\Omega)$ which satisfies the equation
	\begin{gather}
		\label{5E1}
		\int_{\partial\Omega} f_i(y)\mathcal{N}(y)\cdot\frac{y-x}{|x-y|^3}d\sigma(y)=\int_{\partial\Omega}\frac{\operatorname{BS}_{\Omega}(\Gamma_i)(y)\cdot \mathcal{N}(y)}{|x-y|}d\sigma(y)\text{ for all }x\in \Omega
	\end{gather}
	and that the vector fields 
	\begin{gather}
		\label{5E2}
j_i:=\operatorname{BS}_{\Omega}(\Gamma_i)\times \mathcal{N}+\nabla f_i\times \mathcal{N}\in W^{-\frac{1}{2},2}\mathcal{V}_0(\partial\Omega)
\end{gather}
	provide a basis of $\operatorname{Ker}(\operatorname{BS}_{\partial\Omega})$. We can then let $\tilde{f}_i\in \bigcap_{1\leq p<\infty}W^{2,p}(\Omega)$ be the harmonic extension of the $f_i$, i.e. $\Delta \widetilde{f}_i=0$ in $\Omega$ and $\widetilde{f}_i=f_i$ on $\partial\Omega$, c.f. \cite[Theorem 2.4.2.5]{Gris85}. We compute first for fixed $x\in \Omega$
	\begin{gather}
		\nonumber
		\int_{\partial\Omega}f_i(y)\mathcal{N}(y)\cdot\frac{y-x}{|x-y|^3}d\sigma(y)=\int_{\partial\Omega}\widetilde{f}_i(y)\mathcal{N}(y)\cdot\frac{y-x}{|x-y|^3}d\sigma(y)=\int_{\Omega}\nabla \widetilde{f}_i(y)\cdot \frac{y-x}{|y-x|^3}d^3y+4\pi \widetilde{f}_i(x).
	\end{gather}
	Further, using that $\operatorname{BS}_{\Omega}(\Gamma_i)$ is div-free, we find $\int_{\partial\Omega}\frac{\operatorname{BS}_{\Omega}(\Gamma_i)(y)\cdot\mathcal{N}(y)}{|x-y|}d\sigma(y)=\int_{\Omega}\operatorname{BS}_{\Omega}(\Gamma_i)(y)\cdot \frac{x-y}{|x-y|^3}d^3y$. We insert this into (\ref{5E1}) and obtain
	\begin{gather}
		\label{5E3}
		\widetilde{f}_i(x)=\frac{1}{4\pi}\int_{\Omega}\nabla \widetilde{f}_i(y)\cdot\frac{x-y}{|x-y|^3}d^3y+\frac{1}{4\pi}\int_{\Omega}\operatorname{BS}_{\Omega}(\Gamma_i)(y)\cdot \frac{x-y}{|x-y|^3}d^3y\text{ for all }x\in \Omega.
	\end{gather}
	We observe now that $\frac{x-y}{|x-y|^3}=\nabla_y\frac{1}{|x-y|}$ and that $\operatorname{BS}_{\Omega}(\Gamma_i)$ is div-free. Therefore it follows from the Hodge-decomposition theorem \cite[Theorem B.1]{G24} that, if we let $Z_i$ denote the $L^2(\Omega)$-orthogonal projection of $\operatorname{BS}_{\Omega}(\Gamma_i)$ onto $\mathcal{H}_{\operatorname{ex}}(\Omega)$, we find $\int_{\Omega}\operatorname{BS}_{\Omega}(\Gamma_i)(y)\cdot \frac{x-y}{|x-y|^3}d^3y=\int_{\Omega}Z_i(y)\cdot \frac{x-y}{|x-y|^3}d^3y$. Further, it follows from (\ref{2E16}) and the symmetry of the volume Biot-Savart operator that $\int_{\Omega}\nabla f\cdot \operatorname{BS}_{\Omega}(\Gamma_i)d^3y=\int_{\Omega}\Gamma_i\cdot \operatorname{BS}_{\Omega}(\nabla f)d^3y=0$ for every $\nabla f\in \mathcal{H}_D(\Omega)$. We conclude that $\operatorname{BS}_{\Omega}(\Gamma_i)\in \mathcal{H}^{\perp_{L^2(\Omega)}}_D(\Omega)$ and consequently $Z_i\in \mathcal{H}_{\operatorname{ex}}(\Omega)\cap \mathcal{H}^{\perp_{L^2(\Omega)}}_D(\Omega)$. Further, we observe that we may assume that $\nabla \widetilde{f}_i\in \mathcal{H}^{\perp_{L^2(\Omega)}}_D(\Omega)$. To see this we may label as usual the boundary components of $\partial\Omega$ by $\partial\Omega_0$,$\partial\Omega_1,\dots$,$\partial\Omega_M$ for suitable $M\in \mathbb{N}_0$ and observe that the gradients of the following (unique) functions
	\begin{gather}
		\label{5E4}
		\Delta h_k=0\text{ in }\Omega\text{, }h_k|_{\partial\Omega_i}=\delta_{ki}
	\end{gather} 
	for $k=1,\dots,M$ form a basis of $\mathcal{H}_D(\Omega)$. We further observe that, letting $\Omega_k$ denote the finite volumes enclosed by the $\partial\Omega_k$ and taking $\partial\Omega_0$ as the unique boundary component with $\Omega\subset \Omega_0$, we find
	\begin{gather}
		\nonumber
		\int_{\partial\Omega}h_k(y)\mathcal{N}(y)\cdot \frac{y-x}{|x-y|^3}d\sigma(y)=\int_{\partial\Omega_k}\mathcal{N}(y)\cdot \frac{y-x}{|y-x|^3}d\sigma(y)=0
	\end{gather}
	for every fixed $x\in \Omega$ because $\Omega\cap \Omega_k=\emptyset$ for every $1\leq k\leq M$. Hence, we may subtract the projection of $\nabla \widetilde{f}_i$ onto $\mathcal{H}_D(\Omega)$ from $\nabla \widetilde{f}_i$ which will lead to a possibly modified function $f_i$ still solving (\ref{5E1}) with a harmonic extension $\widetilde{f}_i$ satisfying $\nabla \widetilde{f}_i\in \mathcal{H}^{\perp_{L^2(\Omega)}}_D(\Omega)$. The key now is that according to \cite[Lemma 6.3]{G24} any two solutions $f_i$,$\hat{f}_i$ of (\ref{5E1}) lead to the same current in (\ref{5E2}).
	
	We can then take the gradient in (\ref{5E3}) and observe the identity
	\begin{gather}
		\nonumber
		\nabla \widetilde{f}_i=T(\nabla \widetilde{f}_i)+T(Z_i).
	\end{gather}
	We then apply \Cref{4L1} to conclude
	\begin{gather}
		\nonumber
		\mathcal{N}\cdot \nabla \widetilde{f}_i=\frac{\mathcal{N}\cdot \nabla \widetilde{f}_i}{2}-w^{\operatorname{Tr}}_{\Omega}(\mathcal{N}\cdot \nabla \widetilde{f}_i)+\left(\frac{\operatorname{Id}}{2}-w^{\operatorname{Tr}}_{\Omega}\right)(\mathcal{N}\cdot Z_i)
		\\
		\nonumber
		\Leftrightarrow \left(\frac{\operatorname{Id}}{2}+w^{\operatorname{Tr}}_{\Omega}\right)(\mathcal{N}\cdot \nabla \widetilde{f}_i)=\left(\frac{\operatorname{Id}}{2}-w^{\operatorname{Tr}}_{\Omega}\right)(\mathcal{N}\cdot Z_i).
	\end{gather}
	Since $\mathcal{N}\cdot \nabla \widetilde{f}_i$, $\mathcal{N}\cdot Z_i\in \mathcal{D}=\left\{\mathcal{N}\cdot \nabla f\mid \nabla f\in \mathcal{H}_{\operatorname{ex}}(\Omega)\cap \mathcal{H}^{\perp_{L^2(\Omega)}}_D(\Omega)\right\}$ we may invert the operator $\frac{\operatorname{Id}}{2}+w^{\operatorname{Tr}}_{\Omega}$, recall \Cref{4L2}, and hence we see that $\widetilde{f}_i$ solves the Neumann problem
	\begin{gather}
		\nonumber
		\Delta \widetilde{f}_i=0\text{ in }\Omega\text{ and }\mathcal{N}\cdot \nabla\widetilde{f}_i=\left(\frac{\operatorname{Id}}{2}+w^{\operatorname{Tr}}_{\Omega}\right)^{-1}\left(\frac{\operatorname{Id}}{2}-w^{\operatorname{Tr}}_{\Omega}\right)(\mathcal{N}\cdot Z_i)\text{ on }\partial\Omega.
	\end{gather}
	Finally, we observe that $\mathcal{N}\cdot Z_i=\mathcal{N}\cdot \operatorname{BS}_{\Omega}(\Gamma_i)$ because the remaining components of the Hodge decomposition of $\operatorname{BS}_{\Omega}(\Gamma_i)$ are tangent to the boundary, \cite[Theorem B.1]{G24}. We conclude that there exist solutions of the BVPs
	\begin{gather}
		\nonumber
		\Delta g_i=0\text{ in }\Omega\text{, }\mathcal{N}\cdot \nabla g_i=\left(\frac{\operatorname{Id}}{2}+w^{\operatorname{Tr}}_{\Omega}\right)^{-1}\left(\frac{\operatorname{Id}}{2}-w^{\operatorname{Tr}}_{\Omega}\right)(\mathcal{N}\cdot \operatorname{BS}_{\Omega}(\Gamma_i))\text{ on }\partial\Omega\text{ and }\int_{\Omega}g_id^3x=0
	\end{gather}
	and that $\nabla g_i=\nabla \widetilde{f}_i$ so that according to (\ref{5E2}) the vector fields $j_i:=\operatorname{BS}_{\Omega}(\Gamma_i)\times \mathcal{N}+\nabla g_i\times \mathcal{N}$ form a basis of $\operatorname{Ker}(\operatorname{BS}_{\partial\Omega})$ as claimed.
\end{proof}
\subsection{Proof of \Cref{2T12}}
\begin{proof}[Proof of \Cref{2T12}]
	We recall first that according to \Cref{2T11} and its proof we have for any fixed $\Gamma\in \mathcal{H}_N(\Omega)$, $\mathcal{N}\cdot \operatorname{BS}_{\Omega}(\Gamma)\in \mathcal{D}=\{\mathcal{N}\cdot \nabla f\mid \nabla f\in \mathcal{H}_{\operatorname{ex}}(\Omega)\cap \mathcal{H}^{\perp_{L^2(\Omega)}}_D(\Omega)\}$ and $j:=\operatorname{BS}_{\Omega}(\Gamma)\times \mathcal{N}+\nabla f\times \mathcal{N}\in \operatorname{Ker}(\operatorname{BS}_{\partial\Omega})$ where $f$ is the unique solution of the BVP
	\begin{gather}
		\nonumber
		\Delta f=0\text{ in }\Omega\text{, }\mathcal{N}\cdot \nabla f=\left(\frac{\operatorname{Id}}{2}+w^{\operatorname{Tr}}_{\Omega}\right)^{-1}\left(\left(\frac{\operatorname{Id}}{2}-w^{\operatorname{Tr}}_{\Omega}\right)(\mathcal{N}\cdot \operatorname{BS}_{\Omega}(\Gamma))\right)\text{ on }\partial\Omega\text{, }\int_{\Omega}fd^3x=0.
	\end{gather}
	Now, since $\mathcal{N}\cdot \operatorname{BS}_{\Omega}(\Gamma)\in \mathcal{D}$ it follows that for any fixed $n\in \mathbb{N}$, $b_n:=\sum_{k=1}^{n}\left(\frac{\operatorname{Id}}{2}-w^{\operatorname{Tr}}_{\Omega}\right)^k(\mathcal{N}\cdot \operatorname{BS}_{\Omega}(\Gamma))\in \mathcal{D}$ and consequently $\int_{\partial\Omega}b_nd\sigma=0$ so that each $b_n$ satisfies the compatibility condition for the existence of a solution $f_n\in H^1(\Omega)$ of the corresponding Neumann problem
	\begin{gather}
		\nonumber
		\Delta f_n=0\text{ in }\Omega\text{, }\mathcal{N}\cdot \nabla f_n=b_n\text{ on }\partial\Omega\text{ and }\int_{\partial\Omega}f_nd^3x=0.
	\end{gather}
	One can argue now in the spirit of the proof of \Cref{2T9}, namely letting $j_n:=\operatorname{BS}_{\Omega}(\Gamma)\times \mathcal{N}+\nabla f_n\times \mathcal{N}$ we find
	\begin{gather}
		\nonumber
		\|j-j_n\|_{W^{-\frac{1}{2},2}(\partial\Omega)}=\|\nabla f\times \mathcal{N}-\nabla f_n\times \mathcal{N}\|_{W^{-\frac{1}{2},2}(\partial\Omega)}\leq c\|\nabla f-\nabla f_n\|_{H(\operatorname{curl},\Omega)}
		\\
		\nonumber
		=c\|\nabla f-\nabla f_n\|_{L^2(\Omega)}\leq C\|\mathcal{N}\cdot\nabla f-\mathcal{N}\cdot \nabla f_n\|_{W^{-\frac{1}{2},2}(\partial\Omega)}
	\end{gather}
	where we used once more \Cref{CT1} in the last step and $c,C>0$ are constants independent of $f$ and $n$. Now we can use the boundary conditions satisfied by $\nabla f,\nabla f_n$
	\begin{gather}
		\nonumber
		\|\mathcal{N}\cdot\nabla f-\mathcal{N}\cdot \nabla f_n\|_{W^{-\frac{1}{2},2}(\partial\Omega)}
		\\
		\nonumber
		=\left\|\left(\frac{\operatorname{Id}}{2}+w^{\operatorname{Tr}}_{\Omega}\right)^{-1}\left(\frac{\operatorname{Id}}{2}-w^{\operatorname{Tr}}_{\Omega}\right)(\mathcal{N}\cdot \operatorname{BS}_{\Omega}(\Gamma))-\sum_{k=1}^n\left(\frac{\operatorname{Id}}{2}-w^{\operatorname{Tr}}_{\Omega}\right)^k(\mathcal{N}\cdot \operatorname{BS}_{\Omega}(\Gamma))\right\|_{W^{-\frac{1}{2},2}(\partial\Omega)}
		\\
		\nonumber
		\leq \left\|\left(\frac{\operatorname{Id}}{2}+w^{\operatorname{Tr}}_{\Omega}\right)^{-1}\circ \left(\frac{\operatorname{Id}}{2}-w^{\operatorname{Tr}}_{\Omega}\right)-\sum_{k=1}^n\left(\frac{\operatorname{Id}}{2}-w^{\operatorname{Tr}}_{\Omega}\right)^k\right\|_{W^{-\frac{1}{2},2}(\partial\Omega)}\|\mathcal{N}\cdot \operatorname{BS}_{\Omega}(\Gamma)\|_{W^{-\frac{1}{2},2}(\partial\Omega)}
		\\
		\nonumber
		\leq c\left\|\left(\frac{\operatorname{Id}}{2}+w^{\operatorname{Tr}}_{\Omega}\right)^{-1}\circ \left(\frac{\operatorname{Id}}{2}-w^{\operatorname{Tr}}_{\Omega}\right)-\sum_{k=1}^n\left(\frac{\operatorname{Id}}{2}-w^{\operatorname{Tr}}_{\Omega}\right)^k\right\|\|\operatorname{BS}_{\Omega}(\Gamma)\|_{L^2(\Omega)}
	\end{gather}
	for some suitable constant $c>0$ independent of $\Gamma$ and $n$, where we used the continuity of the normal trace with respect to $H(\operatorname{div},\Omega)$-norm, that $\operatorname{BS}_{\Omega}(\Gamma)$ is div-free and the fact that the norm $\|\cdot \|$ induced by the inner product in (\ref{4E8}) is equivalent to the $W^{-\frac{1}{2},2}(\partial\Omega)$-norm. As in the proof of \Cref{2T9} we have the expression $\left(\frac{\operatorname{Id}}{2}+w^{\operatorname{Tr}}_{\Omega}\right)^{-1}=\sum_{k=0}^{\infty}\left(\frac{\operatorname{Id}}{2}-w^{\operatorname{Tr}}_{\Omega}\right)^k$ and so utilising \Cref{4C4} we can estimate
	\begin{gather}
		\nonumber
		\left\|\left(\frac{\operatorname{Id}}{2}+w^{\operatorname{Tr}}_{\Omega}\right)^{-1}\circ \left(\frac{\operatorname{Id}}{2}-w^{\operatorname{Tr}}_{\Omega}\right)-\sum_{k=1}^n\left(\frac{\operatorname{Id}}{2}-w^{\operatorname{Tr}}_{\Omega}\right)^k\right\|\leq \frac{\lambda^{n+1}}{1-\lambda}
	\end{gather}
	where $0<\lambda<1$ is the contraction constant of the operator $T$.
	
	We hence arrive at the inequality
	\begin{gather}
		\nonumber
		\|j_n-j\|_{W^{-\frac{1}{2},2}(\partial\Omega)}\leq \tilde{c}\frac{\lambda^{n+1}}{1-\lambda}\|\operatorname{BS}_{\Omega}(\Gamma)\|_{L^2(\Omega)}
	\end{gather}
	where $\tilde{c}>0$ is a constant independent of $n$ and $\Gamma$. To obtain the desired estimate we note that by means of the H\"{o}lder inequality and Hardy-Littlewood-Sobolev inequality we can estimate
	\begin{gather}
		\nonumber
		\|\operatorname{BS}_{\Omega}(\Gamma)\|_{L^2(\Omega)}\leq c(\Omega)\|\operatorname{BS}_{\Omega}(\Gamma)\|_{L^6(\Omega)}\leq c(\Omega)\|\operatorname{BS}_{\Omega}(\Gamma)\|_{L^6(\mathbb{R}^3)}\leq \tilde{c}\|\Gamma\|_{L^2(\Omega)}
	\end{gather}
	once more for suitable constants $c,\tilde{c}>0$ which are independent of $\Gamma$. We conclude overall
	\begin{gather}
		\nonumber
		\|j_n-j\|_{W^{-\frac{1}{2},2}(\partial\Omega)}\leq \widetilde{C}\frac{\lambda^{n+1}}{1-\lambda}\|\Gamma\|_{L^2(\Omega)}
	\end{gather}
	for a suitable $\widetilde{C}>0$ independent of $n$ and $\Gamma$. This proves (\ref{2E20}).
	
	The estimate of $\operatorname{BS}_{\partial\Omega}(j_n)$ follows from the continuity of the Biot-Savart operator \cite[Lemma C.1]{G24} and the fact that $\operatorname{BS}_{\partial\Omega}(j)=0$ since $j\in \operatorname{Ker}(\operatorname{BS}_{\partial\Omega})$.
\end{proof}
\subsection{Proof of \Cref{2T13}}
\begin{proof}[Proof of \Cref{2T13}]
	We will prove that the traces $\nabla g_i\times \mathcal{N}$ of the functions $g_i$ from \Cref{2T11} may be equivalently characterised as the traces $\nabla f_i\times \mathcal{N}$ of the functions $f_i$ solving the exterior boundary value problems
	\begin{gather}
		\label{5E5}
		\Delta f_i=0\text{ in }\mathbb{R}^3\setminus \overline{\Omega}\text{, }\mathcal{N}\cdot \nabla f_i=-\mathcal{N}\cdot \operatorname{BS}_{\Omega}(\Gamma_i)\text{ on }\partial\Omega\text{, }f_i\rightarrow 0\text{ as }x\rightarrow\infty.
	\end{gather}
	Then the theorem will follow immediately from the characterisation of the kernel in \Cref{2T11}. To this end we fix a basis $\Gamma_1,\dots,\Gamma_n$ of $\mathcal{H}_N(\Omega)$ as in \Cref{2T11} and fix from now on some $1\leq i\leq n$ and drop for notational simplicity the index, i.e. we set $\Gamma\equiv \Gamma_i$ for our fixed index $i$.  It follows then first from the proof of \cite[Proposition 5.8]{G24} that $\operatorname{BS}_{\Omega}(\Gamma)|_{\partial\Omega}\in \bigcap_{0<\alpha<1}C^{1,\alpha}(\partial\Omega)$. Denoting as usual the boundary components of $\Omega$ by $\partial\Omega_0,\partial\Omega_1,\dots,\partial\Omega_m$ for suitable $m\in \mathbb{N}_0$ where $\partial\Omega_0$ is the unique component which encloses a finite volume $\Omega_0$ with $\Omega\subset\Omega_0$. The remaining components will also enclose finite regions $\Omega_i$ and it follows from the div-theorem and the fact that $\operatorname{BS}_{\Omega}(\Gamma)$ is div-free throughout all of $\mathbb{R}^3$ that $\int_{\partial\Omega_i}\mathcal{N}\cdot \operatorname{BS}_{\Omega}(\Gamma)d\sigma=0$ for all $0\leq i\leq m$. It then follows from \cite[Theorem 6.43]{RCM21} that there exists a solution $\phi\in \bigcap_{0<\alpha<1}C^{0,\alpha}(\partial\Omega)$ of the equation
	\begin{gather}
		\label{5E6}
		\left(\frac{\operatorname{Id}}{2}+w^{\operatorname{Tr}}_{\Omega}\right)(\phi)=\mathcal{N}\cdot \operatorname{BS}_{\Omega}(\Gamma)\text{ on }\partial\Omega
	\end{gather}
	and that the single layer potential
	\begin{gather}
		\label{5E7}
		v_{\Omega^c}[\phi](x):=\frac{1}{4\pi}\int_{\partial\Omega}\frac{\phi(y)}{|x-y|}d\sigma(y)\in \bigcap_{0<\alpha<1}C^{1,\alpha}_{\operatorname{loc}}(\mathbb{R}^3\setminus \Omega)
	\end{gather}
	provides a weak solution of the BVP (\ref{5E5}). It follows further from \cite[Theorem 6.43]{RCM21} that the difference of any other solution of (\ref{5E5}) and $v_{\Omega^c}[\phi]$ is locally constant and is identically zero on the unbounded component of $\overline{\Omega}^c$. Therefore, fixing the averages on each $\Omega_i$, $i=1,\dots,m$ provides a unique solution to the BVP (\ref{5E5}) and further the gradient of this solution coincides with $\nabla_xv_{\Omega^c}[\phi](x)$. Since the constructed currents $j_i$ only depend on the gradient of the BVP (\ref{5E5}) we may work with the solutions given by the single layer potential $v_{\Omega^c}[\phi](x)$. We now recall that by the characterisation of the kernel in \Cref{2T11} we want to show that $\nabla g\times \mathcal{N}=\nabla v_{\Omega^c}[\phi]\times \mathcal{N}$ on $\partial\Omega$ where $g$ is the unique solution to the interior Neumann BVP
	\begin{gather}
		\label{5E8}
		\Delta g=0\text{ in }\Omega\text{, }\mathcal{N}\cdot \nabla g=\left(\frac{\operatorname{Id}}{2}+w^{\operatorname{Tr}}_{\Omega}\right)^{-1}\left(\left(\frac{\operatorname{Id}}{2}-w^{\operatorname{Tr}}_{\Omega}\right)(\operatorname{BS}_{\Omega}(\Gamma)\cdot \mathcal{N})\right)\text{ on }\partial\Omega\text{ and }\int_{\Omega}gd^3x=0.
	\end{gather}
	We observe that $\frac{\operatorname{Id}}{2}+w^{\operatorname{Tr}}_{\Omega}$ and $\frac{\operatorname{Id}}{2}-w^{\operatorname{Tr}}_{\Omega}$ commute and hence so do $\left(\frac{\operatorname{Id}}{2}+w^{\operatorname{Tr}}_{\Omega}\right)^{-1}$ and $\frac{\operatorname{Id}}{2}-w^{\operatorname{Tr}}_{\Omega}$. With this observation we can express the boundary conditions in (\ref{5E8}) equivalently as
	\begin{gather}
		\label{5E9}
		\mathcal{N}\cdot \nabla g=\left(\frac{\operatorname{Id}}{2}-w^{\operatorname{Tr}}_{\Omega}\right)\left(\left(\frac{\operatorname{Id}}{2}+w^{\operatorname{Tr}}_{\Omega}\right)^{-1}(\operatorname{BS}_{\Omega}(\Gamma)\cdot \mathcal{N})\right).
	\end{gather}
	We further conclude from (\ref{5E6})
	\begin{gather}
		\label{5E10}
		\left(\frac{\operatorname{Id}}{2}+w^{\operatorname{Tr}}_{\Omega}\right)\left(\phi-\left(\frac{\operatorname{Id}}{2}+w^{\operatorname{Tr}}_{\Omega}\right)^{-1}(\operatorname{BS}_{\Omega}(\Gamma)\cdot \mathcal{N})\right)=0.
	\end{gather}
	We claim now that $\phi$ in (\ref{5E6}) can be chosen such that $\phi \in \mathcal{D}=\left\{\mathcal{N}\cdot \nabla f\mid \nabla f\in \mathcal{H}_{\operatorname{ex}}(\Omega)\cap \mathcal{H}_D^{\perp_{L^2(\Omega)}}(\Omega)\right\}$. To see this, we note first that according to \cite[Lemma 6.11]{RCM21}
	\begin{gather}
		\nonumber
		\int_{\partial\Omega}\phi d\sigma=\int_{\partial\Omega}\left(\frac{\operatorname{Id}}{2}+w^{\operatorname{Tr}}_{\Omega}\right)(\phi)d\sigma=\int_{\partial\Omega}\mathcal{N}\cdot \operatorname{BS}_{\Omega}(\Gamma)d\sigma=0
	\end{gather}
	where we used that $\operatorname{BS}_{\Omega}(\Gamma)$ is div-free in the last step. Hence there is some $\nabla f\in \mathcal{H}_{\operatorname{ex}}(\Omega)$ with $\mathcal{N}\cdot \nabla f=\phi$. Further, it follows from \cite[Lemma 6.28]{RCM21} that $\phi$ can be chosen such that in addition $\int_{\partial\Omega_i}\phi d\sigma=0$ for all $1\leq i\leq m$. For this specific choice of $\phi$ we can then decompose further by finite dimensionality of $\mathcal{H}_D(\Omega)$, $\nabla f=\nabla h+\nabla \kappa$ where $\nabla h\in \mathcal{H}_{\operatorname{ex}}(\Omega)\cap \mathcal{H}_D^{\perp_{L^2(\Omega)}}(\Omega)$ and $\nabla \kappa\in \mathcal{H}_D(\Omega)$. According to (\ref{5E4}) we can find constants $\alpha_j\in \mathbb{R}$, $1\leq j\leq m$ with $\nabla \kappa=\sum_{j=1}^m\alpha_j \nabla \kappa_j$ where the $\kappa_j$ are given as the solutions of the BVPs in (\ref{5E4}). Using that the decomposition $\nabla f=\nabla h+\nabla \kappa$ is $L^2(\Omega)$-orthogonal, that $\mathcal{N}\cdot \nabla f=\phi$ and the properties of $\phi$ one concludes easily by performing an integration by parts that $\|\nabla h\|^2_{L^2(\Omega)}=0$ and consequently $\phi \in \mathcal{D}$ as claimed. It hence follows from (\ref{5E10}) and \Cref{4L2} that $\phi=\left(\frac{\operatorname{Id}}{2}+w^{\operatorname{Tr}}_{\Omega}\right)^{-1}(\operatorname{BS}_{\Omega}(\Gamma)\cdot \mathcal{N})$ and consequently (\ref{5E9}) reads $\mathcal{N}\cdot \nabla g=\left(\frac{\operatorname{Id}}{2}-w^{\operatorname{Tr}}_{\Omega}\right)(\phi)$. It then follows from the representation formula for the interior Neumann problem, c.f. \cite[Theorem 6.42]{RCM21}, that a solution to the BVP (\ref{5E8}) with a possibly non-zero mean within $\Omega$ is given by 
	\begin{gather}
		\label{5E11}
		v_{\Omega}[\phi](x):=\frac{1}{4\pi}\int_{\partial\Omega}\frac{\phi(y)}{|x-y|}d\sigma(y)\in \bigcap_{0<\alpha<1}C^{1,\alpha}(\overline{\Omega}).
	\end{gather}
	Again, the gradient of the solution to (\ref{5E8}) coincides with $\nabla v_{\Omega}[\phi]$. We observe further that $v_{\Omega^c}[\phi]$,$v_{\Omega}[\phi]$ as defined in (\ref{5E7}) and (\ref{5E11}) give rise to well-defined continuous functions on all of $\mathbb{R}^3$ so that their traces on $\partial\Omega$ coincide. This implies $\nabla_{\partial\Omega} v_{\Omega}[\phi]=\nabla_{\partial\Omega} v_{\Omega^c}[\phi]$ on $\partial\Omega$ and consequently $\nabla v_{\Omega}[\phi]\times \mathcal{N}=\nabla v_{\Omega^c}[\phi]\times \mathcal{N}$ on $\partial\Omega$ since only the tangential gradients contribute to these expressions. We conclude that $\nabla f\times \mathcal{N}=\nabla g\times \mathcal{N}$ for any solution $f$ of the BVP (\ref{5E5}) and any solution $g$ of the BVP (\ref{5E8}) which proves the theorem.
\end{proof}
\section*{Acknowledgements}
This work has been supported by the Inria AEX StellaCage. The research was supported in part by the MIUR Excellence Department Project awarded to Dipartimento di Matematica, Università di Genova, CUP D33C23001110001.
\appendix
\section{$L^2$-equivalent norm on $L^2\mathcal{H}(\Omega)$}
\label{AppC}
We recall that $L^2\mathcal{H}(\Omega)$ consists of all square integrable fields $B$ which are div- and curl-free and that $\mathcal{H}_{\operatorname{ex}}(\Omega)\subset L^2\mathcal{H}(\Omega)$ consists of those $B\in L^2\mathcal{H}(\Omega)$ for which there is some $f\in H^1(\Omega)$ with $B=\nabla f$. By standard elliptic estimates any element $B\in L^2\mathcal{H}(\Omega)$ is analytic within $\Omega$ and since $L^2\mathcal{H}(\Omega)\subset H(\operatorname{div},\Omega)\cap H(\operatorname{curl},\Omega)$ every $B\in L^2\mathcal{H}(\Omega)$ has a normal trace $\mathcal{N}\cdot B\in W^{-\frac{1}{2},2}\mathcal{V}_0(\partial\Omega)$ and a twisted tangential trace $B\times \mathcal{N}\in \left(W^{-\frac{1}{2},2}(\partial\Omega)\right)^3$, recall the discussion preceding \Cref{2L8}. In addition, if $\Omega\subset\mathbb{R}^3$ is a $C^{1,1}$-solid torus, we may fix any simple closed $C^1$-loop $\sigma$ within $\Omega$ which represents a non-trivial element of the first fundamental group of $\Omega$ and define the \textit{circulation} $I_B$  of a given $B\in L^2\mathcal{H}(\Omega)$ as the line integral $\int_{\sigma}B$ which is always well-defined since $B$ is analytic within $\Omega$. We note that $I_B$ depends on the chosen $\sigma$ only via its orientation, i.e. if we pick any other non-trivial simple closed curve $\tilde{\sigma}\subset \Omega$ then $|\tilde{I}_B|=|I_B|$ and $\tilde{I}_B$ and $I_B$ have the same sign iff $\sigma$ and $\tilde{\sigma}$ are oriented in the same way. We have the following result, which may be seen as a generalisation of \cite[Lemma 11]{PRS22} since $\|j\|_{W^{-\frac{1}{2},2}(\partial\Omega)}\leq c(\partial\Omega)\|j\|_{L^2(\partial\Omega)}$ for some suitable $c>0$ which is independent of $j$.
\begin{thm}[Equivalent $L^2\mathcal{H}(\Omega)$-norm]
	\label{CT1}
	Let $\Omega\subset\mathbb{R}^3$ be a bounded $C^{1,1}$-domain with possibly disconnected boundary. Then there exists some $0<c_1(\Omega),c_2(\Omega)<\infty$ such that for all $\nabla f\in \mathcal{H}_{\operatorname{ex}}(\Omega)$ we have
	\begin{gather}
		\nonumber
		c_1\|\mathcal{N}\cdot\nabla f\|_{W^{-\frac{1}{2},2}(\partial\Omega)}\leq \|\nabla f\|_{L^2(\Omega)}\leq c_2\|\mathcal{N}\cdot\nabla f\|_{W^{-\frac{1}{2},2}(\partial\Omega)}.
	\end{gather}
	If in addition $\Omega$ is a solid torus and we fix some non-trivial simple closed $C^1$-loop $\sigma\subset\Omega$. Then there exist constants $0<c_1(\Omega),c_2(\Omega)<\infty$ such that for all $B\in L^2\mathcal{H}(\Omega)$ we have
	\begin{gather}
		\nonumber
		\|B\|_{L^2(\Omega)}\leq c_1\sqrt{\|\mathcal{N}\cdot B\|^2_{W^{-\frac{1}{2},2}(\partial\Omega)}+\left(\int_{\sigma}B\right)^2}\leq c_2\|B\|_{L^2(\Omega)}.
	\end{gather}
\end{thm}
\begin{proof}[Proof of \Cref{CT1}]
To simplify notation set $\Sigma:=\partial\Omega$. We fix any $\nabla f\in \mathcal{H}_{\operatorname{ex}}(\Omega)$ and observe that
\begin{gather}
	\nonumber
	\|\nabla f\|^2_{L^2(\Omega)}=\int_{\Sigma}f\left(\mathcal{N}\cdot\nabla f\right)d\sigma(x)\leq \|f\|_{W^{\frac{1}{2},2}(\Sigma)}\|\mathcal{N}\cdot \nabla f\|_{W^{-\frac{1}{2},2}(\Sigma)}
\end{gather}
where we integrated by parts and used the definition of the $W^{-\frac{1}{2},2}$-norm. We recall that the scalar potential is fixed only up to a constant so that we may assume that $\int_{\Omega}fd^3x=0$. Then, using the standard trace inequality and Poincar\'{e}'s inequality we can estimate $\|f\|_{W^{\frac{1}{2},2}(\Sigma)}\leq \tilde{c}(\Omega)\|\nabla f\|_{L^2(\Omega)}$ for some $\tilde{c}>0$ independent of $f$ so that
\begin{gather}
	\nonumber
	\|\nabla f\|_{L^2(\Omega)}\leq \tilde{c}(\Omega)\|\mathcal{N}\cdot \nabla f\|_{W^{-\frac{1}{2},2}(\Sigma)}.
\end{gather}
The converse inequality $\|\mathcal{N}\cdot \nabla f\|_{W^{-\frac{1}{2},2}(\Sigma)}\leq \hat{c}(\Omega)\|\nabla f\|_{L^2(\Omega)}$ for some $\hat{c}$ independent of $f$ follows from the continuity of the normal trace with respect to the $H(\operatorname{div},\Omega)$-norm and the fact that $\nabla f$ is a div-free field.
\newline
\newline
Now we suppose additionally that $\Omega\subset\mathbb{R}^3$ is a solid torus and we fix any $B\in L^2\mathcal{H}(\Omega)$. We first perform a Hodge-decomposition of $B$, \cite[Theorem B.1]{G24}, and write $B=\nabla f+\Gamma$ for suitable $\nabla f\in \mathcal{H}_{\operatorname{ex}}(\Omega)$, $\Gamma\in \mathcal{H}_N(\Omega)$. We note that this decomposition is $L^2(\Omega)$-orthogonal so that $\|B\|^2_{L^2(\Omega)}=\|\nabla f\|^2_{L^2(\Omega)}+\|\Gamma\|^2_{L^2(\Omega)}$. We first observe that $\dim\left(\mathcal{H}_N(\Omega)\right)=1$ because $\Omega$ is a solid torus, \cite[Theorem 2.6.1]{S95}, and that $\int_{\sigma}\Gamma=0$ implies that $\Gamma$ is a gradient field because $\sigma$ is a generator of the first fundamental group of $\Omega$. But since $\mathcal{H}_N(\Omega)$ and $\mathcal{H}_{\operatorname{ex}}(\Omega)$ are $L^2(\Omega)$-orthogonal we conclude $\int_{\sigma}\Gamma=0\Leftrightarrow \Gamma=0$. Hence we may fix some $\Gamma_0\in \mathcal{H}_N(\Omega)$ with $\int_{\sigma}\Gamma_0=1$ so that consequently $\Gamma=\kappa \Gamma_0$ with $\kappa=\int_{\sigma}\Gamma=\int_{\sigma}B$ for any $\Gamma\in \mathcal{H}_N(\Omega)$ where we used that $\sigma$ is a closed curve and $B$ and $\Gamma$ differ by a gradient field. Hence we find $\|\Gamma\|^2_{L^2(\Omega)}=\kappa^2\|\Gamma_0\|^2_{L^2(\Omega)}$, i.e. there is some $c(\Omega)>0$ independent of $B$ with
\begin{gather}
	\nonumber
	\|\Gamma\|^2_{L^2(\Omega)}=c(\Omega)\left(\int_\sigma B\right)^2.
\end{gather}
Combining the estimates for $\|\Gamma\|^2_{L^2(\Omega)}$ and $\|\nabla f\|_{L^2(\Omega)}$ and using that $\mathcal{N}\cdot \nabla f=\mathcal{N}\cdot B$ because by definition of $\mathcal{H}_N(\Omega)$ we have $\mathcal{N}\cdot \Gamma=0$ we arrive at the desired result.
\end{proof}
\section{Average poloidal and toroidal windings of kernel elements of $\operatorname{BS}_{\Sigma}$}
Let $\Sigma\subset\mathbb{R}^3$ be a $C^{1,1}$-surface which bounds a solid torus $\Omega$. We have already seen in the comparison of the procedure (\ref{2E13}) and our proposed algorithm \Cref{2T6} \& \Cref{2T9} that we can express a given current $j\in L^2\mathcal{V}_0(\Sigma)$ as $j=\nabla f\times \mathcal{N}+\alpha \gamma_t\times \mathcal{N}+\beta \gamma_p\times \mathcal{N}$ where $f\in H^1(\Sigma)$ is a suitable function, $\alpha,\beta\in \mathbb{R}$ and $\gamma_p,\gamma_t$ are a basis of the space $\mathcal{H}(\Sigma)$ of harmonic fields on our surface $\Sigma$ which are uniquely determined by demanding $\int_{\sigma_t}\gamma_t=1=\int_{\sigma_p}\gamma_p$ and $\int_{\sigma_p}\gamma_t=0=\int_{\sigma_t}\gamma_p=0$ for fixed poloidal and toroidal curves $\sigma_p$ and $\sigma_t$ on $\Sigma$. If we assume that $\sigma_t$ is homotopic to a toroidal loop $\tilde{\sigma}_t$ within the toroidal "plasma region" $P$, with $\overline{P}\subset \Omega$, then we may let $I_p:=\int_{\tilde{\sigma}_t}B_T$ where $B_T\in L^2\mathcal{H}(P)$ is our target magnetic field and we have argued that we have $\|\operatorname{BS}_{\Sigma}(j)-B_T\|^2_{L^2(P)}\geq (\alpha-I_p)^2\|\widetilde{\Gamma}\|^2_{L^2(P)}$ for some suitable (non-zero) element $\widetilde{\Gamma}$ of $\mathcal{H}_N(P)$. Consequently we must have $\alpha=I_p$ to be able to approximate the magnetic field $B_T$ well. We also know that according to \Cref{2T3} the kernel of $\operatorname{BS}_{\Sigma}$ is $1$-dimensional whenever $\Sigma$ bounds a $C^{1,1}$-solid torus. The goal of the present section is to show that fixing the parameter $\beta$ in the expression $j=\nabla f\times \mathcal{N}+I_p\gamma_t\times \mathcal{N}+\beta \gamma_p\times \mathcal{N}$ essentially determines the "contribution" of $\operatorname{Ker}(\operatorname{BS}_{\Sigma})$ to $j$ in the sense as stated in \Cref{DC2}.
\newline
\newline
Before we formulate our results we recall here from the discussion preceding \Cref{2P10} that for a Lipschitz-continuous, div-free current $j$ on $\Sigma$ we denote for fixed $x\in \Sigma$ by $\sigma_x$ the field line of $j$ starting at $x$ and that then, c.f. \cite[Definition 2.11, Lemma 2.12]{G24SurfaceHelicityArXiv}, the quantities $\hat{q}(x):=\lim_{T\rightarrow\infty}\frac{1}{T}\int_{\sigma_x[0,T]}\gamma_t$, $\hat{p}(x):=\lim_{T\rightarrow\infty}\frac{1}{T}\int_{\sigma_x[0,T]}\gamma_p$ are well-defined $L^1(\Sigma)$ functions whose averages may be computed according to
\begin{gather}
	\nonumber
	\overline{P}(j):=\frac{1}{|\Sigma|}\int_{\Sigma}\hat{p}(x)d\sigma(x)=\frac{1}{|\Sigma|}\int_{\Sigma}j(x)\cdot \gamma_p(x)d\sigma(x),
\end{gather}
	\begin{AppB}
	\label{ABE1}
	\overline{Q}(j):=\frac{1}{|\Sigma|}\int_{\Sigma}\hat{q}(x)d\sigma(x)=\frac{1}{|\Sigma|}\int_{\Sigma}j(x)\cdot \gamma_t(x)d\sigma(x).
\end{AppB}
$\overline{P}(j)$ and $\overline{Q}(j)$ can be interpreted as the average poloidal and toroidal wrappings of the field lines of $j$ around $\Sigma$ respectively. In particular, the formulas in (\ref{ABE1}) allow us to extend the notions of average poloidal and toroidal windings to any elements $j\in W^{-\frac{1}{2},2}\mathcal{V}_0(\Sigma)$.
\begin{thm}[Poloidal \& toroidal windings of $\operatorname{Ker}(\operatorname{BS}_{\Sigma})$]
	\label{DT1}
	Let $\Sigma\subset\mathbb{R}^3$ be a $C^{1,1}$-surface which bounds a solid torus $\Omega$ and let $\sigma_p$ and $\sigma_t$ be a poloidal and a toroidal curve in $\Sigma$ respectively, i.e. we assume that $\sigma_p$ bounds a disc in $\Omega$ and $\sigma_t$ bounds a surface outside of $\Omega$. Then there is some $j_0\in \operatorname{Ker}(\operatorname{BS}_{\Sigma})$ which satisfies
	\begin{gather}
		\nonumber
		\overline{P}(j_0)=0\text{ and }\overline{Q}(j_0)=1.
	\end{gather} 
\end{thm}
\begin{proof}[Proof of \Cref{DT1}]
	It follows first from \cite[Theorem 2.33]{G24SurfaceHelicityArXiv} and its proof that there exists some $j_0\in \operatorname{Ker}(\operatorname{BS}_{\Sigma})$ with $\overline{Q}(j_0)\neq 0$ so that by scaling and linearity of $\overline{Q}$ we find some $j_0\in \operatorname{Ker}(\operatorname{BS}_{\Sigma})$ with $\overline{Q}(j_0)=1$. We are left with proving that every element $j$ of $\operatorname{Ker}(\operatorname{BS}_{\Sigma})$ satisfies $\overline{P}(j)=0$ which will conclude the proof. According to (\ref{ABE1}) we need to show that
	\begin{AppB}
		\label{ABE2}
		\int_{\Sigma}j(x)\cdot \gamma_p(x)d\sigma(x)=0.
	\end{AppB}
	We recall that since $\Sigma$ is toroidal $\dim\left(\operatorname{Ker}(\operatorname{BS}_{\Sigma})\right)=1$, c.f. \Cref{2T3}, and that according to \Cref{2T11} a non-zero element of $\operatorname{Ker}(\operatorname{BS}_{\Sigma})$ is given by $\operatorname{BS}_{\Omega}(\Gamma)\times \mathcal{N}+\nabla g\times \mathcal{N}$ where $\Gamma\in \mathcal{H}_N(\Omega)\setminus \{0\}$ can be arbitrarily fixed and $g\in H^1(\Omega)$ is a suitable function depending on the choice of $\Gamma$. In particular, we may fix $\Gamma\in \mathcal{H}_N(\Omega)\setminus \{0\}$ such that the $L^2(\Sigma)$-orthogonal projection $\gamma$ of $\Gamma|_{\Sigma}$ onto $\mathcal{H}(\Sigma)$ satisfies $\|\gamma\|_{L^2(\Sigma)}=1$. We may therefore assume without loss of generality that $j=\operatorname{BS}_{\Omega}(\Gamma)\times \mathcal{N}+\nabla g\times \mathcal{N}$ for a suitable $g\in H^1(\Omega)$ and $\Gamma\in \mathcal{H}_N(\Omega)$ with $\|\gamma\|_{L^2(\Sigma)}=1$. But since $\mathcal{H}(\Sigma)$ is $L^2(\Sigma)$-orthogonal to the co-exact fields we see that the integral in (\ref{ABE2}) becomes
	\begin{gather}
		\nonumber
		\int_{\Sigma}j(x)\cdot \gamma_p(x)d\sigma=\int_{\Sigma}\left(\operatorname{BS}_{\Omega}(\Gamma)\times \mathcal{N}\right)\cdot \gamma_pd\sigma(x)
		\\
		\nonumber
		=\int_{\Sigma}(\mathcal{N}\times \gamma_p)\cdot \operatorname{BS}_{\Omega}(\Gamma)d\sigma=\int_{\Sigma}(\mathcal{N}\times \gamma_p)\cdot \left(\pi_{\mathcal{H}(\Sigma)}(\operatorname{BS}^{\parallel}_{\Omega}(\Gamma))\right)d\sigma
	\end{gather}
	where $\operatorname{BS}^{\parallel}_{\Omega}(\Gamma)$ denotes the part of $\operatorname{BS}_{\Omega}(\Gamma)|_{\Sigma}$ which is tangent to $\Sigma$, where $\pi_{\mathcal{H}(\Sigma)}$ denotes the $L^2(\Sigma)$-orthogonal projection onto the space $\mathcal{H}(\Sigma)$ and where we used that $\mathcal{N}\times \gamma_p\in \mathcal{H}(\Sigma)$ because $\gamma_p$ is. It follows finally from the proof of \cite[Proposition C.4]{G24SurfaceHelicityArXiv} that $\pi_{\mathcal{H}(\Sigma)}(\operatorname{BS}^{\parallel}_{\Omega}(\Gamma))=\operatorname{Flux}(\Gamma)\gamma_p$ where $\operatorname{Flux}(\Gamma)$ denotes the flux of $\Gamma$ through the disc bounded by $\sigma_p$. Most importantly, $\pi_{\mathcal{H}(\Sigma)}(\operatorname{BS}^{\parallel}_{\Omega}(\Gamma))$ is parallel to $\gamma_p$ and consequently we find $\int_{\Sigma}j(x)\cdot \gamma_p(x)d\sigma=0$ as desired.
\end{proof}
We obtain the following corollary
\begin{cor}
	\label{DC2}
	Let $\Sigma\subset\mathbb{R}^3$ be a $C^{1,1}$-surface which bounds a solid torus $\Omega$ and let $\sigma_p$, $\sigma_t$ be a poloidal and toroidal curve in $\Sigma$ respectively, i.e. $\sigma_p$ bounds a disc in $\Omega$ and $\sigma_t$ bounds a surface outside of $\Omega$. Then for given $j_i\in L^2\mathcal{V}_0(\Sigma)$, $i=1,2$, we denote their Hodge-decomposition by $j_i=\nabla f_i\times \mathcal{N}+\alpha_i\gamma_t\times \mathcal{N}+\beta_i\gamma_p\times \mathcal{N}$ for suitable $f_i\in H^1(\Sigma)$, $\alpha_i,\beta_i\in \mathbb{R}$. Then the following are equivalent
	\begin{enumerate}
		\item $\overline{Q}(j_1)=\overline{Q}(j_2)$,
		\item $\beta_1=\beta_2$.
	\end{enumerate}
	Further, if $\widetilde{j}\in \operatorname{Ker}(\operatorname{BS}_{\Sigma})$ and $j\in L^2\mathcal{V}_0(\Sigma)$, we have $\overline{Q}(j+\widetilde{j})=\overline{Q}(j)\Leftrightarrow \widetilde{j}=0$.
\end{cor}
\begin{proof}[Proof of \Cref{DC2}]
	\underline{(i)$\Leftrightarrow$(ii):} By linearity it is enough to show that $\overline{Q}(j_1)=0$ if and only if $\beta=0$. From definition we immediately obtain $\overline{Q}(j_1)=0\Leftrightarrow \int_{\Sigma}\gamma_t\cdot j_1d\sigma=0\Leftrightarrow \int_{\Sigma}(\gamma_p\times \mathcal{N})\cdot \gamma_td\sigma\beta=0$ since $\gamma_t\in \mathcal{H}(\Sigma)$ is $L^2$-orthogonal to the co-exact fields and is pointwise everywhere orthogonal to $\gamma_t\times \mathcal{N}$. We observe that $\gamma_p\times \mathcal{N}\in \mathcal{H}(\Sigma)$ (in the language of differential forms taking the cross product with the outer normal corresponds to applying the Hodge star operator to the associated $1$-form) and $\gamma_p\times \mathcal{N}$ is linearly independent of $\gamma_p$ and therefore $\int_{\Sigma}(\gamma_p\times \mathcal{N})\cdot \gamma_td\sigma\neq 0$ so that $\overline{Q}(j_1)=0$ if and only if $\beta=0$.
	\newline
	\newline
	\underline{$\overline{Q}(j+\widetilde{j})=\overline{Q}(j)\Leftrightarrow \widetilde{j}=0$:} According to \Cref{DT1} we can write $\widetilde{j}=\kappa j_0$, $\kappa\in \mathbb{R}$, with $j_0$ as in \Cref{DT1} since $\operatorname{Ker}(\operatorname{BS}_{\Sigma})$ is $1$-dimensional, c.f. \Cref{2T3}. By linearity we have then to show that $\overline{Q}(\widetilde{j})=0\Leftrightarrow \widetilde{j}=0$ but this follows immediately from $\overline{Q}(\widetilde{j})=\kappa\overline{Q}(j_0)=\kappa$.
\end{proof}
Mathematically, we have the following deeper result, where we define
\begin{gather}
	\nonumber
	W^{-\frac{1}{2},2}\mathcal{V}^{\overline{Q}=0}_0(\Sigma):=\{j\in W^{-\frac{1}{2},2}\mathcal{V}_0(\Sigma)\mid \overline{Q}(j)=0 \}
\end{gather}
which defines a closed subspace of $W^{-\frac{1}{2},2}\mathcal{V}_0(\Sigma)$. We note that the condition $\overline{Q}(j)=0$ is independent of the chosen toroidal and poloidal curves $\sigma_p$ and $\sigma_t$ on $\Sigma$ since $\overline{Q}$ and $\overline{P}$ depend on $\sigma_t$ and $\sigma_p$ only via their orientation which at most leads to a change in sign of $\overline{Q}$ and $\overline{P}$.
\begin{thm}
	\label{DT3}
	Let $\Sigma\subset\mathbb{R}^3$ be a $C^{1,1}$-surface which bounds a solid torus $\Omega$. Then
	\begin{gather}
		\nonumber
		\operatorname{BS}_{\Sigma}:W^{-\frac{1}{2},2}\mathcal{V}^{\overline{Q}=0}_0(\Sigma)\rightarrow L^2\mathcal{H}(\Omega)
	\end{gather}
	is a linear isomorphism and there exist constants $0<c_1(\Sigma),c_2(\Sigma)<\infty$ such that
	\begin{gather}
		\nonumber
		\|j\|_{W^{-\frac{1}{2},2}(\Sigma)}\leq c_1(\Sigma)\|\operatorname{BS}_{\Sigma}(j)\|_{L^2(\Omega)}\leq c_2(\Sigma)\|j\|_{W^{-\frac{1}{2},2}(\Sigma)}\text{ for all }j\in W^{-\frac{1}{2},2}\mathcal{V}^{\overline{Q}=0}_0(\Sigma).
	\end{gather}
\end{thm}
\begin{proof}[Proof of \Cref{DT3}]
	First we observe that $\Sigma$ is connected and hence $\dim\left(\mathcal{H}_D(\Omega)\right)=\#\Sigma-1=0$ so that according to \Cref{2T2} we see that $\operatorname{BS}_{\Sigma}$ is surjective as a map from $W^{-\frac{1}{2},2}\mathcal{V}_0(\Sigma)$ into $L^2\mathcal{H}(\Omega)$. Now fix any $B\in L^2\mathcal{H}(\Omega)$ and preimage $j\in W^{-\frac{1}{2},2}\mathcal{V}_0(\Sigma)$ of $B$. Then by means of \Cref{DT1} we know there exists $j_0\in \operatorname{Ker}(\operatorname{BS}_{\Sigma})$ with $\overline{Q}(j_0)=\overline{Q}(j)$ and therefore $j-j_0\in W^{-\frac{1}{2},2}\mathcal{V}^{\overline{Q}=0}_0(\Sigma)$ is a preimage of $B$. We conclude that $\operatorname{BS}_{\Sigma}$ remains surjective as a map from $W^{-\frac{1}{2},2}\mathcal{V}^{\overline{Q}=0}_0(\Sigma)$ into $L^2\mathcal{H}(\Omega)$. As for the injectivity we observe that $\operatorname{BS}_{\Sigma}(j)=0$ implies $j\in \operatorname{Ker}(\operatorname{BS}_{\Sigma})$ and since the kernel is $1$-dimensional when $\Omega$ is a solid torus, c.f. \Cref{2T3}, we find a $\kappa\in \mathbb{R}$ such that $j=\kappa j_0$ where $j_0$ now denotes the kernel element obtained from \Cref{DT1}. Further, we know that $0=\overline{Q}(j)=\kappa\overline{Q}(j_0)=\kappa$ by properties of $j_0$ and therefore $j=0$, proving that $\operatorname{BS}_{\Sigma}$ is injective as a map from $W^{-\frac{1}{2},2}\mathcal{V}^{\overline{Q}=0}_0(\Sigma)$ into $L^2\mathcal{H}(\Omega)$. Lastly, according to \cite[Lemma 5.1]{G24}, the Biot-Savart operator is continuous. Hence the bounded-inverse theorem implies that the inverse remains a bounded operator.
\end{proof}
We obtain therefore in general the following equivalent norm where we note that $|\overline{Q}(j)|$ does not depend on the chosen poloidal and toroidal curves since this choice at most affects the sign of $\overline{Q}(j)$ and $\overline{P}(j)$.
\begin{cor}[Equivalent $W^{-\frac{1}{2},2}\mathcal{V}_0(\Sigma)$-norm]
	\label{DC4}
	Let $\Sigma\subset\mathbb{R}^3$ be a $C^{1,1}$-surface which bounds a solid torus $\Omega$ and suppose that $\sigma_p,\sigma_t\subset \Sigma$ are $C^1$-curves which bound a disc within $\Omega$ and a surface outside of $\Omega$ respectively. Then there exist constants $0<C_1(\Sigma),C_2(\Sigma)<\infty$ (independent of the chosen $\sigma_p,\sigma_t$) such that
	\begin{gather}
		\nonumber
		\|j\|_{W^{-\frac{1}{2},2}(\Sigma)}\leq C_1\sqrt{\|\operatorname{BS}_{\Sigma}(j)\|^2_{L^2(\Omega)}+|Q(j)|^2}\leq C_2\|j\|_{W^{-\frac{1}{2},2}(\Sigma)}\text{ for all }j\in W^{-\frac{1}{2},2}\mathcal{V}_0(\Sigma).
	\end{gather}
\end{cor}
\begin{proof}[Proof of \Cref{DC4}]
	Given $j\in W^{-\frac{1}{2},2}\mathcal{V}(\Sigma)$ we may let $\kappa:=\overline{Q}(j)$ and if we let $j_0\in \operatorname{Ker}(\operatorname{BS}_{\Sigma})$ as in \Cref{DT1} we can write $j=j-\kappa j_0+\kappa j_0$ with $j-\kappa j_0\in W^{-\frac{1}{2},2}\mathcal{V}_0^{\overline{Q}=0}(\Sigma)$ so that by means of \Cref{DT3} we have
	\begin{gather}
		\nonumber
		\|j\|_{W^{-\frac{1}{2},2}}\leq \|j-\kappa j_0\|_{W^{-\frac{1}{2}}}+|\kappa|\|j_0\|_{W^{-\frac{1}{2},2}}\leq c_1(\Sigma)\|\operatorname{BS}_{\Sigma}(j-\kappa j_0)\|_{L^2(\Omega)}+|\kappa|\|j_0\|_{W^{-\frac{1}{2},2}}
		\\
		\nonumber
		=c_1(\Sigma)\|\operatorname{BS}_{\Sigma}(j)\|_{L^2(\Omega)}+|\overline{Q}(j)|\|j_0\|_{W^{-\frac{1}{2},2}}
	\end{gather}
	with the constant $c_1(\Sigma)$ from \Cref{DT3} and where we used that $j_0\in \operatorname{Ker}(\operatorname{BS}_{\Sigma})$. The remaining inequality follows immediately from the fact that $\operatorname{BS}_{\Sigma}:W^{-\frac{1}{2},2}(\Sigma)\rightarrow L^2\mathcal{H}(\Omega)$ and $\overline{Q}:W^{-\frac{1}{2},2}(\Sigma)\rightarrow\mathbb{R}$ are bounded linear operators.
\end{proof}
\section{Friedrichs decomposition on $C^{1,1}$-domains}
\label{AppE}
Here we establish the Friedrichs decomposition of the space $L^2\mathcal{H}(\Omega)$. We will follow closely the proof given in \cite[Theorem 2.4.8]{S95} which applies to general $k$-forms and deals with smooth domains. In contrast, here we translate the proofs into the language of vector fields for the convenience of the reader and also explain how the $C^{1,1}$-regularity of the boundary comes into play.
\begin{lem}[Ellipticity of Dirichlet-integral]
	\label{AppEL1}
	Let $\Omega\subset\mathbb{R}^3$ be a bounded $C^{1,1}$-domain. Then there exist $c_1(\Omega),c_2(\Omega)>0$ such that for all $X\in H:=\{X\in H^1(\Omega,\mathbb{R}^3)\mid X\times \mathcal{N}=0\}\cap \mathcal{H}^{\perp_{L^2(\Omega)}}_D(\Omega)$ we have the inequality
	\begin{gather}
		\nonumber
		\|X\|_{H^1(\Omega)}\leq c_1\sqrt{\|\operatorname{curl}(X)\|^2_{L^2(\Omega)}+\|\operatorname{div}(X)\|^2_{L^2(\Omega)}}\leq c_2\|X\|_{H^1(\Omega)}.
	\end{gather}
\end{lem}
\begin{proof}[Proof of \Cref{AppEL1}]
	The proof follows closely the exposition of the proof given in \cite[Proposition 2.2.3]{S95}. We start with the following Gaffney type inequality valid on $C^{1,1}$-domains, \cite[Lemma 2.11, Equation (2.12)]{ABDG98},
	\begin{AppE}
		\label{AppEE1}
		\|X\|^2_{H^1(\Omega)}\leq c\left(\|X\|^2_{L^2(\Omega)}+\|\operatorname{curl}(X)\|^2_{L^2(\Omega)}+\|\operatorname{div}(X)\|^2_{L^2(\Omega)}\right)\text{ for all }X\in H
	\end{AppE}
	for some $c>0$ independent of $X$. We can now take a sequence $(X_n)_n$ with $\|X_n\|_{L^2(\Omega)}=1$ with $\lim_{n\rightarrow\infty}\left(\|\operatorname{curl}(X_n)\|^2_{L^2(\Omega)}+\|\operatorname{div}(X_n)\|^2_{L^2(\Omega)}\right)=\inf_{X\in H,\|X\|_{L^2(\Omega)}=1}\left(\|\operatorname{curl}(X)\|^2_{L^2(\Omega)}+\|\operatorname{div}(X)\|^2_{L^2(\Omega)}\right)$. According to Gaffney's inequality (\ref{AppEE1}) we conclude that the $H^1(\Omega)$-norm of the $X_n$ are bounded and hence by the Rellich-Kondrachov theorem we may assume that the $X_n$ converge strongly to some $X$ in $L^2(\Omega)$ and due to the fact that $H$ is an $H^1$-closed subspace of $H^1(\Omega,\mathbb{R}^3)$ and hence a Hilbert space in its own right we may further assume that the $X_n$ converge weakly in $H^1(\Omega)$ to the same element $X\in H$. In particular, $\operatorname{curl}(X_n)$ and $\operatorname{div}(X_n)$ converge weakly to $\operatorname{curl}(X)$ and $\operatorname{div}(X)$ in $L^2(\Omega)$ norm respectively. Then the lower semi-continuity of the norms and the strong $L^2$-convergence imply that $\|X\|_{L^2(\Omega)}=1$ and
	\begin{gather}
		\nonumber
			\|\operatorname{curl}(X)\|^2_{L^2(\Omega)}+\|\operatorname{div}(X)\|^2_{L^2(\Omega)}=\inf_{Y\in H,\|Y\|_{L^2(\Omega)}=1}\left(\|\operatorname{curl}(Y)\|^2_{L^2(\Omega)}+\|\operatorname{div}(Y)\|^2_{L^2(\Omega)}\right).
	\end{gather}
We note that $c_0:=\|\operatorname{curl}(X)\|^2_{L^2(\Omega)}+\|\operatorname{div}(X)\|^2_{L^2(\Omega)}>0$, since $c_0=0$ would imply $\operatorname{curl}(X)=0=\operatorname{div}(X)$ and since $X\times \mathcal{N}=0$ (as $X\in H$) we would find $X\in \mathcal{H}_D(\Omega)$. But $X\in H\subset \mathcal{H}^{\perp_{L^2(\Omega)}}_D(\Omega)$ so that we would find $X=0$ contradicting $\|X\|_{L^2(\Omega)}=1$. If $B\in H$ is any other arbitrary element, we then find
\begin{AppE}
	\label{AppEE2}
	\|\operatorname{curl}(B)\|^2_{L^2(\Omega)}+\|\operatorname{div}(B)\|^2_{L^2(\Omega)}\geq c_0\|B\|^2_{L^2(\Omega)}
\end{AppE}
for some $c_0>0$ independent of $B$. Combining (\ref{AppEE2}) and (\ref{AppEE1}) proves that
\begin{gather}
	\nonumber
	\|X\|^2_{H^1(\Omega)}\leq c_1\left(\|\operatorname{curl}(X)\|^2_{L^2(\Omega)}+\|\operatorname{div}(X)\|^2_{L^2(\Omega)}\right) 
\end{gather}
for all $X\in H$ for a suitable $c_1>0$ independent of $X$. Lastly, the second inequality in the statement of the lemma is trivial to verify.
\end{proof}
We recall that $L^2\mathcal{H}(\Omega)$ is the space of square-integrable vector fields on $\Omega$ which are div- and curl-free in the weak sense on $\Omega$.
\begin{lem}
	\label{AppEL2}
	Let $\Omega\subset\mathbb{R}^3$ be a bounded $C^{1,1}$-domain. Then there exists some $c>0$ such that for all $X\in L^2\mathcal{H}(\Omega)\cap \mathcal{H}^{\perp_{L^2(\Omega)}}_D(\Omega)$ there is some $A\in H^1(\Omega,\mathbb{R}^3)$ with $\operatorname{div}(A)=0$, $\mathcal{N}\cdot A=0$, $\operatorname{curl}(A)=X$ and $\|A\|_{H^1(\Omega)}\leq c\|X\|_{L^2(\Omega)}$.
\end{lem}
\begin{proof}[Proof of \Cref{AppEL2}]
	The proof follows closely the idea presented in \cite[Theorem 2.4.8]{S95} and is based on the existence of Dirichlet-potentials \cite[Theorem 2.2.4]{S95}. Here we adapt the proof once more to the language of vector fields for the convenience of the reader. Fix $X\in L^2\mathcal{H}(\Omega)\cap \mathcal{H}^{\perp_{L^2(\Omega)}}_D(\Omega)$ and define the space $H:=\{Y\in H^1(\Omega,\mathbb{R}^3)\mid Y\times \mathcal{N}=0\}\cap \mathcal{H}^{\perp_{L^2(\Omega)}}_D(\Omega)$ and the following linear functional
	\begin{gather}
		\nonumber
		T:H\rightarrow \mathbb{R}\text{, }B\mapsto \int_{\Omega}B\cdot X d^3x
	\end{gather}
	which is clearly bounded with respect to the $H^1$-norm. According to \Cref{AppEL1} the inner product $\langle A,B\rangle:=\int_{\Omega}\operatorname{curl}(A)\cdot \operatorname{curl}(B)d^3x+\int_{\Omega}\operatorname{div}(A)\operatorname{div}(B)d^3x$ turns $H$ into a Hilbert space whose norm is equivalent to the $H^1$-norm. We conclude that $T$ is bounded with respect to the induced norm $\|\cdot \|$ of the inner product $\langle\cdot,\cdot\rangle$. Then by means of Riesz representation theorem we conclude that there exists (a unique) $A\in H$ satisfying
	\begin{AppE}
		\label{AppEE3}
		\int_{\Omega}B\cdot Xd^3x=T(B)=\langle B,A\rangle=\int_{\Omega}\operatorname{curl}(A)\cdot \operatorname{curl}(B)d^3x+\int_{\Omega}\operatorname{div}(A)\operatorname{div}(B)d^3x\text{ for all }B\in H.
	\end{AppE}
	Given any $h\in L^2(\Omega)$ we can now solve the Dirichlet problem $\Delta f=h$ in $\Omega$ and $f|_{\partial\Omega}=0$ which admits a unique solution of class $f\in W^{1,2}_0(\Omega)\cap W^{2,2}(\Omega)$, c.f. \cite[Theorem 2.4.2.5]{Gris85}. We observe that $\nabla f\in H$ due to the boundary conditions of $f$ and since each such $\nabla f$ is $L^2(\Omega)$ to the space $L^2\mathcal{H}(\Omega)$. We can therefore set $B=\nabla f$ in (\ref{AppEE3}) and use the fact that $\nabla f$ and $X\in L^2\mathcal{H}(\Omega)$ are $L^2(\Omega)$-orthogonal to each other to conclude
	\begin{gather}
		\nonumber
		0=\int_{\Omega}\nabla f\cdot Xd^3x=\int_{\Omega}\operatorname{curl}(\nabla f)\cdot \operatorname{curl}(A)d^3x+\int_{\Omega}\operatorname{div}(A)\cdot \Delta fd^3x=\int_{\Omega}\operatorname{div}(A)\cdot hd^3x
	\end{gather}
	from which we conclude, by selecting $h=\operatorname{div}(A)$, $\operatorname{div}(A)=0$. Coming back to (\ref{AppEE3}) we obtain
	\begin{gather}
		\nonumber
		\int_{\Omega}B\cdot Xd^3x=\int_{\Omega}\operatorname{curl}(A)\cdot \operatorname{curl}(B)d^3x\text{ for all }B\in H.
	\end{gather}
	We observe now that $X$ is $L^2(\Omega)$-orthogonal to the space $\mathcal{H}_D(\Omega)$ so that for an arbitrary $\widetilde{B}\in \{Y\in H^1(\Omega,\mathbb{R}^3)\mid Y\times\mathcal{N}=0\}$ we can decompose, due to the finite dimensionality of $\mathcal{H}_D(\Omega)$, $\widetilde{B}=B+\nabla g$ for suitable $B\in H$ and $\nabla g\in \mathcal{H}_D(\Omega)$ from which we easily conclude
	\begin{gather}
		\nonumber
		\int_{\Omega}\widetilde{B}\cdot Xd^3x=\int_{\Omega}\operatorname{curl}(\widetilde{B})\cdot \operatorname{curl}(A)d^3x\text{ for all }\widetilde{B}\in H^1(\Omega,\mathbb{R}^3)\text{ with }\widetilde{B}\times \mathcal{N}=0.
	\end{gather}
	In particular, we may set $\widetilde{B}=\Psi$ for any given $\Psi\in C^{\infty}_c(\Omega,\mathbb{R}^3)$ and find
	\begin{gather}
		\nonumber
		\int_{\Omega}\Psi\cdot Xd^3x=\int_{\Omega}\operatorname{curl}(\Psi)\cdot \operatorname{curl}(A)d^3x\text{ for all }\Psi\in C^{\infty}_c(\Omega,\mathbb{R}^3).
	\end{gather}
	This implies by definition of the weak curl, that $\operatorname{curl}(A)\in H(\operatorname{curl},\Omega)$ with $\operatorname{curl}(\operatorname{curl}(A))=X$. We may then set $w:=\operatorname{curl}(A)$ and observe first that $\operatorname{div}(w)=0$, $\operatorname{curl}(w)=X$ and $\mathcal{N}\cdot w=0$ because $A\times \mathcal{N}=0$. We can further plug in $B=A\in H$ in (\ref{AppEE3}), keeping in mind that $\operatorname{div}(A)=0$, to conclude that $\|w\|^2_{L^2(\Omega)}\leq \|X\|_{L^2(\Omega)}\|A\|_{L^2(\Omega)}$. Further, since $A\in H$ it follows from (\ref{AppEE2}) that $\|A\|_{L^2(\Omega)}\leq c\|\operatorname{curl}(A)\|_{L^2(\Omega)}=c\|w\|_{L^2(\Omega)}$ for some $c>0$ independent of $A$ and where we used again that $\operatorname{div}(A)=0$. We conclude that $\|w\|_{L^2(\Omega)}\leq c\|X\|_{L^2(\Omega)}$ and that $\|\operatorname{curl}(w)\|_{L^2(\Omega)}=\|X\|_{L^2(\Omega)}$ and therefore $\|w\|_{H(\operatorname{curl},\Omega)}\leq c\|X\|_{L^2(\Omega)}$ for some $c>0$ independent of $X$. We are left with upgrading the $H(\operatorname{curl},\Omega)$ estimate to an $H^1$-estimate which follows from the corresponding Gaffney type inequality for $C^{1,1}$-domains for vector fields satisfying a tangent to the boundary condition, c.f. \cite[Theorem 2.9 \& Lemma 2.11, Equation (2.12)]{ABDG98}.
\end{proof}
\begin{thm}[Friedrichs decomposition]
	\label{AppET3}
	Let $\Omega\subset\mathbb{R}^3$ be a bounded $C^{1,1}$-domain. Then for every $B\in L^2\mathcal{H}(\Omega)$ there exists some $A\in H^1(\Omega,\mathbb{R}^3)$,$\nabla f\in \mathcal{H}_D(\Omega)$ with $\operatorname{div}(A)=0$, $\mathcal{N}\cdot A=0$ and such that the following $L^2(\Omega)$-orthogonal decomposition holds
	\begin{gather}
		\nonumber
		B=\operatorname{curl}(A)+\nabla f.
	\end{gather}
	Further there exists a constant $c(\Omega)>0$ independent of $B$ such that we have the a priori estimate
	\begin{gather}
		\nonumber
		\|A\|_{H^1(\Omega)}+\|\nabla f\|_{L^2(\Omega)}\leq c\|B\|_{L^2(\Omega)}.
	\end{gather}
\end{thm}
\begin{proof}[Proof of \Cref{AppET3}]
	Since $\mathcal{H}_D(\Omega)$ is finite dimensional we may decompose $B=\widetilde{B}+\nabla f$ for $\nabla f\in \mathcal{H}_D(\Omega)$ and $\widetilde{B}\in L^2\mathcal{H}(\Omega)\cap \mathcal{H}^{\perp_{L^2(\Omega)}}_D(\Omega)$. Since this decomposition is $L^2(\Omega)$-orthogonal we have $\|\nabla f\|_{L^2(\Omega)}\leq \|B\|_{L^2(\Omega)}$ and $\|\widetilde{B}\|_{L^2(\Omega)}\leq \|B\|_{L^2(\Omega)}$. We can now apply \Cref{AppEL2} to $\widetilde{B}$.
\end{proof}
\bibliographystyle{plain}
\bibliography{mybibfileNOHYPERLINK}

\begin{thebibliography}{10}

\bibitem{ABDG98}
C.~Amrouche, C.~Bernardi, M.~Dauge, and V.~Girault.
\newblock Vector {P}otentials in {T}hree-dimensional {N}on-smooth {D}omains.
\newblock {\em Mathematical Methods in the Applied Sciences}, 23:823--864,
  1998.

\bibitem{CDG01}
J.~Cantarella, D.~DeTurck, and H.~Gluck.
\newblock The {B}iot-{S}avart operator for application in knot theory, fluid
  dynamics, and plasma physics.
\newblock {\em Journal of Mathematical Physics}, 42(2):876--905, 2001.

\bibitem{CDG02}
J.~Cantarella, D.~DeTurck, and H.~Gluck.
\newblock Vector calculus and the topology of domains in 3-space.
\newblock {\em American Mathematical Monthly}, 109(5):409--442, 2002.

\bibitem{Cial95}
A.~Cialdea.
\newblock A {G}eneral {T}heory of {H}ypersurface {P}otentials.
\newblock {\em Annali di Matematica pura ed applicata}, CLXVIII:37--61, 1995.

\bibitem{RCM21}
M.~Dalla~Riva, M.~Lanza~de Cristoforis, and P.~Musolino.
\newblock {\em Singularly {P}erturbed {B}oundary {V}alue {P}roblems}.
\newblock Springer, 2021.

\bibitem{Dal18}
J.~Dalphin.
\newblock Uniform ball property and existence of optimal shapes for a wide
  class of geometric functionals.
\newblock {\em Interfaces Free Bound.}, 20(2):211--260, 2018.

\bibitem{DD12}
F.~Demengel and G.~Demengel.
\newblock {\em Functional {S}paces for the {T}heory of {E}lliptic {P}artial
  {D}ifferential {E}quations}.
\newblock Springer Verlag, 2012.

\bibitem{EG15}
L.C. Evans and R.F. Gariepy.
\newblock {\em Measure Theory and Fine Properties of Functions: Revisited
  Edition}.
\newblock Apple Academic Press Inc., 2015.

\bibitem{G24SurfaceHelicityArXiv}
W.~{Gerner}.
\newblock {Asymptotic windings, surface helicity and their applications in
  plasma physics}.
\newblock {\em arXiv e-prints}, page
  \href{https://arxiv.org/abs/2408.00492}{arXiv:2408.00492}, 2024.

\bibitem{G24}
W.~Gerner.
\newblock Properties of the {B}iot-{S}avart operator acting on surface
  currents.
\newblock {\em SIAM J. Math. Anal.}, 56(5):6446--6482, 2024.

\bibitem{GT01}
D.~Gilbarg and N.~Trudinger.
\newblock {\em Elliptic {P}artial {D}ifferential {E}quations of {S}econd
  {O}rder}.
\newblock Springer Verlag, 2001.

\bibitem{GR86}
V.~Girault and P.-A. Raviart.
\newblock {\em Finite {E}lement {M}ethods for {N}avier-{S}tokes {E}quations}.
\newblock Springer, 1986.

\bibitem{Gris85}
P.~Grisvard.
\newblock {\em Elliptic {P}roblems in {N}onsmooth {D}omains}.
\newblock Pitman Publishing, 1985.

\bibitem{HHPH21}
S.A. Henneberg, S.R. Hudson, D.~Pfefferl\'{e}, and P.~Helander.
\newblock Combined plasma-coil optimization algorithms.
\newblock {\em J. Plasma Phys.}, 87:905870226, 2021.

\bibitem{HW83}
S.P. Hirshman and J.C. Whitson.
\newblock Steepest-descent moment method for three-dimensional
  mangetohydrodynamic equilibria.
\newblock {\em Phys. Fluids}, 26:3553--3568, 1983.

\bibitem{IGPW24}
L.-M. Imbert-Gerard, E.J. Paul, and A.M. Wright.
\newblock {\em An {I}ntroduction to {S}tellarators}.
\newblock SIAM, 2024.

\bibitem{JGLRW23}
R.~Jorge, A.~Goodman, M.~Landreman, J.~Rodrigues, and F.~Wechsung.
\newblock Single-stage stellarator optimization: combining coils with fixed
  boundary equilibria.
\newblock {\em Plasma Phys. Control. Fusion}, 65:074003, 2023.

\bibitem{L17}
M.~Landreman.
\newblock An improved current potential method for fast computation of
  stellarator coil shapes.
\newblock {\em Nucl. Fusion}, 57:046003, 2022.

\bibitem{Li88}
E.L. Lima.
\newblock The {J}ordan-{B}rouwer separation theorem for smooth hypersurfaces.
\newblock {\em The American Mathematical Monthly}, 95(1):39--42, 1988.

\bibitem{M87}
P.~Merkel.
\newblock Solution of stellarator boundary value problems with external
  currents.
\newblock {\em Nucl. Fusion}, 27:867--871, 1987.

\bibitem{NedPlan73}
J.-C. Nedelec and J.~Planchard.
\newblock Une m\'ethode variationnelle d{\textquoteright}\'el\'ements finis
  pour la r\'esolution num\'erique d{\textquoteright}un probl\`eme ext\'erieur
  dans $\mathbf {R}^3$.
\newblock {\em Revue fran\c{c}aise d'automatique informatique recherche
  op\'erationnelle. Analyse num\'erique}, 7(R3):105--129, 1973.

\bibitem{PLBD18}
E.J. Paul, M.~Landreman, A.~Bader, and W.~Dorland.
\newblock An adjoint method for gradient-based optimization of stellarator coil
  shapes.
\newblock {\em Nuclear Fusion}, 58:076015, 2018.

\bibitem{PPPV24}
D.~Pereira~Botelho, V.~Prost, L.B. Pina~Pereira, and F.A. Volpe.
\newblock {Simplified magnet design and manufacture based on patterning of wide
  conductors}.
\newblock {\em arXiv e-prints}, page
  \href{https://arxiv.org/abs/2409.20143}{arXiv:2409.20143}, September 2024.

\bibitem{PRS22}
Y.~Privat, R.~Robin, and Sigalotti M.
\newblock Optimal shape of stellarators for magnetic confinement fusion.
\newblock {\em J. Math. Pures Appl.}, 163:231--264, 2022.

\bibitem{S95}
G.~Schwarz.
\newblock {\em Hodge Decomposition - A Method for Solving Boundary Value
  Problems.}
\newblock Springer Verlag, 1995.

\bibitem{S70}
E.M. Stein.
\newblock {\em Singular {I}ntegrals and {D}ifferentiability {P}roperties of
  {F}unctions}.
\newblock Princeton {U}niversity {P}ress, 1970.

\bibitem{Xu16}
Y.~Xu.
\newblock A general comparison between tokamak and stellarator plasmas.
\newblock {\em Matter and Radiation at Extremes}, 1(4):192--200, 2016.

\end{thebibliography}
\footnotesize
\end{document}